\documentclass[12pt, draftclsnofoot, peerreviewca, onecolumn]{IEEEtran}

\usepackage{cite}
\usepackage{graphicx}
\usepackage[cmex10]{amsmath}
\usepackage{amsfonts}
\usepackage{amssymb}
\ifCLASSOPTIONcompsoc
\usepackage[caption=false,font=normalsize,labelfont=sf,textfont=sf]{subfig}
\else
\usepackage[caption=false,font=footnotesize]{subfig}
\fi
\usepackage{color}
\usepackage{multirow}
\usepackage{cases}
\usepackage{algorithm}
\usepackage{algpseudocode}

\begin{document}
\newtheorem{proposition}{Proposition}
\newtheorem{theorem}{Theorem}
\newtheorem{corollary}{Corollary}
\newtheorem{lemma}{Lemma}
\newtheorem{problem}{Problem}
\newtheorem{remark}{Remark}
\newtheorem{assumption}{Assumption}
\newtheorem{claim}{Claim}
\newtheorem{definition}{Definition}

\title{Cellular Multi-User Two-Way MIMO AF Relaying via Signal Space
Alignment: Minimum Weighted SINR Maximization}

\author{Eddy~Chiu and Vincent~K.~N.~Lau
\thanks{E.~Chiu and V.~K.~N.~Lau are with the Department
of Electronic and Computer Engineering, Hong Kong University of
Science and Technology, Hong Kong (e-mail: echiua@ieee.org and
eeknlau@ust.hk). This work is funded by RGC~614910. The results in
this paper were presented in part at IEEE ICC'12, Jun. 2012.}}

\markboth{To appear in IEEE Transactions on Signal Processing,~2012}
{{Chiu and Lau}: Cellular Multi-User Two-Way MIMO AF Relaying via
Signal Space Alignment}

\maketitle

\begin{abstract}
In this paper, we consider linear MIMO transceiver design for a
cellular two-way amplify-and-forward relaying system consisting of a
single multi-antenna base station, a single multi-antenna relay
station, and multiple multi-antenna mobile stations (MSs). Due to
the two-way transmission, the MSs could suffer from tremendous
multi-user interference. We apply an interference management model
exploiting signal space alignment and propose a transceiver design
algorithm, which allows for alleviating the loss in spectral
efficiency due to half-duplex operation and providing flexible
performance optimization accounting for each user's quality of
service priorities. Numerical comparisons to conventional two-way
relaying schemes based on bidirectional channel inversion and
spatial division multiple access-only processing show that the
proposed scheme achieves superior error rate and average data rate
performance.
\end{abstract}



\section{Introduction} \label{Sec:Introduction}
The use of relays to improve link reliability and coverage of
cellular wireless communication systems has attracted significant
research interest since the pioneer works
\cite{Jnl:Cooperative_diversity:Wornell,
Jnl:Cooperative_diversity:Erkip1, Jnl:Cooperative_diversity:Erkip2,
Jnl:Relay_capacity_scaling_laws:Bolcskei}, and various relaying
protocols are embraced by state-of-the-art and next generation
commercial standards \cite{Std:16m, Jnl:3GPP_LTE-A_relays,
Std:Winner_final_innovation_report}. In practice, most relaying
protocols operate in a half-duplex manner and transmission is
divided into two phases using orthogonal channel accesses: in the
first phase the source node broadcasts its message, and in the
second phase the relay station (RS) forwards the source message to
the destination node. The deficiency of half-duplex relaying is that
when there is no direct link between the source and destination
nodes the end-to-end transmission can only achieve half the degrees
of freedom (DoF) of the channel.

Two-way relaying is a promising means to alleviate the loss in
spectral efficiency due to half-duplex operation. Specifically,
given \emph{a pair} of terminal nodes that are to \emph{exchange}
data, we allow transmissions in both directions to occur
concurrently and reduce the total transmission time by half. The
bidirectional transmissions will mutually interfere with each other;
nonetheless, we can mitigate the impact of interference by employing
spatial division multiple access (SDMA) processing at the RS
\cite{Jnl:Multiuser_two_way_relaying_no_SI_cancellation:Sayed}, or
by applying the principles of analogue network coding (ANC)
\cite{Jnl:Two_way_relay_intro:Wittneben} or physical layer network
coding (PNC) \cite{Cnf:Pnc:Zhang}. By means of ANC, the RS performs
amplify-and-forward (AF) relaying, and the terminal nodes utilize
the a priori knowledge of their own transmitted signals to cancel
the self-induced backward propagated interference. On the other
hand, by means of PNC the RS attempts to decode-and-forward (DF) the
network coded version of the terminal node signals, and the terminal
nodes utilize the knowledge of their own signals to decode the
network code. Note that PNC has strict feasibility requirements for
the precoding and modulation and coding (MCS) schemes used at the
terminal nodes (cf. \cite[Proposition~1]{Cnf:Pnc:Zhang}).

It is, however, nontrivial to extend the two-way relaying protocol
to cellular multi-user systems. In this case, each node experiences
\emph{self-induced interference} as well as \emph{multi-user
interference}, and thus necessitates more sophisticated interference
management techniques. To shed insight on designing an efficient
scheme for cellular multi-user two-way relaying, we first review the
qualities and limitations of the prominent related works.

\textbf{Single-User Two-Way AF Relaying:} In
\cite{Jnl:Two_way_relay_intro:Wittneben,
Jnl:Two_way_single_user_AF_single_stream,
Thesis:Two_way_relay:Unger,
Jnl:Single_user_two_way_relaying_gradient:Korea,
Cnf:ANOMAX_rank_restored,
Jnl:Single_user_two_way_relaying_adaptive_mod:Alouini}, the authors
consider two-way AF relaying between two terminal nodes, and propose
linear transceiver designs subject to various performance metrics
(e.g., sum rate maximization and error rate minimization). In
\cite{Jnl:AF_two-way_relay_error_exponents}, the authors analyze the
random coding error exponent in two-way AF relay networks and
investigate rate and power allocation. However, these single-user
designs cannot be easily extended to multi-user systems as they do
not accommodate for the presence of multi-user interference.

\textbf{Multi-User Two-Way AF Relaying with Fixed Terminal Node
Transceivers:} In
\cite{Jnl:Multiuser_two_way_relaying_no_SI_cancellation:Sayed,
Jnl:Multiuser_two_way_relaying_no_SI_cancellation_User_Sel:Sayed},
the authors consider two-way AF relaying between multiple pairs of
terminal nodes. These works neither exploit self-interference
cancelation nor optimize the terminal node transceivers (i.e., the
transceivers are predetermined offline). Instead, they rely solely
on conventional SDMA processing at the RS to mitigate the effects of
interference. In \cite{Cnf:Two_way_one_BS_multi_MS_AF:Toh,
Jnl:Multiuser_two_way_relaying_scheme_and_analysis:Ding,
Cnf:Two_way_one_BS_multi_MS_AF:Sun}, the authors consider cellular
two-way relaying between a multi-antenna base station (BS) and
multiple single-antenna mobile stations (MSs). Since the MSs are
only equipped with a single antenna, they cannot apply MIMO
interference mitigation techniques. The focus of these works is to
jointly design the BS and RS transceivers to heuristically perform
\emph{\hbox{UL\,-DL} bidirectional channel inversion}, thereby
simultaneously zero-force the multi-user interference from the
received signals of the BS and the MSs.

\textbf{Two-Way DF Relaying:} In \cite{Cnf:Pnc:Zhang,
Jnl:Pnc_channel_coding:Zhang, Jnl:Pnc_modulation:Tarokh}, the
authors consider single-user two-way DF relaying with PNC. Albeit
theoretically promising, the application of PNC is subject to
stringent feasibility requirements that greatly restrict the choices
of the MCS schemes that could be employed, and the decoding
operation at the RS has high computational complexity. In
\cite{Jnl:Mimo_Y_channel:Lee}, the authors consider the unique
scenario of three-user three-way DF relaying using the arguments of
interference alignment (IA) \cite{Jnl:IA:Cadambe_Jafar,
Jnl:Distributed_IA:Gomadam_Cadambe_Jafar, Jnl:Piaid:Huang}. Yet,
this scheme cannot be easily extended to a general number of users
as IA may be infeasible.

In this paper, we consider a cellular system consisting of a BS, an
RS, and multiple MSs. All nodes are equipped with multiple antennas.
We seek to design linear MIMO transceiver for each node to
facilitate efficient two-way AF relaying. The contributions and
technical challenges of this work are as follows.

\begin{list}{\labelitemi}{\leftmargin=0.5em}
\item \textbf{Two-Way Relaying by Virtue of Signal Space Alignment:}
We show that for the cellular multi-user two-way relaying system
under study, the MSs could suffer from tremendous multi-user
interference. Yet, exploiting the advantage of self-interference
cancelation, we can align the signal spaces of the uplink (UL) and
downlink (DL) signals to reduce the dimensions occupied by
multi-user interference at each MS. Ultimately, this allows us to
alleviate the half-duplex loss and achieve the DoF of the channel.

The paradigm of two-way relaying exploiting signal space alignment
is also considered in different contexts in
\cite{Cnf:Two-way_relay_eigen-direction_alignment,
Jnl:Signal_alignment_CDM}. Specifically, in
\cite{Cnf:Two-way_relay_eigen-direction_alignment} the authors
consider a single-user two-way MIMO relaying system, for which they
propose a precoding design that align the two-way signals and deduce
an algorithm for optimizing the basis of the aligned signal space.
However, this single-user design cannot be easily extended to
multi-user systems as it does not accommodate for the presence of
multi-user interference. On the other hand, in \cite{
Jnl:Signal_alignment_CDM} the authors consider a multi-user two-way
multi-carrier relaying system, for which they propose different
frequency domain precoding designs based on aligning the two-way
signals of each pair of communicating terminal nodes. These
frequency domain precoding designs neglect the impact of multi-user
interference, and rely on the intrinsic high frequency diversity to
mitigate the impact of interference. Note that it is non-trivial to
extend these precoding designs to practical two-way MIMO relaying
systems whose signal space dimensions are not large.

\item \textbf{Algorithm for Two-Way Relay Transceiver
Design with Quality of Service Constraints:} In consideration of the
fact that users in cellular systems have different quality of
service (QoS) priorities, we formulate the two-way relay transceiver
design problem to maximize the minimum weighted per stream
signal-to-interference plus noise ratio (SINR) among all UL and DL
data streams. This problem does not lead to closed-form solutions
and is non-convex, and we propose to solve it using a two-stage
algorithm. In the first stage we focus on attaining signal space
alignment, and in the second stage we aim at optimizing the weighted
per stream SINRs\footnote{Note that for multi-user DL unicast
systems, it is shown in \cite{Jnl:Analysis_max-min_weighted_SINR_DL}
that max-min weighted SINR MIMO transceiver designs can achieve
optimal solutions under the special case of rank one channels.}. We
show that the second stage subproblem belongs to the class of
multigroup multicast problems, which are NP-hard
\cite[Claim~2]{Jnl:MBS-SDMA_using_SDR_with_perfect_CSI:Sidiropoulos_Luo}.
So we further propose an algorithm to efficiently solve the second
stage subproblem using second-order cone programming (SOCP)
techniques \cite[Section~4.4.2]{Bok:Convex_optimization:Boyd}.
\end{list}

\textbf{Outline:} The rest of this paper is organized as follows. In
Section~\ref{Sec:SystemModel} we present the system model. In
Section~\ref{Sec:ProblemFormulation} we discuss the interference
management model and formulate the two-way relay transceiver design
problem. In Section~\ref{Sec:ProposedSolution} we present the
proposed transceiver design algorithm. In
Section~\ref{Sec:SimsAndDiscussions} we present numerical simulation
results. Finally, in Section~\ref{Sec:Conclusions} we conclude the
paper.

\textbf{Notations:} $\mathbb{C}_{}^{\,M \times N}$ denotes the set
of complex $M \times N$ matrices. Upper and lower case bold letters
denote matrices and vectors, respectively. $\textrm{vec} (
\textbf{X} )$ denotes the column-by-column vectorization of
$\textbf{X}$. $[\, \textbf{X}_{\,1}^{} \,;\, \ldots \,;\,
\textbf{X}_{N}^{} \,]$ and $[\, \textbf{X}_{\,1}^{}, \ldots,
\textbf{X}_{N}^{} \,]$ denote the matrices obtained by vertically
and horizontally concatenating $\textbf{X}_{\,1}^{}, \ldots,
\textbf{X}_{N}^{}$, respectively. $\textrm{diag}(\,
\textbf{X}_{\,1}^{}, \ldots, \textbf{X}_{N}^{} )$ denotes a block
diagonal matrix having $\textbf{X}_{\,1}^{}, \ldots,
\textbf{X}_{N}^{}$ in the main diagonal. $[\, \textbf{X}
\,]_{(a\,:\,b\,,\,c\,:\,d)}^{}$ denotes the \hbox{$a$-th} to the
\hbox{$b$-th} row and the \hbox{$c$-th} to the \hbox{$d$-th} column
of $\textbf{X}$. $( \cdot )_{}^{T}$, $( \cdot )_{}^{\dag}$, and $(
\cdot )_{}^{\ast}$ denote transpose, Hermitian transpose, and
conjugate, respectively. $\textrm{range} ( \textbf{X} )$ denotes the
column space of $\textbf{X}$. $\textrm{null} ( \textbf{X} )$ denotes
the orthonormal basis for the null space of $\textbf{X}$.
$\textrm{rank} ( \textbf{X} )$ and $\textrm{nullity} ( \textbf{X} )$
denote the rank and the nullity of $\textbf{X}$, respectively.
$\textrm{pinv} ( \textbf{X} )$ denotes the pseudo-inverse of
$\textbf{X}$. $\Re ( y )$ denotes the real component of $y$. $||\,
\textbf{X} \,||$ denotes the Frobenius norm of $\textbf{X}$.
$\succeq_{K}^{}$ denotes the generalized inequality with respect to
the second-order cone, i.e., $[\, y \,;\, \textbf{x} \,]
\succeq_{K}^{} 0$ means that $y \ge ||\, \textbf{x} \,||$.
$\textbf{X}_{\,1}^{} \otimes \textbf{X}_{\,2}^{}$ denotes the
Kronecker product of $\textbf{X}_{\,1}^{}$ and
$\textbf{X}_{\,2}^{}$. $\textbf{x} \sim \mathcal{CN} (
\boldsymbol{\mu}, \boldsymbol{\Xi} )$ denotes that $\textbf{x}$ is
complex Gaussian distributed with mean $\boldsymbol{\mu}$ and
covariance matrix $\boldsymbol{\Xi}$. $\mathbb{E}( \cdot )$ denotes
expectation. $\mathcal{K}$ denotes the index set $\{ 1, \ldots, K
\,\}$ and $\mathcal{L}_{\,k}^{}$ denotes the index set $\{ 1,
\ldots, L_{\,k}^{} \}$. $\textbf{0}_{M \times N}^{}$ denotes an $M
\times N$ matrix of zeros and $\textbf{I}_{N}^{}$ denotes an $N
\times N$ identity matrix.

\section{System Model} \label{Sec:SystemModel} We consider a
multi-user system where a BS communicates with multiple MSs as
illustrated in \figurename~\ref{Fig:MultiUserRelaySystem}. Due to
the effects of path loss and shadowing, there is no direct link
between the BS and the MSs, and a half-duplex RS is deployed to
assist data transmission. In conventional relay systems, UL and DL
transmissions utilize non-overlapping channel accesses (cf.
\figurename~\ref{Fig:OneWayUl} and \figurename~\ref{Fig:OneWayDl}).
We adopt the two-phase two-way relaying protocol whereby UL and DL
transmissions share the channel: first in the multi-access (MAC)
phase the BS and the MSs concurrently transmit to the RS (cf.
\figurename~\ref{Fig:TwoWayMac}), then in the broadcast (BC) phase
the RS forwards the aggregate signals to the BS and the MSs (cf.
\figurename~\ref{Fig:TwoWayBc}). Specifically, we are interested in
a time division duplex (TDD) system where conventional one-way
relaying requires four time slots to complete the UL and DL
transmissions while two-way relaying requires only two time slots as
depicted in \figurename~\ref{Fig:TimeLine}.

The detailed model of the system under study is shown in
\figurename~\ref{Fig:SystemModel}. We consider two-way relaying
between one BS and $K$ MSs. For ease of exposition, we focus on the
\hbox{$k$-th} MS and the same model applies to all the MSs. The BS
is equipped with $N_{B}^{}$ antennas, the RS is equipped with
$N_{R}^{}$ antennas, and the \hbox{$k$-th} MS is equipped with
$N_{\,k}^{}$ antennas. In the DL, the BS transmits $L_{\,k}^{}$ data
streams $\textbf{s}_{D}^{(k)} \in \mathbb{C}_{}^{\,L_{k}^{} \times
1}$ to the \hbox{$k$-th} MS and a total of $L = \sum_{k = 1}^{K} \!
L_{\,k}^{}$ data streams $\textbf{s}_{D}^{} \triangleq [\,
\textbf{s}_{D}^{(1)} \,;\, \ldots \,;\, \textbf{s}_{D}^{(K)} \,]$ to
all the MSs. In the UL, the \hbox{$k$-th} MS transmits $L_{\,k}^{}$
data streams $\textbf{s}_{\,U}^{(k)} \in \mathbb{C}_{}^{\,L_{k}^{}
\times 1}$ to the BS, and altogether the MSs transmit $L$ data
streams $\textbf{s}_{\,U}^{} \triangleq [\, \textbf{s}_{\,U}^{(1)}
\,;\, \ldots \,;\, \textbf{s}_{\,U}^{(K)} \,]$ to the BS. We make
the following assumptions about the data model.

\begin{assumption}[Data Model]\label{Assumption:DataModel} All data
streams are independent and have unit power. The covariance matrix
of the DL data streams is given by $\mathbb{E}(\, \textbf{s}_{D}^{}
( \textbf{s}_{D}^{} )_{}^{\dag\,} ) = \textbf{I}_{L}^{}$, and the
covariance matrix of the UL data streams is given by $\mathbb{E}(\,
\textbf{s}_{\,U}^{(k)} ( \textbf{s}_{\,U}^{(k)} )_{}^{\dag\,} ) =
\textbf{I}_{L_{k}^{}}^{}$.~ \hfill\IEEEQEDclosed
\end{assumption}

\subsection{Two-Way Relaying MAC Phase} In the MAC phase, the BS
precodes the DL data streams $\textbf{s}_{D}^{}$ using the precoder
matrix \hbox{$\textbf{W}_{\!B}^{} \triangleq [\,
\textbf{W}_{\!B,1}^{}, \ldots, \textbf{W}_{\!B,K}^{} \,]$}, where
$\textbf{W}_{\!B,\,k}^{} \in \mathbb{C}_{}^{\,N_{B}^{} \times
L_{k}^{}}$ is the precoder matrix for data streams
$\textbf{s}_{D}^{(k)}$. Thus, the transmitted signals of the BS are
given by $\textbf{x}_{B}^{} =
\textbf{W}_{\!B}^{}\,\textbf{s}_{D}^{}$. Similarly, the
\hbox{$k$-th} MS precodes the data streams $\textbf{s}_{\,U}^{(k)}$
using the precoder matrix $\textbf{W}_{k}^{} \in
\mathbb{C}_{}^{\,N_{k}^{} \times L_{k}^{}}$, and the transmitted
signals of the \hbox{$k$-th} MS are given by $\textbf{x}_{\,k}^{} =
\textbf{W}_{k}^{}\,\textbf{s}_{\,U}^{(k)}$. We make the following
assumptions about the transmit power constraints of the BS and the
MSs.

\begin{assumption}[BS and MS Transmit Power
Constraints]\label{Assumption:BsMsTxPowerConstraints} The maximum
transmit power of the BS is given by $\mathbb{E}( ||\,
\textbf{x}_{B}^{} \,||_{}^{\,2\,} ) = ||\, \textbf{W}_{\!B}^{}
\,||_{}^{\,2\,} \leq P_{B}^{}$. The maximum transmit power of the
\hbox{$k$-th} MS is given by $\mathbb{E}( ||\, \textbf{x}_{\,k}^{}
\,||^{\,2\,} ) = ||\, \textbf{W}_{k}^{} \,||^{\,2\,} \leq
P_{\,k}^{}$.~ \hfill\IEEEQEDclosed
\end{assumption}

Let $\textbf{H}_{R\,,B}^{} \!\in\! \mathbb{C}_{}^{\,N_{R}^{} \times
N_{B}^{}}$ denote the channel matrix from the BS to the RS, and let
$\textbf{H}_{R\,,\,k}^{} \!\in\! \mathbb{C}_{}^{\,N_{R}^{} \times
N_{k}^{}}$ denote the channel matrix from the \hbox{$k$-th} MS to
the RS. It follows that the received signals of the RS can be
expressed as
\begin{IEEEeqnarray*}{l}
\textstyle\textbf{y}_{\!R}^{} = \sum_{m = 1}^{K} \!
\textbf{H}_{R\,,m}^{}\, \textbf{x}_{\,m}^{} +
\textbf{H}_{R\,,B}^{}\, \textbf{x}_{B}^{} + \textbf{n}_{R}^{} =
\textbf{U} \, \textbf{s}_{\,U}^{} + \textbf{D} \, \textbf{s}_{D}^{}
+ \textbf{n}_{R}^{},\IEEEyesnumber\label{Eqn:RsRecvSignal}
\end{IEEEeqnarray*}
where $\textbf{U} \triangleq [\, \textbf{H}_{R\,,1}^{}
\textbf{W}_{1}^{}, \ldots, \textbf{H}_{R\,,K}^{} \textbf{W}_{\!K}^{}
\,]$ and $\textbf{D} \triangleq \textbf{H}_{R\,,B}^{}
\textbf{W}_{\!B}^{}$ represent the MAC phase effective channel
matrices of the UL and DL data streams, respectively, and
$\textbf{n}_{R}^{} \sim \mathcal{CN}\big( \textbf{0}_{N_{R}^{}
\times 1}^{}, N_{\,0}^{} \, \textbf{I}_{N_{R}^{}}^{} \big)$ is the
AWGN. We make the following assumptions about the channel model.

\begin{assumption}[Channel Model]\label{Assumption:ChannelModel} All
channels are independent and exhibit quasi-static fading such that
the channel matrices $\{ \textbf{H}_{R\,,1}^{}, \ldots,
\textbf{H}_{R\,,K}^{}, \textbf{H}_{R\,,B}^{} \}$ remain unchanged
during a fading block of two time slots spanning the MAC and BC
phases. Moreover, the forward and reverse channels are
\emph{reciprocal}: the channel matrix from the RS to the
\hbox{$k$-th} MS is given by $( \textbf{H}_{R\,,\,k}^{} )_{}^{T}$,
and the channel matrix from the RS to the BS is given by $(
\textbf{H}_{R\,,B}^{} )_{}^{T}$. Without loss of generality, we
assume $\textrm{rank} ( \textbf{H}_{R\,,\,k}^{} ) \ge L_{\,k}^{}$
and $\textrm{rank} ( \textbf{H}_{R\,,B}^{} ) \ge L$.~
\hfill\IEEEQEDclosed
\end{assumption}

\subsection{Two-Way Relaying BC Phase} In the BC phase, the RS
amplifies the received signals $\textbf{y}_{\!R}^{}$ using the
transformation matrix $\textbf{W}_{\!R}^{} \in
\mathbb{C}_{}^{\,N_{R}^{} \times N_{R}^{}}$, and the forwarded
signals are given by
\begin{IEEEeqnarray*}{l}
\textbf{x}_{R}^{} = \textbf{W}_{\!R}^{}\,\textbf{y}_{\!R}^{} =
\textbf{W}_{\!R}^{} \, \textbf{U} \, \textbf{s}_{\,U}^{} +
\textbf{W}_{\!R}^{} \, \textbf{D} \, \textbf{s}_{D}^{} +
\textbf{W}_{\!R}^{} \,
\textbf{n}_{R}^{}.\IEEEyesnumber\label{Eqn:RsTxSignal}
\end{IEEEeqnarray*}
We make the following assumption about the transmit power constraint
of the RS.

\begin{assumption}[RS Transmit Power
Constraint]\label{Assumption:RsTxPowerConstraints} The maximum
transmit power of the RS is given by $\mathbb{E}( ||\,
\textbf{x}_{R}^{} \,||^{\,2\,} ) = \sum_{m = 1}^{K} \! ||\,
\textbf{W}_{\!R}^{} \, \textbf{H}_{R\,,m}^{} \textbf{W}_{m}^{}
\,||^{\,2\,} + ||\, \textbf{W}_{\!R}^{} \, \textbf{H}_{R\,,B}^{}
\textbf{W}_{\!B}^{} \,||^{2\,} + N_{\,0}^{} ||\, \textbf{W}_{\!R}^{}
\,||^{\,2\,} \leq P_{R}^{}$. ~ \hfill\IEEEQEDclosed
\end{assumption}

Accordingly, the received signals of the BS are given by
\begin{IEEEeqnarray*}{l}
\textbf{y}_{\!B}^{} = ( \textbf{H}_{R\,,B}^{} )_{}^{T}
\textbf{x}_{R}^{} + \textbf{n}_{B}^{} = \underbrace{(
\textbf{H}_{R\,,B}^{} )_{}^{T\,} \textbf{W}_{\!R}^{} \, \textbf{U}
\, \textbf{s}_{\,U}^{}}_{\textrm{desired signals}} + \underbrace{(
\textbf{H}_{R\,,B}^{} )_{}^{T\,} \textbf{W}_{\!R}^{} \;
\textbf{n}_{R}^{} + \textbf{n}_{B}^{}}_{\textrm{aggregate noise}}
+\; \textbf{i}_{B}^{},\IEEEyesnumber\label{Eqn:BsRecvSignal}
\end{IEEEeqnarray*}
where $\textbf{n}_{B}^{} \sim \mathcal{CN}\big( \textbf{0}_{N_{B}^{}
\times 1}^{}, N_{\,0}^{} \, \textbf{I}_{N_{B}^{}}^{} \big)$ is the
AWGN and $\textbf{i}_{B}^{} \triangleq ( \textbf{H}_{R\,,B}^{}
)_{}^{T\,} \textbf{W}_{\!R}^{} \, \textbf{D} \, \textbf{s}_{D}^{}$
is the backward propagated self-interference. Likewise, the received
signals of the \hbox{$k$-th} MS are given by
\begin{IEEEeqnarray*}{ll}
\textbf{y}_{k}^{} &\;= ( \textbf{H}_{R\,,\,k}^{} )_{}^{T}
\textbf{x}_{R}^{} +
\textbf{n}_{k}^{}\IEEEyesnumber\label{Eqn:MsRecvSignal}\\
&\;=\underbrace{( \textbf{H}_{R\,,\,k}^{} )_{}^{T\,}
\textbf{W}_{\!R}^{} \, \textbf{H}_{R\,,B}^{}
\textbf{W}_{\!B,\,k}^{}\, \textbf{s}_{D}^{(k)}}_{\textrm{desired
signals}} + \underbrace{( \textbf{H}_{R\,,\,k}^{} )_{}^{T\,}
\textbf{W}_{\!R\,}^{} ( \widetilde{\textbf{D}}_{k\,}^{}
\textbf{s}_{D}^{} + \widetilde{\textbf{U}}_{k\,}^{}
\textbf{s}_{\,U}^{} )}_{\textrm{multi-user interference}} +
\underbrace{( \textbf{H}_{R\,,\,k}^{} )_{}^{T\,} \textbf{W}_{\!R}^{}
\; \textbf{n}_{R}^{} + \textbf{n}_{\,k}^{}}_{\textrm{aggregate
noise}} +\; \textbf{i}_{\,k}^{},
\end{IEEEeqnarray*}
where $\textbf{n}_{\,k}^{} \sim \mathcal{CN}\big(
\textbf{0}_{N_{k}^{} \times 1}^{}, N_{\,0}^{} \,
\textbf{I}_{N_{k}^{}}^{}\big)$ is the AWGN, $\textbf{i}_{\,k}^{}
\triangleq ( \textbf{H}_{R\,,\,k}^{} )_{}^{T\,} \textbf{W}_{\!R}^{}
\; \textbf{H}_{R\,,\,k}^{} \textbf{W}_{k}^{}
\,\textbf{s}_{\,U}^{(k)}$ is the self-interference, and
$\widetilde{\textbf{U}}_{k}^{} \triangleq [\, \textbf{H}_{R\,,1}^{}
\textbf{W}_{1}^{}, \ldots, \textbf{H}_{R\,,\,k-1}^{}
\textbf{W}_{k-1}^{}, \textbf{0}_{N_{R}^{} \times L_{k}^{}}^{},
\textbf{H}_{R\,,\,k+1}^{} \textbf{W}_{k+1}^{}, \ldots,
\textbf{H}_{R\,,K}^{} \textbf{W}_{\!K}^{} \,]$ and
$\widetilde{\textbf{D}}_{k}^{} \triangleq \textbf{H}_{\,R,B}^{} \:
[\, \textbf{W}_{\!B,1}^{}, \ldots, \textbf{W}_{\!B,\,k-1}^{},
\textbf{0}_{N_{B}^{} \times L_{k}^{}}^{}, \textbf{W}_{\!B,\,k+1}^{},
\ldots, \textbf{W}_{\!B,K}^{} \,]$ are the MAC phase effective
channel matrices of the UL and DL interference streams,
respectively.

\subsection{Receive Processing} The BS and MS receive processing
consists of two steps. Inherent to the two-way relaying protocol,
the BS and the MSs can exploit the a priori knowledge of their own
transmitted signals to cancel the backward propagated
self-interference\footnote{We shall elaborate in
Section~\ref{Sec:Implementation} the assumptions on the side
information available at each node to facilitate transceiver design
and self-interference cancelation.}. After that, the BS and the MSs
process the resultant signals using linear equalizers to produce
data stream estimates. Specifically, the BS cancels the
self-interference $\textbf{i}_{B}^{}$ from the received signals
$\textbf{y}_{\!B}^{}$ and processes them using the equalizer matrix
$\textbf{V}_{\!B}^{} \triangleq [\, \textbf{V}_{\!B}^{\,(1)} \,;\,
\ldots \,;\, \textbf{V}_{\!B}^{\,(K)} \,]$, where
$\textbf{V}_{\!B}^{\,(k)} \in \mathbb{C}_{}^{\,L_{k}^{} \times
N_{B}^{}}$ is the equalizer matrix for data streams
$\textbf{s}_{\,U}^{(k)}$. The UL data stream estimates are given by
\begin{IEEEeqnarray*}{l}
\label{Eqn:BsSignalEstimates} \widetilde{\textbf{s}}_{\,U}^{} =
\textbf{V}_{\!B}^{} \, (\, \textbf{y}_{\!B}^{} \!- \textbf{i}_{B}^{}
\,) = [\, \widetilde{\textbf{s}}_{\,U}^{\,(1)} \,;\, \ldots \,;\,
\widetilde{\textbf{s}}_{\,U}^{\,(K)} \,],\IEEEyessubnumber\\
\widetilde{\textbf{s}}_{\,U}^{\,(k)} = \textbf{V}_{\!B}^{\,(k)} (
\textbf{H}_{R\,,B}^{} )_{}^{T\,} \textbf{W}_{\!R}^{} \, \textbf{U}
\, \textbf{s}_{\,U}^{} + \textbf{V}_{\!B}^{\,(k)} (
\textbf{H}_{R\,,B}^{} )_{}^{T\,} \textbf{W}_{\!R}^{} \,
\textbf{n}_{R}^{} \!+\! \textbf{V}_{\!B}^{\,(k)} \,
\textbf{n}_{B}^{}.\;\;\;\IEEEyessubnumber
\end{IEEEeqnarray*}
In the same way, the \hbox{$k$-th} MS cancels the self-interference
$\textbf{i}_{\,k}^{}$ from the received signals $\textbf{y}_{k}^{}$
and processes them using the equalizer matrix $\textbf{V}_{k}^{} \in
\mathbb{C}_{}^{\,L_{k}^{} \times N_{k}^{}}$ to produce the DL data
stream estimates
\begin{IEEEeqnarray*}{l}
\widetilde{\textbf{s}}_{D}^{\,(k)} = \textbf{V}_{k}^{} \, (\,
\textbf{y}_{k}^{} \!- \textbf{i}_{\,k}^{}
\,)\IEEEyesnumber\label{Eqn:MsSignalEstimates}\\
=\! \textbf{V}_{k}^{} ( \textbf{H}_{R\,,\,k}^{} )_{}^{T\,}
\textbf{W}_{\!R}^{} \, \textbf{H}_{R\,,B}^{}
\textbf{W}_{\!B,\,k}^{}\, \textbf{s}_{D}^{(k)} \!+\!
\textbf{V}_{k}^{} ( \textbf{H}_{R\,,\,k}^{} )_{}^{T\,}
\textbf{W}_{\!R\,}^{} ( \widetilde{\textbf{D}}_{k\,}^{}
\textbf{s}_{D}^{} \!+\! \widetilde{\textbf{U}}_{k\,}^{}
\textbf{s}_{\,U}^{} ) \!+\! \textbf{V}_{k}^{} (
\textbf{H}_{R\,,\,k}^{} )_{}^{T\,} \textbf{W}_{\!R}^{} \;
\textbf{n}_{R}^{} \!+\! \textbf{V}_{k}^{} \, \textbf{n}_{\,k}^{}.
\end{IEEEeqnarray*}

In the UL the SINR of the data stream estimate
$[\,\widetilde{\textbf{s}}_{\,U}^{\,(k)}\,]_{(l)}^{}$ is given by
\begin{IEEEeqnarray*}{l} \gamma_{\,U}^{(k,\,l)}
\triangleq \frac{\substack{\textstyle|\, [\,
\textbf{V}_{\!B}^{\,(k)} \,]_{(l,\,:)}^{} ( \textbf{H}_{R\,,B}^{}
)_{}^{T\,} \textbf{W}_{\!R}^{} \, \textbf{H}_{R\,,\,k}^{} \, [\,
\textbf{W}_{k}^{} \,]_{(:,\,l)}^{}
\,|_{}^{\,2\,}}}{\left(\substack{\textstyle||\, [\,
\textbf{V}_{\!B}^{\,(k)} \,]_{(l,\,:)}^{} ( \textbf{H}_{R\,,B}^{}
)_{}^{T\,} \textbf{W}_{\!R}^{} \, \textbf{U} \,||_{}^{\,2\,} -
\textstyle|\, [\, \textbf{V}_{\!B}^{\,(k)} \,]_{(l,\,:)}^{} (
\textbf{H}_{R\,,B}^{} )_{}^{T\,} \textbf{W}_{\!R}^{} \,
\textbf{H}_{R\,,\,k}^{} \, [\,
\textbf{W}_{k}^{} \,]_{(:,\,l)}^{} \,|_{}^{\,2\,}\\
\textstyle+ N_{\,0}^{} ( ||\, [\, \textbf{V}_{\!B}^{\,(k)}
\,]_{(l,\,:)}^{} ( \textbf{H}_{R\,,B}^{} )_{}^{T\,}
\textbf{W}_{\!R}^{} \,||_{}^{\,2\,} + ||\, [\,
\textbf{V}_{\!B}^{\,(k)} \,]_{(l,\,:)}^{} \,||_{}^{\,2\,}
)}\right)},\IEEEyesnumber\label{Eqn:UlSinr}
\end{IEEEeqnarray*}
and in the DL the SINR of the data stream estimate
$[\,\widetilde{\textbf{s}}_{D}^{\,(k)}\,]_{(l)}^{}$ is given by
\begin{IEEEeqnarray*}{l} \gamma_{D}^{(k,\,l)}
\triangleq \frac{\substack{\textstyle|\, [\, \textbf{V}_{k}^{}
\,]_{(l,\,:)}^{} ( \textbf{H}_{R\,,\,k}^{} )_{}^{T\,}
\textbf{W}_{\!R}^{} \, \textbf{H}_{R\,,B}^{} \, [\,
\textbf{W}_{\!B,\,k}^{} \,]_{(:,\,l)}^{}
\,|_{}^{\,2\,}}}{\left(\substack{\textstyle||\, [\,
\textbf{V}_{k}^{} \,]_{(l,\,:)}^{} ( \textbf{H}_{R\,,\,k}^{}
)_{}^{T\,} \textbf{W}_{\!R}^{} \, [\, \textbf{D},
\widetilde{\textbf{U}}_{k}^{} \,] \,||_{}^{\,2\,} \!- |\, [\,
\textbf{V}_{k}^{} \,]_{(l,\,:)}^{} ( \textbf{H}_{R\,,\,k}^{}
)_{}^{T\,} \textbf{W}_{\!R}^{} \, \textbf{H}_{R\,,B}^{} \, [\,
\textbf{W}_{\!B,\,k}^{} \,]_{(:,\,l)}^{} \,|_{}^{\,2\,}\\
\textstyle+ N_{\,0}^{} ( ||\, [\, \textbf{V}_{k}^{} \,]_{(l,\,:)}^{}
( \textbf{H}_{R\,,\,k}^{} )_{}^{T\,} \textbf{W}_{\!R}^{}
\,||_{}^{\,2\,} \!+ ||\, [\, \textbf{V}_{k}^{} \,]_{(l,\,:)}^{}
\,||_{}^{\,2\,})}\right)}.\;\;\;\;\IEEEyesnumber\label{Eqn:DlSinr}
\end{IEEEeqnarray*}
Furthermore, the achievable data rate for each data stream can be
expressed as
\begin{IEEEeqnarray*}{l}
\textrm{UL: } C_{\,U}^{\,(k,\,l)} = \frac{1}{2} \log_2( 1 +
\gamma_{\,U}^{(k,\,l)} ), \; \textrm{DL: } C_{D}^{\,(k,\,l)} =
\frac{1}{2} \log_2( 1 + \gamma_{D}^{(k,\,l)}
),\IEEEyesnumber\label{Eqn:Cap}
\end{IEEEeqnarray*}
where the factor of $1/2$ accounts for the half-duplex loss.

\section{Interference Management and Transceiver Design Problem
Formulation} \label{Sec:ProblemFormulation} In this section, we
first discuss the motivations behind the interference management
model to exploit signal space alignment. We then proceed to
formulate the transceiver design problem.

\subsection{Interference Management via Signal Space Alignment}
\label{Sec:InterferenceManagement} As shown in
\eqref{Eqn:BsSignalEstimates} and \eqref{Eqn:MsSignalEstimates}, the
UL and DL data stream estimates are given by
\begin{IEEEeqnarray*}{l}
\widetilde{\textbf{s}}_{\,U}^{\,(k)} \!=\! \textbf{V}_{\!B}^{\,(k)}
( \textbf{H}_{R\,,B}^{} )_{}^{T\,} \textbf{W}_{\!R}^{} \, \textbf{U}
\, \textbf{s}_{\,U}^{} + \textbf{V}_{\!B}^{\,(k)} (
\textbf{H}_{R\,,B}^{} )_{}^{T\,} \textbf{W}_{\!R}^{} \,
\textbf{n}_{R}^{} + \textbf{V}_{\!B}^{\,(k)} \, \textbf{n}_{B}^{},\\
\widetilde{\textbf{s}}_{D}^{\,(k)} \!=\!
\underbrace{\textbf{V}_{k}^{} ( \textbf{H}_{R\,,\,k}^{} )_{}^{T\,}
\textbf{W}_{\!R}^{} \, \textbf{H}_{R\,,B}^{}
\textbf{W}_{\!B,\,k}^{}\, \textbf{s}_{D}^{(k)}}_{\textrm{desired
signals}} \!+\! \underbrace{\textbf{V}_{k}^{} (
\textbf{H}_{R\,,\,k}^{} )_{}^{T\,} \textbf{W}_{\!R\,}^{} (
\widetilde{\textbf{D}}_{k\,}^{} \textbf{s}_{D}^{} \!+\!
\widetilde{\textbf{U}}_{k\,}^{} \textbf{s}_{\,U}^{}
)}_{\textrm{multi-user interference}} \!+ \textbf{V}_{k}^{} (
\textbf{H}_{R\,,\,k}^{} )_{}^{T\,} \textbf{W}_{\!R}^{} \;
\textbf{n}_{R}^{} \!+\! \textbf{V}_{k}^{} \, \textbf{n}_{\,k}^{},
\end{IEEEeqnarray*}
where the DL data stream estimates
$\widetilde{\textbf{s}}_{D}^{\,(k)}$ are prone to multi-user
interference and it is nontrivial to mitigate its effects. On the
one hand, it is detrimental to performance if we naively treat the
multi-user interference as noise since its strength could be
comparable to the desired signals. On the other hand, it is not
always spectrally efficient if we were to mitigate interference by
solely using conventional SDMA processing at the RS
\cite{Jnl:Multiuser_two_way_relaying_no_SI_cancellation:Sayed}.
Under this approach, \emph{all} the signal streams that constitute
the RS forwarded signals \eqref{Eqn:RsTxSignal} must be linearly
independent, which implies that \hbox{$\textrm{rank} ( [\,
\textbf{W}_{\!R}^{} \, \textbf{U}, \textbf{W}_{\!R}^{} \, \textbf{D}
\,] ) = 2L$} and the channel matrices must satisfy
\begin{IEEEeqnarray*}{l}
L_{\,k}^{} \leq \textrm{rank} ( \textbf{H}_{R\,,\,k}^{} ) \leq
N_{\,k}^{}, \;\forall k \in \mathcal{K}, \;\; L \leq \textrm{rank} (
\textbf{H}_{R\,,B}^{}
) \leq N_{B}^{},\IEEEyessubnumber\label{Eqn:SdmaChannelConditionsRank}\\
2L \leq \textrm{rank} ( \textbf{H}_{R\,,\,k}^{} ) \!+\!
\textrm{nullity} ( \textbf{H}_{R\,,\,k}^{} ), \; \textrm{rank} (
\textbf{H}_{R\,,B}^{} ) \!+\! \textrm{nullity} (
\textbf{H}_{R\,,B}^{} ) =
N_{R}^{},\IEEEyessubnumber\label{Eqn:SdmaChannelConditions}
\end{IEEEeqnarray*}
so the MSs (BS) can only transmit $L \leq N_{R}^{} \,/ 2$ data
streams in the UL (DL).

Taking into consideration that each node is capable of canceling the
backward propagated self-interference in the received signals, we
can allow the self-interference to overlap with the desired signals
since ultimately it does not affect the decoding of the desired
signals. As such, we can facilitate interference management by
\emph{perfectly aligning} the UL and DL signal spaces to reduce the
dimension of the multi-user interference space at each node as
exemplified in \figurename~\ref{Fig:SignalSpace}. Mathematically,
aligning the UL and DL signal spaces can be represented as
\begin{IEEEeqnarray*}{Rl}
\label{Eqn:SpaceAlign}\textrm{BS: }&\textrm{range} ( (
\textbf{H}_{R\,,B}^{} )_{}^{T\,} \textbf{W}_{\!R}^{} \,
\textbf{H}_{R\,,m}^{} \textbf{W}_{\!m\,}^{} ) = \textrm{range} ( (
\textbf{H}_{R\,,B}^{} )_{}^{T\,} \textbf{W}_{\!R}^{} \,
\textbf{H}_{R\,,B}^{} \textbf{W}_{\!B,\,m\,}^{} ), \;\;\forall m \in
\mathcal{K},\;\;\;\;\IEEEyessubnumber\label{Eqn:SpaceAlignBs}\\
\textrm{\hbox{$k$-th} MS: }&\textrm{range} ( (
\textbf{H}_{R\,,\,k}^{} )_{}^{T\,} \textbf{W}_{\!R}^{} \,
\textbf{H}_{R\,,B}^{} \textbf{W}_{\!B,\,m\,}^{} ) = \textrm{range} (
( \textbf{H}_{R\,,\,k}^{} )_{}^{T\,} \textbf{W}_{\!R}^{} \,
\textbf{H}_{R\,,m}^{} \textbf{W}_{\!m\,}^{} ), \;\;\forall m \in
\mathcal{K},\;\;\;\;\IEEEyessubnumber\label{Eqn:SpaceAlignMs}
\end{IEEEeqnarray*}
and this can be manifested by constructing the RS forwarded signals
\eqref{Eqn:RsTxSignal} such that
\begin{IEEEeqnarray*}{l}
\underbrace{\textrm{range} ( \textbf{W}_{\!R}^{} \,
\textbf{H}_{R\,,m}^{} \textbf{W}_{\!m\,}^{} )}_{\textrm{UL signals}}
= \underbrace{\textrm{range} ( \textbf{W}_{\!R}^{} \,
\textbf{H}_{R\,,B}^{} \textbf{W}_{\!B,\,m\,}^{} )}_{\textrm{DL
signals}}.\IEEEyesnumber\label{Eqn:SpaceAlignCondition}
\end{IEEEeqnarray*}
In order for all the UL and DL signal streams to be linearly
independent, the rank of the RS forwarded signals should be
\begin{IEEEeqnarray*}{l}
\label{Eqn:SpaceAlignRankAll} \textrm{UL: }\textrm{rank} ( [\,
\textbf{W}_{\!R}^{} \, \textbf{H}_{R\,,1}^{} \textbf{W}_{1\,}^{},
\ldots, \textbf{W}_{\!R}^{} \, \textbf{H}_{R\,,K}^{}
\textbf{W}_{\!K\,}^{}
\,] ) = L,\IEEEyessubnumber\label{Eqn:SpaceAlignRankSeparateUL}\\
\textrm{DL: }\textrm{rank} ( [\, \textbf{W}_{\!R}^{} \,
\textbf{H}_{R\,,B}^{} \textbf{W}_{\!B,\,1\,}^{}, \ldots,
\textbf{W}_{\!R}^{} \, \textbf{H}_{R\,,B}^{}
\textbf{W}_{\!B,\,K\,}^{} \,] ) =
L,\IEEEyessubnumber\label{Eqn:SpaceAlignRankSeparateDL}\\
\textrm{rank} ( [\, \underbrace{\textbf{W}_{\!R}^{} \,
\textbf{H}_{R\,,1}^{} \textbf{W}_{1\,}^{}, \ldots,
\textbf{W}_{\!R}^{} \, \textbf{H}_{R\,,K}^{}
\textbf{W}_{\!K\,}^{}\!}_{\textrm{UL signals}},
\underbrace{\textbf{W}_{\!R}^{} \, \textbf{H}_{R\,,B}^{}
\textbf{W}_{\!B,\,1\,}^{}, \ldots, \textbf{W}_{\!R}^{} \,
\textbf{H}_{R\,,B}^{} \textbf{W}_{\!B,\,K\,}^{}\!}_{\textrm{DL
signals}} \,] ) = L,\IEEEyessubnumber\label{Eqn:SpaceAlignRank}
\end{IEEEeqnarray*}
and from \eqref{Eqn:SpaceAlignRankAll} it suffices that the channel
matrices satisfy
\begin{IEEEeqnarray*}{l}
\label{Eqn:AlignmentChannelConditionsAll} L_{\,k}^{} \leq
\textrm{rank} ( \textbf{H}_{R\,,\,k}^{} ) \leq N_{\,k}^{}, \;\forall
k \in \mathcal{K}, \;\; L \leq \textrm{rank} ( \textbf{H}_{R\,,B}^{}
) \leq
N_{B}^{},\IEEEyessubnumber\label{Eqn:AlignmentChannelConditionsRank}\\
L \leq \textrm{rank} ( \textbf{H}_{R\,,\,k}^{} ) \!+\!
\textrm{nullity} ( \textbf{H}_{R\,,\,k}^{} ), \; \textrm{rank} (
\textbf{H}_{R\,,B}^{} ) \!+\! \textrm{nullity} (
\textbf{H}_{R\,,B}^{} ) =
N_{R}^{},\IEEEyessubnumber\label{Eqn:AlignmentChannelConditions}
\end{IEEEeqnarray*}
so the MSs (BS) can transmit $L \leq N_{R}^{}$ data streams in the
UL (DL). Comparing \eqref{Eqn:AlignmentChannelConditions} and
\eqref{Eqn:SdmaChannelConditions} shows that we can achieve superior
multiplexing gain by exploiting signal space alignment than by
performing conventional SDMA processing.

Consider again the UL and DL data stream estimates in
\eqref{Eqn:BsSignalEstimates} and \eqref{Eqn:MsSignalEstimates}.
Exploiting signal space alignment, in the DL data stream estimates
$\widetilde{\textbf{s}}_{D}^{\,(k)}$ the UL and DL multi-user
interference streams span the same signal space and \emph{appear} as
if they were one set of streams. In effect, UL and DL transmissions
perform similarly to \emph{separated} one-way relaying
transmissions.

\begin{remark}[Feasibility of Signal Space
Alignment]\label{Remark:AlignmentFeasibility} Aligning the UL and DL
signal spaces as per \eqref{Eqn:SpaceAlignCondition} requires that
the two-way signals between the BS and the \hbox{$k$-th} MS be
aligned when received by the RS (i.e., a \emph{single} node), which
then broadcasts the aligned signals back to the BS and the
\hbox{$k$-th} MS. Using the arguments of coordinated transmission
and reception \cite{Jnl:Interference_channel_DoF:Jafar}, it can be
shown that the alignment operation is feasible if the number of
antennas and data streams for each node and the rank of the channel
matrices satisfy \eqref{Eqn:AlignmentChannelConditionsAll}. Note
that this is unlike conventional IA for interference channels, which
is subject to stringent feasibility conditions
\cite{Jnl:IA:Cadambe_Jafar,
Jnl:Distributed_IA:Gomadam_Cadambe_Jafar, Jnl:Piaid:Huang}, due to
the requirement to \emph{simultaneously} align interference at
\emph{multiple} nodes.~ \hfill\IEEEQEDclosed
\end{remark}

\begin{remark}[DoF of One- and Two-Way AF
Relaying]\label{Remark:RelayChannelDoF} As per
\cite{Jnl:AF_relay_scaling:Gallager,
Jnl:Interference_channel_DoF:Jafar}, the achievable DoF of
multi-user one-way half-duplex AF relaying is $1/2 \min\{
\textrm{rank} ( \textbf{H}_{R\,,B}^{} ), \textrm{rank} ( [\,
\textbf{H}_{R\,,1}^{}, \ldots, \textbf{H}_{R\,,K}^{} \,] ) \}$,
where the factor of $1/2$ accounts for the half-duplex loss.
Exploiting signal space alignment, the achievable DoF of two-way
relaying is $\min\{ \textrm{rank} ( \textbf{H}_{R\,,B}^{} ),
\textrm{rank} ( [\, \textbf{H}_{R\,,1}^{}, \ldots,
\textbf{H}_{R\,,K}^{} \,] ) \}$.~ \hfill\IEEEQEDclosed
\end{remark}

\begin{corollary}\label{Corollary:RelayChannelDoFIid}
For independent and identically distributed (i.i.d.) Rayleigh fading
channels, with probability 1, $\textrm{rank} ( \textbf{H}_{R\,,k}^{}
) = \min\{ N_{\,k}^{}, N_{R}^{} \,\}$, $\textrm{rank} ( [\,
\textbf{H}_{R\,,1}^{}, \ldots, \textbf{H}_{R\,,K}^{} \,] ) = \min\{
N_{R}^{}, \sum_{k = 1}^{K} \! N_{\,k}^{} \}$, and $\textrm{rank} (
\textbf{H}_{R\,,B}^{} ) = \min\{ N_{B}^{}, N_{R}^{} \,\}$. Thus, the
achievable DoF of two-way relaying is given by\\
\indent\indent\indent\indent\indent$\min\{ \min\{ N_{B}^{}, N_{R}^{}
\,\}, \min\{ N_{R}^{}, \sum_{k = 1}^{K} \! N_{\,k}^{} \} \} = \min\{
N_{B}^{}, N_{R}^{}, \sum_{k = 1}^{K} \! N_{\,k}^{} \}$.~
\hfill\IEEEQEDclosed
\end{corollary}

\subsection{Transceiver Design
Optimization}\label{Sec:OptimizationProb} In the preceding
discussion, we have focused on interference management without
regard for QoS considerations. In practice, the users might have
different service priorities, and their data streams might have
heterogenous requirements (for example, in terms of throughput and
reliability). However, as shown in
\eqref{Eqn:UlSinr}--\eqref{Eqn:Cap}, the achievable data rate for
each data stream is intricately related to the channel qualities of
all links and the transceiver matrices of all nodes. One issue is
that typically the transmit power of the BS is substantially higher
than the MSs, and the transceiver design should accommodate for the
unequal transmit powers to ensure that both UL and DL transmissions
are of satisfactory performance. Toward this end, we seek to
optimize the transceiver design to maximize the minimum
\emph{weighted SINR} among all data streams. Specifically, we
associate with each UL and DL data stream a weight factor that
corresponds to its priority, where the higher the priority of a data
stream the larger its weight factor. In the UL let $[\,
\boldsymbol{\omega}_{U}^{(k)} \,]_{(l)}^{} \ge 1$ denote the weight
factor for data stream $[\,\textbf{s}_{\,U}^{(k)}\,]_{(l)}^{}$, and
in the DL let $[\, \boldsymbol{\omega}_{\!D}^{(k)} \,]_{(l)}^{} \ge
1$ denote the weight factor for data stream
$[\,\textbf{s}_{D}^{(k)}\,]_{(l)}^{}$. We define the weighted per
stream SINRs as
\vspace{-0.5em}
\begin{IEEEeqnarray*}{l}
\textrm{UL: }\big( [\, \boldsymbol{\omega}_{U}^{(k)} \,]_{(l)}^{}
\big)_{}^{-1} \gamma_{\,U}^{(k,\,l)}, \; \textrm{DL: }\big( [\,
\boldsymbol{\omega}_{\!D}^{(k)} \,]_{(l)}^{} \big)_{}^{-1}
\gamma_{D}^{(k,\,l)}.\IEEEyesnumber\label{Eqn:WeighingSinrDefn}
\end{IEEEeqnarray*}
By maximizing the minimum weighted per stream SINR,
we can simultaneously enhance the performance of all data streams
while accounting for their relative priorities. Altogether, we
formulate the two-way relaying transceiver design problem as
follows.

\begin{problem}[Two-Way Relaying Transceiver Design with QoS Constraints]
\label{Prob:ProbMain} Given the transmit power constraint of each
node and the priority weight factor of each data stream, we design
the two-way relaying transceiver processing -- exploiting signal
space alignment -- to maximize the minimum weighted per stream
SINR.\\\indent\indent$\big\{ \textbf{W}_{\!B}^{\,\star}, \{\:\!
\textbf{W}_{k}^{\,\star\,} \}_{k=1}^{K},\!
\textbf{W}_{\!R}^{\,\star}, \textbf{V}_{\!B}^{\,\star}, \{\:\!
\textbf{V}_{k}^{\,\star\,} \}_{k=1}^{K} \big\} := \mathcal{Q}\big\{
P_{B}^{}, \{ P_{\,k}^{} \}_{k=1}^{K}, P_{R}^{}, \{
\boldsymbol{\omega}_{U}^{(k)}\!, \boldsymbol{\omega}_{\!D}^{(k)\,}
\}_{k=1}^{K} \big\}$\nopagebreak[4]
\begin{subnumcases}{\label{Eqn:ProbMain}}
\!\underset{\substack{\textbf{W}_{\!B}^{}, \{\:\!
\textbf{W}_{\!k}^{} \}_{k=1}^{K}\\\textbf{W}_{\!R}^{},
\textbf{V}_{\!B}^{}, \{\:\! \textbf{V}_{\!k}^{} \}_{k=1}^{K}}}{\arg
\max} \!\!\!\!\!\!\!\!\!\! & $\displaystyle\min_{\substack{\forall\,
k \,\in\, \mathcal{K}\\\forall\, l \,\in\, \mathcal{L}_{k}^{}}}
\textstyle\big\{ \big( [\, \boldsymbol{\omega}_{U}^{(k)}
\,]_{(l)}^{} \big)_{}^{-1} \gamma_{\,U}^{(k,\,l)}, \big( [\,
\boldsymbol{\omega}_{\!D}^{(k)} \,]_{(l)}^{} \big)_{}^{-1}
\gamma_{D}^{(k,\,l)} \, \big\}$\label{Eqn:ProbMainObjective}\\
\;\;\;\;\;\;\textrm{s.t.} & $||\, \textbf{W}_{\!B}^{}
\,||_{}^{\,2\,} \leq P_{B}^{}, \;\;||\, \textbf{W}_{k}^{}
\,||^{\,2\,} \leq P_{\,k}^{}, \;\forall k \in
\mathcal{K},$\label{Eqn:ProbMainBsMsPowerConstraint}\\
 & $\textstyle\sum_{m = 1}^{K} \! ||\,
\textbf{W}_{\!R}^{} \, \textbf{H}_{R\,,\,m}^{} \textbf{W}_{m}^{}
\,||^{\,2\,} + ||\, \textbf{W}_{\!R}^{} \, \textbf{H}_{R\,,B}^{}
\textbf{W}_{\!B}^{} \,||^{\,2\,} + N_{\,0}^{} ||\,
\textbf{W}_{\!R}^{} \,||^{\,2\,} \leq
P_{R}^{},$\label{Eqn:ProbMainRsPowerConstraint}\\
 & $\textrm{range} ( \textbf{W}_{\!R}^{} \,
\textbf{H}_{R\,,B}^{} \textbf{W}_{\!B,\,k}^{} ) = \textrm{range} (
\textbf{W}_{\!R}^{} \, \textbf{H}_{R\,,\,k}^{} \textbf{W}_{k}^{} ),
\;\forall k \in
\mathcal{K}.\;\;\;\;\;$\label{Eqn:ProbMainSpaceAlignDesired}
\end{subnumcases}~ \hfill\IEEEQEDclosed
\end{problem}

\vspace{-0.75em} Note that Problem~$\mathcal{Q}$ is difficult to
solve since it does not lead to closed-form solutions and is
non-convex. As we show in the next section, it is also nontrivial to
reformulate Problem~$\mathcal{Q}$ in order to take advantage of its
structures and solve it more easily\footnote{The general two-way
transceiver optimization problem is extremely complex. By imposing
the signal alignment structure, we can simplify the freedoms of
optimization to obtain effective and pragmatic transceiver
designs.}.

\vspace{-0.75em}
\section{Proposed Solution} \label{Sec:ProposedSolution}
\vspace{-0.25em}
\subsection{Two-Stage Transceiver Design Paradigm}
\label{Sec:DesignParadigm} The transceiver design problem,
Problem~$\mathcal{Q}$, does not lead to closed-form solutions and is
non-convex (since the objective function and constraints are not
jointly convex in all the optimization variables); thus, we cannot
efficiently solve for all the transceiver matrices in a
\emph{single-shot} manner. We propose to solve the transceiver
design problem using a two-stage paradigm: in the first stage we
focus on attaining alignment between the UL and DL signal streams,
and in the second stage we aim at optimizing the weighted per stream
SINRs.

To facilitate decomposing the transceiver design problem into two
stages, we first extend the signal model as follows. We divide the
RS transformation matrix into two components
\begin{IEEEeqnarray*}{l}
\textbf{W}_{\!R}^{} \triangleq \textbf{F}_{\!R\,}^{}
\textbf{A}_{R}^{};\IEEEyesnumber\label{Eqn:RsPrecoderDecompose}
\end{IEEEeqnarray*}
we define $\textbf{A}_{R}^{} \triangleq [\, \textbf{A}_{R}^{\!(1)}
\,;\, \ldots \,;\, \textbf{A}_{R}^{\!(K)} \,]$ as the \emph{RS
equalizer matrix} and $\textbf{F}_{\!R}^{} \triangleq [\,
\textbf{F}_{\!R}^{\,(1)}, \ldots, \textbf{F}_{\!R}^{\,(K)} \,]$ as
the \emph{RS precoder matrix}, where the submatrices
$\textbf{A}_{R}^{\!(k)} \in \mathbb{C}_{}^{\,L_{k}^{} \times
N_{R}^{}}$ and $\textbf{F}_{\!R}^{\,(k)} \in
\mathbb{C}_{}^{\,N_{R}^{} \times L_{k}^{}}$ are associated with the
signal streams of the \hbox{$k$-th} MS. Substituting
\eqref{Eqn:RsPrecoderDecompose} into \eqref{Eqn:RsTxSignal}, the RS
forwarded signals are given by $\textbf{x}_{R}^{} =
\textbf{F}_{\!R\,}^{} \big( \textbf{A}_{R}^{} \, [\,
\textbf{H}_{R\,,1}^{} \textbf{W}_{1}^{}, \ldots,
\textbf{H}_{R\,,K}^{} \textbf{W}_{\!K}^{} \,] \, \textbf{s}_{\,U}^{}
+ \textbf{A}_{R}^{} \, [\, \textbf{H}_{R\,,B}^{}
\textbf{W}_{\!B,1}^{}, \ldots, \textbf{H}_{R\,,B}^{}
\textbf{W}_{\!B,K}^{} \,] \, \textbf{s}_{D}^{} + \textbf{n}_{R}^{}
\big)$. Analogous to interference alignment (cf.
\cite{Jnl:Piaid:Huang} and references therein), aligning the UL and
DL signal streams \eqref{Eqn:ProbMainSpaceAlignDesired} can be
encompassed by the following conditions:
\begin{IEEEeqnarray*}{l}
\label{Eqn:AlignmentConditionsPerUser}\textrm{UL:
}\textbf{A}_{R}^{\!(k)} \, \textbf{H}_{R\,,\,k}^{} \textbf{W}_{k}^{}
= \textrm{diag}( \boldsymbol{\phi}_{}^{\,(k)} \,), \;\;
\textbf{A}_{R}^{\!(k)} \, \textbf{H}_{R\,,\,m}^{} \textbf{W}_{m}^{}
= \textbf{0}_{L_{k}^{} \times L_{m}^{}}^{}, \;k \neq
m,\IEEEyessubnumber\label{Eqn:AlignmentConditionsPerUserDl}\\
\textrm{DL: }\textbf{A}_{R}^{\!(k)} \, \textbf{H}_{R\,,B}^{}
\textbf{W}_{\!B,\,k}^{} = \textrm{diag}(
\boldsymbol{\psi}_{}^{\,(k)} \,), \;\; \textbf{A}_{R}^{\!(k)} \,
\textbf{H}_{R\,,B}^{} \textbf{W}_{\!B,\,m}^{} = \textbf{0}_{L_{k}^{}
\times L_{m}^{}}^{}, \;k \neq
m.\IEEEyessubnumber\label{Eqn:AlignmentConditionsPerUserUl}
\end{IEEEeqnarray*}
\vspace{-1em}
\begin{remark}[Interpretation of
\eqref{Eqn:AlignmentConditionsPerUser}]\label{Remark:AlignmentConditionsInterpretation}
We design the precoder matrices $\textbf{W}_{\!B,\,k}^{}$ and
$\textbf{W}_{k}^{}$ such that the two-way signals between the BS and
the \hbox{$k$-th} MS are perfectly aligned at the RS and linearly
independent to other users' signals.~ \hfill\IEEEQEDclosed
\end{remark}
Accordingly, the RS forwarded signals can be expressed as
\vspace{-1em}
\begin{IEEEeqnarray*}{l}
\textbf{x}_{R}^{} = \textbf{F}_{\!R\,}^{} ( \boldsymbol{\Phi} \,
\textbf{s}_{\,U}^{} + \boldsymbol{\Psi} \, \textbf{s}_{D}^{} +
\textbf{A}_{R}^{} \, \textbf{n}_{R}^{}
\,),\IEEEyesnumber\label{Eqn:RsTxSignal2}
\end{IEEEeqnarray*}
where $\boldsymbol{\Phi} \triangleq \textrm{diag}(
\boldsymbol{\phi}_{}^{\,(1)}\!, \ldots, \boldsymbol{\phi}_{}^{\,(K)}
\,)$ and $\boldsymbol{\Psi} \triangleq \textrm{diag}(
\boldsymbol{\psi}_{}^{\,(1)}\!, \ldots, \boldsymbol{\psi}_{}^{\,(K)}
\,)$ represent the effective gains of the UL and DL data streams,
respectively. The transmit power constraint of the RS forwarded
signals is given by $\mathbb{E}( ||\, \textbf{x}_{R}^{} \,||^{\,2\,}
) = ||\, \textbf{F}_{\!R\,}^{} \boldsymbol{\Phi} ||^{\,2\,} + ||\,
\textbf{F}_{\!R\,}^{} \boldsymbol{\Psi} ||^{\,2\,} + N_{\,0}^{} ||\,
\textbf{F}_{\!R\,}^{} \textbf{A}_{R}^{} \,||^{\,2\,} \!\leq
P_{R}^{}$. We define the SINRs of the RS forwarded signals as
\begin{IEEEeqnarray*}{l}\label{Eqn:Sinr1st1}
\textrm{UL: }\widetilde{\gamma}_{\,U}^{(k,\,l)} \!\triangleq\!
\frac{\textstyle|\, [\, \boldsymbol{\phi}_{}^{\,(k)} \,]_{(l)}^{}
\,|_{}^{\,2\,}}{N_{\,0}^{} ||\, [\, \textbf{A}_{R}^{\!(k)}
\,]_{(l,\,:)}^{} \,||_{}^{\,2\,}} \!=\! \frac{\textstyle|\, [\,
\textbf{A}_{R}^{\!(k)} \,]_{(l,\,:)}^{} \, \textbf{H}_{R\,,\,k}^{}
\, [\, \textbf{W}_{k}^{} \,]_{(:,\,l)} \,|_{}^{\,2\,}}{N_{\,0}^{}
||\, [\, \textbf{A}_{R}^{\!(k)}
\,]_{(l,\,:)}^{} \,||_{}^{\,2\,}},\IEEEyessubnumber\label{Eqn:Sinr1st1Ul}\\
\textrm{DL: }\widetilde{\gamma}_{D}^{(k,\,l)} \!\triangleq\!
\frac{\textstyle|\, [\, \boldsymbol{\psi}_{}^{\,(k)}
\,]_{(l)}^{}\,|_{}^{\,2\,}}{N_{\,0}^{} ||\, [\,
\textbf{A}_{R}^{\!(k)} \,]_{(l,\,:)}^{} \,||_{}^{\,2\,}} \!=\!
\frac{\textstyle|\, [\, \textbf{A}_{R}^{\!(k)} \,]_{(l,\,:)}^{} \,
\textbf{H}_{R\,,B\,}^{}\, [\, \textbf{W}_{\!B,\,k}^{}
\,]_{(:,\,l)}^{} \,|_{}^{\,2\,}}{N_{\,0}^{} ||\, [\,
\textbf{A}_{R}^{\!(k)} \,]_{(l,\,:)}^{}
\,||_{}^{\,2\,}}.\IEEEyessubnumber\label{Eqn:Sinr1st1Dl}
\end{IEEEeqnarray*}
From \eqref{Eqn:RsTxSignal2}, \eqref{Eqn:BsRecvSignal} and
\eqref{Eqn:MsRecvSignal}, the end-to-end received signals of the BS
can be expressed as
\begin{IEEEeqnarray*}{l}
\textbf{y}_{\!B}^{} = \underbrace{( \textbf{H}_{R\,,B}^{} )_{}^{T\,}
\textbf{F}_{\!R\,}^{} \boldsymbol{\Phi} \,
\textbf{s}_{\,U}^{}}_{\textrm{desired signals}} + \underbrace{(
\textbf{H}_{R\,,B}^{} )_{}^{T\,} \textbf{F}_{\!R\,}^{}
\textbf{A}_{R}^{} \, \textbf{n}_{R}^{} +
\textbf{n}_{B}^{}}_{\textrm{aggregate noise}} +\;
\textbf{i}_{B}^{},\IEEEyesnumber\label{Eqn:BsRecvSignal2}
\end{IEEEeqnarray*}
and the end-to-end received signals of the \hbox{$k$-th} MS can be
expressed as
\begin{IEEEeqnarray*}{ll}
\textbf{y}_{k}^{} &\;=\underbrace{( \textbf{H}_{R\,,\,k}^{}
)_{}^{T\,} \textbf{F}_{\!R}^{\,(k)} \boldsymbol{\psi}_{}^{\,(k)\,}
\textbf{s}_{D}^{(k)}}_{\textrm{desired signals}} + \underbrace{(
\textbf{H}_{R\,,\,k}^{} )_{}^{T\,} \textbf{F}_{\!R\,}^{} (
\widetilde{\boldsymbol{\Psi}}_{k\,}^{} \textbf{s}_{D}^{} \!+\!
\widetilde{\boldsymbol{\Phi}}_{k\,}^{} \textbf{s}_{\,U}^{}
)}_{\textrm{multi-user interference}} + \underbrace{(
\textbf{H}_{R\,,\,k}^{} )_{}^{T\,} \textbf{F}_{\!R\,}^{}
\textbf{A}_{R}^{} \, \textbf{n}_{R}^{} +\!
\textbf{n}_{\,k}^{}}_{\textrm{aggregate noise}} +\;
\textbf{i}_{\,k}^{},\;\;\;\;\IEEEyesnumber\label{Eqn:MsRecvSignal2}
\end{IEEEeqnarray*}
where $\widetilde{\boldsymbol{\Phi}}_{k}^{} \triangleq
\textrm{diag}( \boldsymbol{\phi}_{}^{\,(1)}\!, \ldots,
\boldsymbol{\phi}_{}^{\,(k-1)}\!, \textbf{0}_{L_{k}^{} \times 1}^{},
\boldsymbol{\phi}_{}^{\,(k+1)}\!, \ldots,
\boldsymbol{\phi}_{}^{\,(K)} \,)$ represent the effective gains of
the UL interference streams, and
$\widetilde{\boldsymbol{\Psi}}_{k}^{} \triangleq \textrm{diag}(
\boldsymbol{\psi}_{}^{\,(1)}\!, \ldots,
\boldsymbol{\psi}_{}^{\,(k-1)}\!, \textbf{0}_{L_{k}^{} \times 1}^{},
\boldsymbol{\psi}_{}^{\,(k+1)}\!, \ldots,
\boldsymbol{\psi}_{}^{\,(K)} \,)$ represent the effective gains of
the DL interference streams. Therefore, the UL and DL data stream
estimates \eqref{Eqn:BsSignalEstimates},
\eqref{Eqn:MsSignalEstimates} can be expressed as
\begin{IEEEeqnarray*}{l}
\!\!\!\!\widetilde{\textbf{s}}_{\,U}^{\,(k)} \!=\!
\textbf{V}_{\!B}^{\,(k)} ( \textbf{H}_{R\,,B}^{} )_{}^{T\,}
\textbf{F}_{\!R\,}^{} \boldsymbol{\Phi} \, \textbf{s}_{\,U}^{} +
\textbf{V}_{\!B}^{\,(k)} ( \textbf{H}_{R\,,B}^{} )_{}^{T\,}
\textbf{F}_{\!R\,}^{} \textbf{A}_{R}^{} \, \textbf{n}_{R}^{} +
\textbf{V}_{\!B}^{\,(k)} \,
\textbf{n}_{B}^{},\IEEEyesnumber\label{Eqn:SignalEstimates2Ul}\\
\!\!\!\!\widetilde{\textbf{s}}_{D}^{\,(k)} \!=\! \textbf{V}_{k}^{} (
\textbf{H}_{R\,,\,k}^{} )_{}^{T\,} \textbf{F}_{\!R}^{\,(k)}
\boldsymbol{\psi}_{}^{\,(k)\,} \textbf{s}_{D}^{(k)} \!+\!
\textbf{V}_{k}^{} ( \textbf{H}_{R\,,\,k}^{} )_{}^{T\,}
\textbf{F}_{\!R\,}^{} ( \widetilde{\boldsymbol{\Psi}}_{k\,}^{}
\textbf{s}_{D}^{} \!+\! \widetilde{\boldsymbol{\Phi}}_{k\,}^{}
\textbf{s}_{\,U}^{} ) \!+\! \textbf{V}_{k}^{} (
\textbf{H}_{R\,,\,k}^{} )_{}^{T\,} \textbf{F}_{\!R\,}^{}
\textbf{A}_{R}^{} \, \textbf{n}_{R}^{} +\! \textbf{V}_{k}^{} \,
\textbf{n}_{\,k}^{},\;\;\;\;\;\;\;\IEEEyesnumber\label{Eqn:SignalEstimates2Dl}
\end{IEEEeqnarray*}
and the end-to-end SINRs of the data stream estimates
\eqref{Eqn:UlSinr}, \eqref{Eqn:DlSinr} are \emph{equivalently} given
by
\begin{IEEEeqnarray*}{l}
\gamma_{\,U}^{(k,\,l)} =
\frac{\substack{\textstyle|\, [\, \textbf{V}_{\!B}^{\,(k)}
\,]_{(l,\,:)}^{} ( \textbf{H}_{R\,,B}^{} )_{}^{T} [\,
\textbf{F}_{\!R}^{\,(k)} \,]_{(:,\,l)}^{} [\,
\boldsymbol{\phi}_{}^{\,(k)} \,]_{(l)}^{}
\,|_{}^{\,2\,}}}{\left(\substack{\textstyle||\, [\,
\textbf{V}_{\!B}^{\,(k)} \,]_{(l,\,:)}^{} ( \textbf{H}_{R\,,B}^{}
)_{}^{T\,} \textbf{F}_{\!R}^{} \, \boldsymbol{\Phi} \,||_{}^{\,2\,}
\!- |\, [\, \textbf{V}_{\!B}^{\,(k)} \,]_{(l,\,:)}^{} (
\textbf{H}_{R\,,B}^{} )_{}^{T} [\, \textbf{F}_{\!R}^{\,(k)}
\,]_{(:,\,l)}^{} [\,
\boldsymbol{\phi}_{}^{\,(k)} \,]_{(l)}^{} \,|_{}^{\,2\,}\\
\textstyle+ N_{\,0}^{} ( ||\, [\, \textbf{V}_{\!B}^{\,(k)}
\,]_{(l,\,:)}^{} ( \textbf{H}_{R\,,B}^{} )_{}^{T\,}
\textbf{F}_{\!R\,}^{} \textbf{A}_{R}^{} \,||_{}^{\,2\,} \!+ ||\, [\,
\textbf{V}_{\!B}^{\,(k)} \,]_{(l,\,:)}^{} \,||_{}^{\,2\,}
)}\right)},\IEEEyesnumber\label{Eqn:Sinr2Ul}\\
\gamma_{D}^{(k,\,l)} = \frac{\substack{\textstyle|\, [\,
\textbf{V}_{k}^{} \,]_{(l,\,:)}^{} ( \textbf{H}_{R\,,\,k}^{}
)_{}^{T} [\, \textbf{F}_{\!R}^{\,(k)} \,]_{(:,\,l)}^{} [\,
\boldsymbol{\psi}_{}^{\,(k)}
\,]_{(l)}^{}\,|_{}^{\,2\,}}}{\left(\substack{\textstyle||\, [\,
\textbf{V}_{k}^{} \,]_{(l,\,:)}^{} ( \textbf{H}_{R\,,\,k}^{}
)_{}^{T\,} \textbf{F}_{\!R}^{} \, [\, \boldsymbol{\Psi},
\widetilde{\boldsymbol{\Phi}}_{k}^{} \,] \,||_{}^{\,2\,} \!- |\, [\,
\textbf{V}_{k}^{} \,]_{(l,\,:)}^{} ( \textbf{H}_{R\,,\,k}^{}
)_{}^{T} [\, \textbf{F}_{\!R}^{\,(k)} \,]_{(:,\,l)}^{} [\,
\boldsymbol{\psi}_{}^{\,(k)} \,]_{(l)}^{}\,|_{}^{\,2\,}\\
\textstyle+ N_{\,0}^{} ( ||\, [\, \textbf{V}_{k}^{} \,]_{(l,\,:)}^{}
( \textbf{H}_{R\,,\,k}^{} )_{}^{T\,} \textbf{F}_{\!R\,}^{}
\textbf{A}_{R}^{} \,||_{}^{\,2\,} \!+ ||\, [\, \textbf{V}_{k}^{}
\,]_{(l,\,:)}^{}
\,||_{}^{\,2\,})}\right)}.\;\;\;\;\IEEEyesnumber\label{Eqn:Sinr2Dl}
\end{IEEEeqnarray*}

\begin{lemma}[Decomposition of Transceiver
Design]\label{Lemma:Decomposition} The transceiver design problem,
Problem~$\mathcal{Q}$, can be equivalently decomposed into two
stages. The first stage processing finds the BS and MS precoder
matrices and the RS equalizer matrix $\big\{ \textbf{W}_{\!B}^{}, \{
\textbf{W}_{k}^{} \}_{k=1}^{K}, \textbf{A}_{R}^{} \big\}$, subject
to the alignment conditions \eqref{Eqn:AlignmentConditionsPerUser},
to maximize the minimum weighted SINR of the RS forwarded
signals.\\\underline{First Stage Processing}\\\indent\;\;\;$\big\{
\textbf{W}_{\!B}^{\,\star}, \{\:\! \textbf{W}_{k}^{\,\star\,}
\}_{k=1}^{K}, \textbf{A}_{R}^{\!\star} \big\} := \mathcal{M}\big\{
P_{B}^{}, \{ P_{\,k}^{} \}_{k=1}^{K}, \{
\boldsymbol{\omega}_{U}^{(k)}\!, \boldsymbol{\omega}_{\!D}^{(k)\,}
\}_{k=1}^{K} \big\}$\nopagebreak[4]
\begin{subnumcases}{\label{Eqn:ProbFstStage}}
\!\underset{\substack{\textbf{W}_{\!B}^{}, \{\:\!
\textbf{W}_{\!k}^{} \}_{k=1\!}^{K},\, \textbf{A}_{R}^{}}}{\arg \max}
\!\!\!\!\!\!\!\! & $\displaystyle\min_{\substack{\forall\, k \,\in\,
\mathcal{K}\\\forall\, l \,\in\, \mathcal{L}_{k}^{}}}
\textstyle\big\{ \big( [\, \boldsymbol{\omega}_{U}^{(k)}
\,]_{(l)}^{} \big)_{}^{-1} \widetilde{\gamma}_{\,U}^{(k,\,l)}, \big(
[\, \boldsymbol{\omega}_{\!D}^{(k)} \,]_{(l)}^{} \big)_{}^{-1}
\widetilde{\gamma}_{D}^{(k,\,l)} \, \big\}$\label{Eqn:ProbFstStageObjective}\\
\;\;\;\;\;\;\textrm{s.t.} & $||\, \textbf{W}_{\!B}^{}
\,||_{}^{\,2\,} \leq P_{B}^{}, \;\;||\, \textbf{W}_{k}^{}
\,||^{\,2\,} \leq P_{\,k}^{}, \;\forall k \in
\mathcal{K},$\label{Eqn:ProbFstStageBsMsPowerConstraint}\\
& $\textbf{A}_{R}^{\!(k)} \, \textbf{H}_{R\,,\,k}^{}
\textbf{W}_{k}^{} = \textrm{diag}( \boldsymbol{\phi}_{}^{\,(k)} \,),
\;\; \textbf{A}_{R}^{\!(k)} \, \textbf{H}_{R\,,\,m}^{}
\textbf{W}_{m}^{} = \textbf{0}_{L_{k}^{} \times L_{m}^{}}^{}, \;k
\neq m,$\label{Eqn:ProbFstStageAlignUl}\\
& $\textbf{A}_{R}^{\!(k)} \, \textbf{H}_{R\,,B}^{}
\textbf{W}_{\!B,\,k}^{} = \textrm{diag}(
\boldsymbol{\psi}_{}^{\,(k)} \,), \;\; \textbf{A}_{R}^{\!(k)} \,
\textbf{H}_{R\,,B}^{} \textbf{W}_{\!B,\,m}^{} = \textbf{0}_{L_{k}^{}
\times L_{m}^{}}^{}, \;k \neq m.$\label{Eqn:ProbFstStageAlignDl}
\end{subnumcases}
The second stage processing finds the RS precoder matrix and the BS
and MS equalizer matrices $\big\{ \textbf{F}_{\!R}^{},\!
\textbf{V}_{\!B}^{}, \{ \textbf{V}_{k}^{} \}_{k=1}^{K} \big\}$ to
maximize the minimum weighted end-to-end SINR of the data stream
estimates.
\\\underline{Second Stage
Processing}\\\indent\indent\indent\indent\indent\indent\;$\big\{
\textbf{F}_{\!R}^{\,\star}, \textbf{V}_{\!B}^{\,\star}, \{
\textbf{V}_{k}^{\,\star\,} \}_{k=1}^{K} \big\} \!:=
\mathcal{B}\big\{ \textbf{W}_{\!B}^{\,\star}, \{\:\!
\textbf{W}_{k}^{\,\star\,} \}_{k=1}^{K}, \textbf{A}_{R}^{\!\star},
P_{R}^{}, \{ \boldsymbol{\omega}_{U}^{(k)}\!,
\boldsymbol{\omega}_{\!D}^{(k)\,} \}_{k=1}^{K}
\big\}$\nopagebreak[4]
\begin{subnumcases}{\label{Eqn:ProbSndStage}}
\!\underset{\substack{\textbf{F}_{\!R}^{}, \textbf{V}_{\!B}^{}, \{
\textbf{V}_{k}^{} \}_{k=1}^{K}}}{\arg \max} \!\!\!\!\!\!\!\! &
$\displaystyle\min_{\substack{\forall\, k \,\in\,
\mathcal{K}\\\forall\, l \,\in\, \mathcal{L}_{k}^{}}}
\textstyle\big\{ \big( [\, \boldsymbol{\omega}_{U}^{(k)}
\,]_{(l)}^{} \big)_{}^{-1} \gamma_{\,U}^{(k,\,l)}, \big( [\,
\boldsymbol{\omega}_{\!D}^{(k)} \,]_{(l)}^{} \big)_{}^{-1}
\gamma_{D}^{(k,\,l)} \, \big\}$\label{Eqn:ProbSndStageObjective}\\
\;\;\;\;\;\;\textrm{s.t.} & $\textstyle||\, \textbf{F}_{\!R\,}^{}
\boldsymbol{\Phi} ||^{\,2\,} + ||\, \textbf{F}_{\!R\,}^{}
\boldsymbol{\Psi} ||^{\,2\,} + N_{\,0}^{} ||\, \textbf{F}_{\!R\,}^{}
\textbf{A}_{R}^{} \,||^{\,2\,} \!\leq
P_{R\,}^{}.$\label{Eqn:ProbSndStageRsPowerConstraint}
\end{subnumcases}
\end{lemma}
\begin{proof}
Refer to Appendix~A.
\end{proof}

The top-level steps of the proposed two-stage transceiver design are
summarized in Algorithm~\ref{Algorithm:TopLevel} and illustrated in
\figurename~\ref{Fig:Problems}. We shall elaborate the details of
the first and second stage processing in the following subsections.

\subsection{First Stage Processing} \label{Sec:FirstStage} To solve
the first stage processing, Problem~$\mathcal{M}$, we focus our
attention on \emph{coordinative eigenmode transmission} at the MSs,
\emph{zero-forcing equalization} at the RS, and \emph{zero-forcing
transmission} at the BS \cite{Jnl:Coordinated_beamforming:Chae}.
\begin{list}{\labelitemi}{\leftmargin=0.5em}
\item \emph{MS Precoder Matrices:} Let $\mathcal{G}_{\,k}^{}$ denote
the set of the right singular vectors of the channel matrix
$\textbf{H}_{R\,,\,k}^{}$. The \hbox{$k$-th} MS precoder matrix
$\textbf{W}_{k}^{}$ is given as $[\, \textbf{W}_{k}^{} \,]_{(:,\,l)}
= \sqrt{\lambda_{\,k}^{(l)}} \textbf{g}_{\,k}^{\,(l)}$, where
$\textbf{g}_{\,k}^{\,(l)} \in \mathcal{G}_{\,k}^{}$ is the beam
direction, and $\lambda_{\,k}^{(l)}$ is the allocated power
satisfying the transmit power constraint $||\, \textbf{W}_{k}^{}
\,||^{\,2\,} = \sum_{\,l = 1}^{\,L_{k}^{}} \lambda_{\,k}^{(l)} =
P_{\,k}^{}$.
\item \emph{RS Equalizer Matrix:} The zero-forcing equalizer
matrix is given as
\begin{IEEEeqnarray*}{l}
\textbf{A}_{R}^{} = \textrm{pinv}( [\, \textbf{H}_{R\,,1}^{} [\,
\textbf{g}_{\,1}^{\,(1)}, \ldots, \textbf{g}_{\,1}^{\,(L_{1}^{})}
\,], \ldots, \textbf{H}_{R\,,K\,}^{} [\, \textbf{g}_{K}^{\,(1)},
\ldots, \textbf{g}_{K}^{\,(L_{K}^{})} \,] \,]
).\IEEEyesnumber\label{Eqn:RsEqualizer}
\end{IEEEeqnarray*}
\item \emph{BS Precoder Matrix:} Let
$\widetilde{\textbf{A}}_{R}^{(k,\,l)}$ denote the matrix obtained by
\emph{removing} $[\, \textbf{A}_{R}^{\!(k)} \,]_{(l,\,:)}^{}$ from
$\textbf{A}_{R}^{}$ and let $\mathcal{G}_{B}^{\,(k,\,l)} =
\textrm{null}( \widetilde{\textbf{A}}_{R}^{(k,\,l)\,}
\textbf{H}_{R\,,B\,}^{} )$. The BS precoder matrix
$\textbf{W}_{\!B}^{}$ is given as \hbox{$[\, \textbf{W}_{\!B,\,k}^{}
\,]_{(:,\,l)}^{} = \sqrt{\lambda_{B}^{(k,\,l)}}
\textbf{g}_{B}^{\,(k,\,l)}$}, where $\textbf{g}_{B}^{\,(k,\,l)} \in
\mathcal{G}_{B}^{\,(k,\,l)}$ is the beam direction, and
$\lambda_{B}^{(k,\,l)}$ is the allocated power satisfying the
transmit power constraint $||\, \textbf{W}_{\!B}^{} \,||_{}^{\,2\,}
= \sum_{k = 1}^{K} \sum_{\,l = 1}^{\,L_{k}^{}} \lambda_{B}^{(k,\,l)}
= P_{B}^{}$.
\end{list}
As such, the SINRs of the RS forwarded signals \eqref{Eqn:Sinr1st1}
can be \emph{equivalently} expressed as
\begin{IEEEeqnarray*}{c}
\widetilde{\gamma}_{\,U}^{(k,\,l)} = \kappa_{\,U}^{(k,\,l)}
\lambda_{\,k}^{(l)}, \;\; \widetilde{\gamma}_{D}^{(k,\,l)} =
\kappa_{D}^{(k,\,l)}
\lambda_{B}^{(k,\,l)},\IEEEyesnumber\label{Eqn:Sinr1st2}
\end{IEEEeqnarray*}
where $\kappa_{\,U}^{(k,\,l)} \triangleq \frac{\textstyle|\, [\,
\textbf{A}_{R}^{\!(k)} \,]_{(l,\,:)}^{} \, \textbf{H}_{R\,,\,k}^{}
\, \textbf{g}_{\,k}^{\,(l)} \,|_{}^{\,2}}{\textstyle N_{\,0}^{} ||\,
[\, \textbf{A}_{R}^{\!(k)} \,]_{(l,\,:)}^{} \,||_{}^{\,2}}$ and
$\kappa_{D}^{(k,\,l)} \triangleq \frac{\textstyle|\, [\,
\textbf{A}_{R}^{\!(k)} \,]_{(l,\,:)}^{} \, \textbf{H}_{R\,,B\,}^{}\,
\textbf{g}_{B}^{\,(k,\,l)} \,|_{}^{\,2}}{\textstyle N_{\,0}^{} ||\,
[\, \textbf{A}_{R}^{\!(k)} \,]_{(l,\,:)}^{} \,||_{}^{\,2}}$.
\begin{lemma}[Beam Directions and Power Allocation at the BS and MSs]
\label{Lemma:BeamSelPowerAllocation} To maximize the minimum
weighted SINR of the RS forwarded signals, the power allocation at
the BS and MSs are, respectively, given by
\begin{IEEEeqnarray*}{l}
\big( \lambda_{B}^{(k,\,l)} \big)_{}^{\!\star} \!\!=\!\!
\frac{\textstyle[\, \boldsymbol{\omega}_{\!D}^{(k)} \,]_{(l)}^{}
\big( \kappa_{D}^{(k,\,l)} \big)_{}^{-1} P_{B}^{}}{\textstyle\sum_{m
= 1}^{K} \! \sum_{\,q = 1}^{\,L_{m}^{}} \, [\,
\boldsymbol{\omega}_{\!D}^{(m)} \,]_{(q)}^{} \big(
\kappa_{D}^{(m,\,q)} \big)_{}^{-1}}, \;\; \big( \lambda_{\,k}^{(l)}
\big)_{}^{\!\star} \!\!=\!\! \frac{\textstyle[\,
\boldsymbol{\omega}_{U}^{(k)} \,]_{(l)}^{} \big(
\kappa_{\,U}^{(k,\,l)} \big)_{}^{-1} P_{\,k}^{}}{\textstyle\sum_{\,q
= 1}^{\,L_{k}^{}} \, [\, \boldsymbol{\omega}_{U}^{(k)} \,]_{(q)}^{}
\big( \kappa_{\,U}^{(k,\,q)}
\big)_{}^{-1}}.\IEEEyesnumber\label{Eqn:PowerAllocation}
\end{IEEEeqnarray*}
It follows that the weighted SINRs of the RS forwarded signals can
be expressed as
\begin{IEEEeqnarray*}{l}\label{Eqn:Eqn:Sinr1stPowerAllocation}
\textrm{UL: }\big( [\, \boldsymbol{\omega}_{U}^{(k)} \,]_{(l)}^{}
\big)_{}^{-1} \widetilde{\gamma}_{\,U}^{(k,\,l)} \!=\!
\frac{P_{\,k}^{}}{\sum_{\,q = 1}^{\,L_{k}^{}} \, [\,
\boldsymbol{\omega}_{U}^{(k)} \,]_{(q)}^{} \big(
\kappa_{\,U}^{(k,\,q)}
\big)_{}^{-1}},\IEEEyessubnumber\label{Eqn:Eqn:Sinr1stPowerAllocationUl}\\
\textrm{DL: }\big( [\, \boldsymbol{\omega}_{\!D}^{(k)} \,]_{(l)}^{}
\big)_{}^{-1} \widetilde{\gamma}_{D}^{(k,\,l)} \!=\!
\frac{P_{B}^{}}{\sum_{m = 1}^{K} \sum_{\,q = 1}^{\,L_{m}^{}} \, [\,
\boldsymbol{\omega}_{\!D}^{(m)} \,]_{(q)}^{} \big(
\kappa_{D}^{(m,\,q)}
\big)_{}^{-1}},\IEEEyessubnumber\label{Eqn:Eqn:Sinr1stPowerAllocationDl}
\end{IEEEeqnarray*}
and selection of the beam directions to maximize the minimum
weighted SINR can be performed using combinatorial search.
\end{lemma}
\begin{proof}
Refer to Appendix~B.
\end{proof}

\subsection{Second Stage Processing} \label{Sec:SecondStage} We now proceed to
describe the algorithm for solving the second stage processing,
Problem~$\mathcal{B}$. As per
\eqref{Eqn:Sinr2Ul}-\eqref{Eqn:Sinr2Dl}, the SINRs of the data
stream estimates are not jointly convex in the RS precoder matrix
and the BS and MS equalizer matrices $\big\{ \textbf{F}_{\!R}^{},\!
\textbf{V}_{\!B}^{}, \{ \textbf{V}_{k}^{} \}_{k=1}^{K} \big\}$.
However, for a fixed precoder matrix $\textbf{F}_{\!R\,}^{}$ there
are closed-form solutions for the equalizer matrices
$\textbf{V}_{\!B}^{}$ and $\textbf{V}_{k}^{}$; conversely, for fixed
$\textbf{V}_{\!B}^{}$ and $\textbf{V}_{k}^{}$ we can cast the
problem of solving for $\textbf{F}_{\!R\,}^{}$ as a
\emph{quasi-convex} problem. This motivates the approach to
progressively refine the transceiver matrices by iteratively
alternate between solving for the BS and MS equalizer matrices and
the RS precoder matrix. In this regard, we alternatingly optimize
each one of the RS precoder matrix and the BS and MS equalizer
matrices in the form of the following subproblems, and the
convergence proof is provided in Appendix~C.\\\underline{RS Precoder
Matrix}\;\;$\big\{ \textbf{F}_{\!R}^{\,\star}, \gamma_{0}^{} \big\}
\!:= \mathcal{B}_{R}^{}\big\{ \textbf{V}_{\!B}^{}, \{
\textbf{V}_{k}^{} \}_{k=1}^{K}, \textbf{W}_{\!B}^{\,\star}, \{\:\!
\textbf{W}_{k}^{\,\star\,} \}_{k=1}^{K}, \textbf{A}_{R}^{\!\star},
P_{R}^{}, \{ \boldsymbol{\omega}_{U}^{(k)}\!,
\boldsymbol{\omega}_{\!D}^{(k)\,} \}_{k=1}^{K}
\big\}$\nopagebreak[4]
\begin{subnumcases}{\qquad\qquad\qquad\qquad\label{Eqn:ProbSndStagePrecoderStdProb}}
\!\underset{\substack{\textbf{F}_{\!R\,}^{},\, \gamma_{0}^{}}}{\max}
\!\!\!\!\!\!\!\!\!\! & \!\!\!\!\!\!\!\!
$\gamma_{\,0}^{}$\label{Eqn:ProbSndStagePrecoderStdProbObjective}\\
\;\textrm{s.t.} & \!\!\!\!\!\!\!\!\!\! $\big( [\,
\boldsymbol{\omega}_{U}^{(k)} \,]_{(l)}^{} \big)_{}^{-1}
\gamma_{\,U}^{(k,\,l)}, \, \big( [\, \boldsymbol{\omega}_{\!D}^{(k)}
\,]_{(l)}^{} \big)_{}^{-1} \gamma_{D}^{(k,\,l)} \ge \gamma_{\,0}^{},
\,\forall k \!\in\! \mathcal{K}, \forall l
\!\in\! \mathcal{L}_{\,k}^{},\;\;\;\;$\label{Eqn:ProbSndStagePrecoderStdProbSinrConstraints}\\
& \!\!\!\!\!\!\!\! $\textstyle||\, \textbf{F}_{\!R\,}^{}
\boldsymbol{\Phi} ||^{\,2\,} + ||\, \textbf{F}_{\!R\,}^{}
\boldsymbol{\Psi} ||^{\,2\,} + N_{\,0}^{} ||\, \textbf{F}_{\!R\,}^{}
\textbf{A}_{R}^{} \,||^{\,2\,} \!\leq
P_{R\,}^{}.$\label{Eqn:ProbSndStagePrecoderStdProbPowerConstraint}
\end{subnumcases}
\underline{BS
Equalizer}\indent\indent\indent\indent\indent~~$\textbf{V}_{\!B}^{\,\star}
:= \mathcal{B}_{B}^{}\big\{ \textbf{F}_{\!R}^{},
\textbf{W}_{\!B}^{\,\star}, \{\:\! \textbf{W}_{k}^{\,\star\,}
\}_{k=1}^{K}, \textbf{A}_{R}^{\!\star}, \{
\boldsymbol{\omega}_{U}^{(k)} \}_{k=1}^{K} \big\}$\nopagebreak[4]
\begin{numcases}{}
\!\!\underset{\textbf{V}_{\!B}^{}}{\arg \max} \!\!\!\!\!\!\!\!\!\! &
$\displaystyle\min_{\substack{\forall\, k \,\in\,
\mathcal{K}\\\forall\, l \,\in\, \mathcal{L}_{k}^{}}}
\textstyle\big\{ \big( [\, \boldsymbol{\omega}_{U}^{(k)}
\,]_{(l)}^{} \big)_{}^{-1} \gamma_{\,U}^{(k,\,l)}
\big\}$.\label{Eqn:ProbSndStageBsEqualizer}
\end{numcases}
\underline{\hbox{$k$-th} MS
Equalizer}\indent\indent\indent~$\textbf{V}_{k}^{\,\star\,} :=
\mathcal{B}_{k}^{}\big\{ \textbf{F}_{\!R}^{},
\textbf{W}_{\!B}^{\,\star}, \{\:\! \textbf{W}_{k}^{\,\star\,}
\}_{k=1}^{K}, \textbf{A}_{R}^{\!\star},
\boldsymbol{\omega}_{\!D}^{(k)} \big\}$\nopagebreak[4]
\begin{numcases}{}
\!\!\underset{\textbf{V}_{\!k}^{}}{\arg \max} \!\!\!\!\!\!\!\!\! &
$\displaystyle\min_{\substack{\forall\, l \,\in\,
\mathcal{L}_{k}^{}}} \textstyle\big\{ \big( [\,
\boldsymbol{\omega}_{\!D}^{(k)} \,]_{(l)}^{} \big)_{}^{-1}
\gamma_{D}^{(k,\,l)} \big\}$.\label{Eqn:ProbSndStageMsEqualizer}
\end{numcases}

First, for Problem~$\mathcal{B}_{B}^{}$ and
Problem~$\mathcal{B}_{k}^{}$, the per stream SINRs are maximized
with a minimum mean squared error (MMSE) equalizer matrix
\cite{Jnl:Unified_framework_convex:Palomar}. Hence, the BS equalizer
matrix is given by
\begin{IEEEeqnarray*}{l}
\textbf{V}_{\!B}^{\,\star} = \big( ( \textbf{H}_{R\,,B}^{}
)_{}^{T\,} \textbf{F}_{\!R}^{} \, \boldsymbol{\Phi}
\big)_{}^{\!\dag\,} \big( ( \textbf{H}_{R\,,B}^{} )_{}^{T\,}
\textbf{F}_{\!R}^{} \, \boldsymbol{\Phi} \big( (
\textbf{H}_{R\,,B}^{} )_{}^{T\,} \textbf{F}_{\!R}^{} \,
\boldsymbol{\Phi} \big)_{}^{\!\dag} + \boldsymbol{\Omega}_{B}^{}
\big)_{}^{-1},\IEEEyesnumber\label{Eqn:BsEqualizerSoln}
\end{IEEEeqnarray*}
where $\boldsymbol{\Omega}_{B}^{} \triangleq N_{\,0}^{} (
\textbf{H}_{R\,,B}^{} )_{}^{T\,} \textbf{F}_{\!R\,}^{}
\textbf{A}_{R}^{} \big( ( \textbf{H}_{R\,,B}^{} )_{}^{T\,}
\textbf{F}_{\!R\,}^{} \textbf{A}_{R}^{} \big)_{}^{\!\dag} +
N_{\,0}^{} \, \textbf{I}_{N_{B}^{}}^{}$ is the covariance matrix of
the aggregate noise at the BS. Likewise, the \hbox{$k$-th} MS
equalizer matrix is given by
\begin{IEEEeqnarray*}{l}
\textbf{V}_{k}^{\,\star\,} = \big( ( \textbf{H}_{R\,,\,k}^{}
)_{}^{T\,} \textbf{F}_{\!R}^{\,(k)} \boldsymbol{\psi}_{}^{\,(k)}
\big)_{}^{\!\dag\,} \big( ( \textbf{H}_{R\,,\,k}^{} )_{}^{T\,}
\textbf{F}_{\!R}^{} \, [\, \boldsymbol{\Psi},
\widetilde{\boldsymbol{\Phi}}_{k}^{} \,] \big( (
\textbf{H}_{R\,,\,k}^{} )_{}^{T\,} \textbf{F}_{\!R}^{} \, [\,
\boldsymbol{\Psi}, \widetilde{\boldsymbol{\Phi}}_{k}^{} \,]
\big)_{}^{\!\dag} + \boldsymbol{\Omega}_{\,k}^{}
\big)_{}^{-1},\IEEEyesnumber\label{Eqn:MsEqualizerSoln}
\end{IEEEeqnarray*}
where $\boldsymbol{\Omega}_{\,k}^{} \triangleq N_{\,0}^{} (
\textbf{H}_{R\,,\,k}^{} )_{}^{T\,} \textbf{F}_{\!R\,}^{}
\textbf{A}_{R}^{} \big( ( \textbf{H}_{R\,,\,k}^{} )_{}^{T\,}
\textbf{F}_{\!R\,}^{} \textbf{A}_{R}^{} \big)_{}^{\!\dag} +
N_{\,0}^{} \, \textbf{I}_{N_{k}^{}}^{}$ is the covariance matrix of
the aggregate noise at the \hbox{$k$-th} MS.

Second, note that Problem~$\mathcal{B}_{R}^{}$ corresponds to
designing precoders for $L$ two-user multicast groups, where the
precoder $[\, \textbf{F}_{\!R}^{\,(k)} \,]_{(:,\,l)}^{}$ is used for
multicasting the signal stream that encapsulates the UL data stream
$[\,\textbf{s}_{\,U}^{\,(k)}\,]_{(l)}^{}$ and the DL data stream
$[\,\textbf{s}_{D}^{\,(k)}\,]_{(l)}^{}$. As per
\cite[Claim~2]{Jnl:MBS-SDMA_using_SDR_with_perfect_CSI:Sidiropoulos_Luo},
this multigroup multicast problem is NP-hard\footnote{Please refer
to \cite{Jnl:Linear_precoding_conic:Eldar,
Jnl:Transceiver_QoS_per_antenna_power:Tolli} for discussion on the
unicast precoder design problem, which is not NP-hard.}. We propose
to solve for the RS precoder matrix $\textbf{F}_{\!R}^{}$ using
Algorithm~\ref{Algorithm:RsPrecoder} as derived in Appendix~D. In a
nutshell, we cast Problem~$\mathcal{B}_{R}^{}$ as a quasi-convex
problem and solve it using the bisection method
\cite[Section~4.2.5]{Bok:Convex_optimization:Boyd}. To do so, we
define the SOCP feasibility problem of designing
$\textbf{F}_{\!R}^{}$ that achieves a \emph{target} value of the
minimum weighted per stream SINR $\gamma_{\,0}^{}$ as\footnote{Due
to page limit, please refer to Appendix~D for the expressions for
\eqref{Eqn:ProbSndStagePrecoderFeasibilityProbSinrConstraints} and
\eqref{Eqn:ProbSndStagePrecoderFeasibilityProbPowerConstraint}.}\\\indent\;$\textbf{F}_{\!R}^{\,\star}
:= \widetilde{\mathcal{B}}_{R}^{}\big\{ \gamma_{\,0}^{},\!
\textbf{V}_{\!B}^{}, \{ \textbf{V}_{k}^{} \}_{k=1}^{K},
\textbf{W}_{\!B}^{\,\star}, \{\:\! \textbf{W}_{k}^{\,\star\,}
\}_{k=1}^{K}, \textbf{A}_{R}^{\!\star}, P_{R}^{}, \{
\boldsymbol{\omega}_{U}^{(k)}\!, \boldsymbol{\omega}_{\!D}^{(k)\,}
\}_{k=1}^{K} \big\}$\nopagebreak[4]
\begin{subnumcases}{\label{Eqn:ProbSndStagePrecoderFeasibilityProb}}
\textrm{find} & \!\!\!\!\!\!\!\!\!\!
$\textbf{F}_{\!R}^{}$\label{Eqn:ProbSndStagePrecoderFeasibilityProbObjective}\\
\,\textrm{s.t.} & \!\!\!\!\!\!\!\!\!\!\! $\big[\,
\widetilde{\alpha}_{\,U}^{\,(k,\,l)} \,;\,
\boldsymbol{\beta}_{\,U}^{\,(k,\,l)} \,; \delta_{\,U}^{\,(k,\,l)}
\,\big] \!\succeq_{K}^{}\! 0, \;\; \big[\,
\widetilde{\alpha}_{D}^{\,(k,\,l)} \,;\,
\boldsymbol{\beta}_{D}^{\,(k,\,l)} \,; \delta_{D}^{\,(k,\,l)}
\,\big] \!\succeq_{K}^{}\! 0, \;\;\forall k \in \mathcal{K}, \forall
l \in
\mathcal{L}_{\,k}^{},\;$\label{Eqn:ProbSndStagePrecoderFeasibilityProbSinrConstraints}\\
& \!\!\!\!\!\!\!\!\!\!\! $\big[\, \sqrt{P_{R}^{}} \;;\,
\boldsymbol{\rho} \,\big] \!\succeq_{K}^{}\!
0.$\label{Eqn:ProbSndStagePrecoderFeasibilityProbPowerConstraint}
\end{subnumcases}
Starting with an interval that is expected to contain the
\emph{optimum} value of $\gamma_{\,0}^{}$, we repeatedly bisect the
interval and select the subinterval in which
Problem~$\widetilde{\mathcal{B}}_{R}^{}$ is feasible until
$\gamma_{\,0}^{}$ converges.

\subsection{Implementation Considerations} \label{Sec:Implementation}
First, we make the following assumptions about the synchronization
requirement on the UL and DL signals.
\begin{assumption}[Synchronization
Requirement]\label{Assumption:Synchronization} The transmitted
signals of the BS and the MSs are frequency and time synchronous
\cite[Section~2.2]{Thesis:Two_way_relay:Unger}. For instance, in a
practical system such as IEEE~802.16m, the BS and MSs would be
scheduled to transmit and receive over the same frequency-time
resource units \cite[Section~16.3.4.1]{Std:16m}.~
\hfill\IEEEQEDclosed
\end{assumption}
Second, we make the following assumptions on the side information
available at each node to facilitate transceiver design and
self-interference cancelation. Under these assumptions, the proposed
transceiver design problem can be solved in a distributed fashion.

\begin{assumption}[Side Information at the
RS]\label{Assumption:SideInfoRs} The RS has knowledge of global
channel state information (CSI) $\mathcal{H}_{R}^{} = \{
\textbf{H}_{R\,,1}^{}, \ldots, \textbf{H}_{R\,,K}^{},
\textbf{H}_{R\,,B}^{} \}$. For instance, the RS can accurately
estimate the channel matrices of all links by observing the
reciprocal reverse channels.~ \hfill\IEEEQEDclosed
\end{assumption}

As per Assumption~\ref{Assumption:SideInfoRs}, the RS can locally
solve the transceiver design problem (using
Algorithm~\ref{Algorithm:TopLevel}) and broadcast the RS
transformation matrix to the BS and MSs.

\begin{assumption}[Side Information at the BS and
MSs]\label{Assumption:SideInfoBsMs} The BS and each MS has knowledge
of the channel matrix between itself and the RS, the two-hop
effective channel matrix, and the RS transformation matrix. Thus,
the side information at the BS and the \hbox{$k$-th} MS include
\begin{IEEEeqnarray*}{l}
\textrm{BS: } \mathcal{H}_{B}^{} \!=\! \{ \textbf{H}_{R\,,B}^{}, (
\textbf{H}_{R\,,B}^{} )_{}^{T\,} \textbf{F}_{\!R}^{} \,
\boldsymbol{\Phi}, \textbf{F}_{\!R\,}^{}, \textbf{A}_{R\,}^{} \}, \;
\textrm{$k$-th MS: } \mathcal{H}_{\,k}^{} \!=\! \{
\textbf{H}_{R\,,\,k}^{}, ( \textbf{H}_{R\,,\,k}^{} )_{}^{T\,}
\textbf{F}_{\!R}^{} \, [\, \boldsymbol{\Psi},
\widetilde{\boldsymbol{\Phi}}_{k}^{} \,], \textbf{F}_{\!R\,}^{},
\textbf{A}_{R\,}^{} \}.
\end{IEEEeqnarray*}
For instance, the BS and each MS can estimate the channel matrix
between itself and the RS by observing the reciprocal reverse
channel, and can estimate the two-hop effective channel matrix using
pilot-assisted techniques.~ \hfill\IEEEQEDclosed
\end{assumption}

As per Assumption~\ref{Assumption:SideInfoBsMs}, the BS and MSs can
locally determine their precoder and equalizer matrices (using
Lemma~\ref{Lemma:BeamSelPowerAllocation},
\eqref{Eqn:BsEqualizerSoln}, and \eqref{Eqn:MsEqualizerSoln}), and
have sufficient information to deduce and cancel self-interference.

\begin{remark}[Channel Estimation and Feedback]
In a practical system such as IEEE~802.16m, pilot symbols are
embedded in frequency-time resource units to facilitate channel
estimation \cite[Section~16.3.4.4]{Std:16m}, and each node can
perform channel estimation using techniques such as those defined in
\cite{Cnf:Analog_feedback_Moto1:Thomas} and references therein. On
the other hand, the RS can broadcast the RS transformation matrix to
the BS and MSs by means of high fidelity unquantized feedback
\cite[Section~16.3.6.2.5.6]{Std:16m}.~ \hfill\IEEEQEDclosed
\end{remark}

\section{Simulation Results and Discussions} \label{Sec:SimsAndDiscussions}
In this section, we provide numerical simulation results to assess
the performance of the proposed transceiver design. For
illustration, we consider the following simulation settings.

\subsection{Simulation Settings} We consider a system with $K = 3$
MSs. In particular, we focus on MIMO configurations similar to those
defined in the IEEE~802.16m standard \cite{Std:16m}: the BS is
equipped with \emph{up to} $N_{B}^{} = 8$ antennas and the MSs are
equipped with $N_{\,k}^{} = \{ 2, 4 \}$ antennas. As an example, we
investigate the scenario in which the BS exchanges $L_{\,1}^{} = 2$,
$L_{\,2}^{} = 1$, and $L_{\,3}^{} = 1$ data streams with the MSs.

We evaluate the performance of the proposed scheme using the packet
error rate (PER) and the average sum rate\footnote{The average sum
rate is defined as $\mathbb{E} \,\big[ \sum_{k = 1}^{K} \sum_{\,l =
1}^{L_{k}^{}} ( C_{\,U}^{\,(k,\,l)} \!+ C_{D}^{\,(k,\,l)} )
\,\big]$, where $C_{\,U}^{\,(k,\,l)}$ and $C_{D}^{\,(k,\,l)}$ are
the UL and DL per stream achievable data rates, respectively, as
given in \eqref{Eqn:Cap}.} as performance metrics. In the PER
simulations, we employ the convolutional turbo code (CTC) defined in
the IEEE~802.16m standard \cite[Section~16.3.10.1.5]{Std:16m}: each
packet contains eight information bytes coded at rate $1/3$ and
modulated using QPSK. We compare the performance of the proposed
scheme against the following prominent baseline schemes. Since these
schemes were originally designed for single-antenna MSs, they do not
consider MS precoder and equalizer designs. We extend these schemes
to generate the \hbox{$k$-th} MS precoder matrix $\textbf{W}_{k}^{}$
from the principal right singular vectors of the channel matrix
$\textbf{H}_{R\,,\,k}^{}$ with equal power allocation across the
data streams, and we obtain the \hbox{$k$-th} MS equalizer matrix as
$\textbf{V}_{k}^{} = ( \textbf{W}_{k}^{} )_{}^{T}$.
\begin{list}{\labelitemi}{\leftmargin=0.5em}
\item \emph{\hbox{Baseline 1 (Bidirectional Channel Inversion Naive
Algorithm
\cite{Jnl:Multiuser_two_way_relaying_scheme_and_analysis:Ding}):}}
The BS precoder and equalizer matrices and the RS transformation
matrix are determined using pseudo-inverse methods.

\item \emph{\hbox{Baseline 2 (Bidirectional Channel Inversion Greedy
Algorithm \cite{Cnf:Two_way_one_BS_multi_MS_AF:Sun}):}} The BS
precoder and equalizer matrices are determined using pseudo-inverse
methods. A greedy iterative algorithm is employed to determine the
RS transformation matrix that maximizes the \emph{asymptotic} per
stream SINRs.

\item \emph{\hbox{Baseline 3 (Two-Way Relaying using Conventional
SDMA Processing
\cite{Jnl:Multiuser_two_way_relaying_no_SI_cancellation:Sayed}):}}
The RS transformation matrix is devised to spatially multiplex
\emph{all} data streams. Since this scheme does not provide BS
precoder and equalizer designs, we generate the BS precoder matrix
$\textbf{W}_{\!B}^{}$ from the principal right singular vectors of
the channel matrix $\textbf{H}_{R\,,B}^{}$ with equal power
allocation across the data streams, and we obtain the BS equalizer
matrix as $\textbf{V}_{\!B}^{} = ( \textbf{W}_{\!B}^{} )_{}^{T}$.
\end{list}

In the simulation results we define the signal-to-noise ratio (SNR)
as $P_{B\,}^{} / N_{\,0}^{}$. We set the RS and MS transmit powers
such that $P_{B\,}^{} / L = P_{R\,}^{} / L = P_{\,k\,}^{} /
L_{\,k}^{}$, so the transmit power per data stream is the same for
all nodes. We assume i.i.d. Rayleigh fading, so the channel matrices
are given by\footnote{Note that $\textrm{rank} (
\textbf{H}_{R\,,B}^{} ) = \min\{ N_{B}^{}, N_{R}^{} \}$ and
$\textrm{rank} ( \textbf{H}_{R\,,k}^{} ) = \min\{ N_{\,k}^{},
N_{R}^{} \}$ with probability 1.} $\textrm{vec} (
\textbf{H}_{R\,,B}^{} ) \sim \mathcal{CN}\big(
\textbf{0}_{N_{B}^{}N_{R}^{} \times 1}^{},
\textbf{I}_{N_{B}^{}N_{R}^{}}^{} \big)$ and $\textrm{vec} (
\textbf{H}_{R\,,\,k}^{} ) \sim \mathcal{CN}\big(
\textbf{0}_{N_{k}^{}N_{R}^{} \times 1}^{},
\textbf{I}_{N_{k}^{}N_{R}^{}}^{} \big)$.

\subsection{Performance Comparisons} In
\figurename~\ref{Fig:Per} and \figurename~\ref{Fig:SumRateAll}, we
present the performance results when the RS is equipped with
$N_{R}^{} = 4$ antennas. Note that in this setting the number of
spatial dimensions at the RS does not suffice for performing two-way
relaying using conventional SDMA processing (i.e., $N_{R}^{} < 2L$),
and so Baseline~3 is not feasible. Moreover, we assume that
\hbox{User 2} has higher service priority than the other users; as
an example, we set the priority weight factors to $[\,
\boldsymbol{\omega}_{U}^{(2)} \,]_{(1)}^{} = [\,
\boldsymbol{\omega}_{\!D}^{(2)} \,]_{(1)}^{} = 2$ and $[\,
\boldsymbol{\omega}_{U}^{(k)} \,]_{(l)}^{} = [\,
\boldsymbol{\omega}_{\!D}^{(k)} \,]_{(l)}^{} = 1$ otherwise.

First, in \figurename~\ref{Fig:Per} we show the PER performance
results when the BS is equipped with $N_{B}^{} = 4$ antennas and the
MSs are equipped with $N_{\,k}^{} = 2$ antennas. It can be seen that
the proposed scheme exhibits better error performance than the
baseline schemes. For instance, the proposed scheme achieves in
excess of 10~dB SNR gain over the baseline schemes at $10_{}^{-2}$
PER. This is attributed to the fact that the proposed scheme
efficiently exploits the multiple spatial dimensions at the MSs,
whereas the baseline schemes were originally designed for single
antenna MSs and cannot efficiently exploit the available spatial
dimensions. On the other hand, reflecting the QoS priority settings,
for the proposed scheme \hbox{User 2} has approximately 3~dB SNR
gain over the other users for all PER values smaller than
$10_{}^{-1}$.

Second, in \figurename~\ref{Fig:SumRateAll} we show the average sum
rate performance results. In \figurename~\ref{Fig:SumRate}, we show
the average data rate versus SNR when the BS is equipped with
$N_{B}^{} = 4$ antennas. It can be seen that the proposed scheme
achieves significant data rate gain over the baseline schemes.
Moreover, the proposed scheme alleviates the half-duplex loss (cf.
Remark~\ref{Remark:RelayChannelDoF} and
Corollary~\ref{Corollary:RelayChannelDoFIid}) and achieves the DoF
equal to $\min\{ N_{B}^{}, N_{R}^{}, \sum_{k = 1}^{K} \! N_{\,k}^{}
\} = 4$. In \figurename~\ref{Fig:SumRateVsNb}, we show the average
sum rate versus the number of BS antennas at 25~dB SNR. It can be
seen that the data rate of the proposed scheme improves
monotonically with the number of antennas at the BS and MSs. Note
that the inferior performance of Baseline~3 is due to the fact this
scheme requires more spatial dimensions at the RS to be feasible.

Finally, in \figurename~\ref{Fig:PerSdma} we show the PER
performance results when the BS is equipped with $N_{B}^{} = 4$
antennas, the MSs are equipped with $N_{\,k}^{} = 2$ antennas, and
the RS is equipped with $N_{R}^{} = 8$ antennas. In this setting,
there are enough spatial dimensions at the RS (i.e., $N_{R}^{} \ge
2L$) for Baseline~3 to be feasible\footnote{Note that Baseline~1 is
infeasible as it requires $N_{B}^{} \ge N_{R}^{}$. Therefore, we
have excluded it from the comparison.}. For simplicity of
comparison, we assume all the users have the same service priority
and we set all priority weight factors to $[\,
\boldsymbol{\omega}_{U}^{(k)} \,]_{(l)}^{} = [\,
\boldsymbol{\omega}_{\!D}^{(k)} \,]_{(l)}^{} = 1$. It can be seen
that the proposed scheme substantially outperforms Baseline~3 (e.g.,
up to 19~dB SNR gain at $10_{}^{-2}$ PER). This is because the
proposed scheme can efficiently exploit the spatial dimensions at
the RS to mitigate interference \emph{as well as} to achieve
beamforming gain, whereas Baseline~3 uses all the spatial dimensions
to null interference.

\section{Conclusions} \label{Sec:Conclusions} In cellular multi-user
two-way AF relaying systems, each node experiences self-induced
backward propagated interference as well as multi-user interference.
As a result, conventional self-interference cancelation approaches
for single-user two-way relay systems do not suffice to mitigate the
impact of interference. We applied an interference management model
exploiting signal space alignment and proposed a linear MIMO
transceiver design algorithm, which allows for alleviating the
half-duplex loss and providing flexible performance optimization
accounting for each user's QoS priorities. Numerical comparisons to
two-way relaying schemes based on bidirectional channel inversion
and SDMA-only processing show that the proposed scheme achieves
superior error rate and average data rate performance.

\section*{Appendix~A: Proof of Lemma~\ref{Lemma:Decomposition}}
We first show that the end-to-end per stream SINRs are predicated by
the first hop. Specifically, the end-to-end SINRs of the data stream
estimates \eqref{Eqn:Sinr2Ul}-\eqref{Eqn:Sinr2Dl} can be expressed
as
\begin{IEEEeqnarray*}{l}
\textrm{UL: }\gamma_{\,U}^{(k,\,l)} = \xi_{\,U}^{(k,\,l)} \,
\widetilde{\gamma}_{\,U}^{(k,\,l)}, \;\; \textrm{DL:
}\gamma_{D}^{(k,\,l)} = \xi_{D}^{(k,\,l)} \,
\widetilde{\gamma}_{D}^{(k,\,l)},\IEEEyesnumber\label{Eqn:Sinr1st2nd}
\end{IEEEeqnarray*}
where $\widetilde{\gamma}_{\,U}^{(k,\,l)}$ and
$\widetilde{\gamma}_{D}^{(k,\,l)}$ are the SINRs of the RS forwarded
signals \eqref{Eqn:Sinr1st1} and
\begin{IEEEeqnarray*}{l}
\xi_{\,U}^{(k,\,l)} \triangleq
\frac{\substack{\textstyle N_{\,0}^{} |\, [\,
\textbf{V}_{\!B}^{\,(k)} \,]_{(l,\,:)}^{} ( \textbf{H}_{R\,,B}^{}
)_{}^{T} [\, \textbf{F}_{\!R}^{\,(k)} \,]_{(:,\,l)}^{}
\,|_{}^{\,2\,} ||\, [\, \textbf{A}_{R}^{\!(k)} \,]_{(l,\,:)}^{}
\,||_{}^{\,2\,}}}{\left(\substack{\textstyle\sum_{m = 1}^{K}
\sum_{\substack{\,q = 1\\q \neq \,l}}^{\,L_{m}^{}} \! |\, [\,
\textbf{V}_{\!B}^{\,(k)} \,]_{(l,\,:)}^{} ( \textbf{H}_{R\,,B}^{}
)_{}^{T} [\, \textbf{F}_{\!R}^{\,(m)} \,]_{(:,\,q)}^{} [\,
\boldsymbol{\phi}_{}^{\,(m)} \,]_{(q)}^{} \,|_{}^{\,2\,}\\
\textstyle+ N_{\,0}^{} ( ||\, [\, \textbf{V}_{\!B}^{\,(k)}
\,]_{(l,\,:)}^{} ( \textbf{H}_{R\,,B}^{} )_{}^{T\,}
\textbf{F}_{\!R\,}^{} \textbf{A}_{R}^{} \,||_{}^{\,2\,} \!+ ||\, [\,
\textbf{V}_{\!B}^{\,(k)} \,]_{(l,\,:)}^{} \,||_{}^{\,2\,}
)}\right)},\IEEEyessubnumber\label{Eqn:Sinr1st2ndUlWeight}\\
\xi_{D}^{(k,\,l)} \triangleq \frac{\substack{\textstyle N_{\,0}^{}
|\, [\, \textbf{V}_{k}^{} \,]_{(l,\,:)}^{} ( \textbf{H}_{R\,,\,k}^{}
)_{}^{T} [\, \textbf{F}_{\!R}^{\,(k)} \,]_{(:,\,l)}^{}
\,|_{}^{\,2\,} ||\, [\, \textbf{A}_{R}^{\!(k)} \,]_{(l,\,:)}^{}
\,||_{}^{\,2\,}}}{\left(\substack{\textstyle\sum_{m = 1}^{K}
\sum_{\substack{\,q = 1\\q \neq \,l}}^{\,L_{m}^{}} \! |\, [\,
\textbf{V}_{k}^{} \,]_{(l,\,:)}^{} ( \textbf{H}_{R\,,\,k}^{}
)_{}^{T} [\, \textbf{F}_{\!R}^{\,(m)} \,]_{(:,\,q)}^{} [\,
\boldsymbol{\psi}_{}^{\,(m)} \,]_{(q)}^{}\,|_{}^{\,2\,}\\
\textstyle+ \sum_{\substack{a = 1\\a \neq \,k}}^{K} \sum_{\,b =
1}^{\,L_{a}^{}} \! |\, [\, \textbf{V}_{k}^{} \,]_{(l,\,:)}^{} (
\textbf{H}_{R\,,\,k}^{} )_{}^{T} [\, \textbf{F}_{\!R}^{\,(a)}
\,]_{(:,\,b)}^{} [\, \boldsymbol{\phi}_{}^{\,(a)} \,]_{(b)}^{}
\,|_{}^{\,2\,}\\
\textstyle+ N_{\,0}^{} ( ||\, [\, \textbf{V}_{k}^{} \,]_{(l,\,:)}^{}
( \textbf{H}_{R\,,\,k}^{} )_{}^{T\,} \textbf{F}_{\!R\,}^{}
\textbf{A}_{R}^{} \,||_{}^{\,2\,} \!+ ||\, [\, \textbf{V}_{k}^{}
\,]_{(l,\,:)}^{}
\,||_{}^{\,2\,})}\right)}.\IEEEyessubnumber\label{Eqn:Sinr1st2ndDlWeight}
\end{IEEEeqnarray*}
Consider $\xi_{\,U}^{(k,\,l)}$ for example. Note that
\begin{IEEEeqnarray*}{l}
||\, [\, \textbf{V}_{\!B}^{\,(k)} \,]_{(l,\,:)}^{} (
\textbf{H}_{R\,,B}^{} )_{}^{T\,} \textbf{F}_{\!R\,}^{}
\textbf{A}_{R}^{} \,||_{}^{\,2\,}\\
\textstyle= ||\, \sum_{m = 1}^{K} \sum_{\substack{\,q = 1\\q \neq
\,l}}^{\,L_{m}^{}} [\, \textbf{V}_{\!B}^{\,(k)} \,]_{(l,\,:)}^{} (
\textbf{H}_{R\,,B}^{} )_{}^{T} [\, \textbf{F}_{\!R}^{\,(m)}
\,]_{(:,\,q)}^{} [\, \textbf{A}_{R}^{\!(m)} \,]_{(q,\,:)}^{} + [\,
\textbf{V}_{\!B}^{\,(k)} \,]_{(l,\,:)}^{} ( \textbf{H}_{R\,,B}^{}
)_{}^{T} [\, \textbf{F}_{\!R}^{\,(k)} \,]_{(:,\,l)}^{} [\,
\textbf{A}_{R}^{\!(k)} \,]_{(l,\,:)}^{} \,||_{}^{\,2\,}\\
\textstyle\stackrel{(a)}{\approx} ||\, [\, \textbf{V}_{\!B}^{\,(k)}
\,]_{(l,\,:)}^{} ( \textbf{H}_{R\,,B}^{} )_{}^{T} [\,
\textbf{F}_{\!R}^{\,(k)} \,]_{(:,\,l)}^{} [\, \textbf{A}_{R}^{\!(k)}
\,]_{(l,\,:)}^{} \,||_{}^{\,2\,},
\end{IEEEeqnarray*}
where (a) follows from the fact that the terms $\big\{ [\,
\textbf{V}_{\!B}^{\,(k)} \,]_{(l,\,:)}^{} ( \textbf{H}_{R\,,B}^{}
)_{}^{T} [\, \textbf{F}_{\!R}^{\,(m)} \,]_{(:,\,q)}^{} \big\}$
should be \emph{negligible} to suppress interference. It can be
deduced that $\xi_{\,U}^{(k,\,l)} \leq 1$ and similarly
$\xi_{D}^{(k,\,l)} \leq 1$, so the end-to-end per stream SINRs are
limited by the SINRs of the RS forwarded signals, i.e.,
\begin{IEEEeqnarray*}{l}
\textrm{UL: }\gamma_{\,U}^{(k,\,l)} \leq
\widetilde{\gamma}_{\,U}^{(k,\,l)}, \;\; \textrm{DL:
}\gamma_{D}^{(k,\,l)} \leq
\widetilde{\gamma}_{D}^{(k,\,l)}.\IEEEyesnumber\label{Eqn:Sinr1st2nd2}
\end{IEEEeqnarray*}

As per \eqref{Eqn:Sinr1st2nd2}, the minimum weighted end-to-end SINR
of the data stream estimates is limited by the minimum weighted SINR
of the RS forwarded signals, i.e.,
\begin{IEEEeqnarray*}{l}
\!\!\!\!\!\!\displaystyle\min_{\substack{\forall\, k \,\in\,
\mathcal{K}\\\forall\, l \,\in\, \mathcal{L}_{k}^{}}}
\!\!\textstyle\big\{ \big( [\, \boldsymbol{\omega}_{U}^{(k)}
\,]_{(l)}^{} \big)_{}^{-1} \! \gamma_{\,U}^{(k,\,l)}, \big( [\,
\boldsymbol{\omega}_{\!D}^{(k)} \,]_{(l)}^{} \big)_{}^{-1} \!
\gamma_{D}^{(k,\,l)} \, \big\} \leq
\displaystyle\min_{\substack{\forall\, k \,\in\,
\mathcal{K}\\\forall\, l \,\in\, \mathcal{L}_{k}^{}}}
\!\!\textstyle\big\{ \big( [\, \boldsymbol{\omega}_{U}^{(k)}
\,]_{(l)}^{} \big)_{}^{-1} \widetilde{\gamma}_{\,U}^{(k,\,l)}, \big(
[\, \boldsymbol{\omega}_{\!D}^{(k)} \,]_{(l)}^{} \big)_{}^{-1}
\widetilde{\gamma}_{D}^{(k,\,l)} \,
\big\}.\;\;\;\;\;\;\;\IEEEyesnumber\label{Eqn:ProbMainObjectiveSplit}
\end{IEEEeqnarray*}
By this property, the transceiver design problem,
Problem~$\mathcal{Q}$, can be decomposed into two stages. In the
first stage processing, we find the BS and MS precoder matrices and
the RS equalizer matrix that maximize the minimum weighted SINR of
the RS forwarded signals, i.e.,
$\underset{\substack{\textbf{W}_{\!B}^{}, \{\:\! \textbf{W}_{\!k}^{}
\}_{k=1\!}^{K},\, \textbf{A}_{R}^{}}}{\arg \max}
\displaystyle\min_{\substack{\forall\, k \,\in\,
\mathcal{K}\\\forall\, l \,\in\, \mathcal{L}_{k}^{}}}
\!\!\textstyle\big\{ \big( [\, \boldsymbol{\omega}_{U}^{(k)}
\,]_{(l)}^{} \big)_{}^{-1} \widetilde{\gamma}_{\,U}^{(k,\,l)}, \big(
[\, \boldsymbol{\omega}_{\!D}^{(k)} \,]_{(l)}^{} \big)_{}^{-1}
\widetilde{\gamma}_{D}^{(k,\,l)} \, \big\}$, and thereby implicitly
maximize the \emph{achievable} minimum weighted end-to-end SINR of
the data streams estimates. Then, in the second stage processing, we
find the RS precoder matrix and the BS and MS equalizer matrices to
\emph{holistically} maximize the minimum weighted end-to-end SINR of
the data streams estimates, i.e.,
$\underset{\substack{\textbf{F}_{\!R}^{}, \textbf{V}_{\!B}^{}, \{
\textbf{V}_{k}^{} \}_{k=1}^{K}}}{\arg \max}
\displaystyle\min_{\substack{\forall\, k \,\in\,
\mathcal{K}\\\forall\, l \,\in\, \mathcal{L}_{k}^{}}}
\!\!\textstyle\big\{ \big( [\, \boldsymbol{\omega}_{U}^{(k)}
\,]_{(l)}^{} \big)_{}^{-1} \! \gamma_{\,U}^{(k,\,l)}, \big( [\,
\boldsymbol{\omega}_{\!D}^{(k)} \,]_{(l)}^{} \big)_{}^{-1} \!
\gamma_{D}^{(k,\,l)} \, \big\}$.

\section*{Appendix~B: Proof of Lemma~\ref{Lemma:BeamSelPowerAllocation}}
Substituting \eqref{Eqn:Sinr1st2} into
\eqref{Eqn:ProbFstStageObjective}, the problem of maximizing the
minimum weighted SINR of the RS forwarded signals can be
reformulated as
\begin{IEEEeqnarray*}{l}
\!\!\!\!\underset{\substack{\{ \textbf{g}_{k}^{(l)} \}, \{
\textbf{g}_{B}^{(k,\,l)} \}, \{ \lambda_{k}^{(l)} \}, \{
\lambda_{B}^{(k,\,l)} \}}}{\arg \max} \min_{\substack{\forall\, k
\,\in\, \mathcal{K}\\\forall\, l \,\in\, \mathcal{L}_{k}^{}}} \big\{
\big( [\, \boldsymbol{\omega}_{U}^{(k)} \,]_{(l)}^{} \big)_{}^{-1}
\widetilde{\gamma}_{\,U}^{(k,\,l)}, \big( [\,
\boldsymbol{\omega}_{\!D}^{(k)} \,]_{(l)}^{} \big)_{}^{-1}
\widetilde{\gamma}_{D}^{(k,\,l)} \,
\big\}\IEEEyesnumber\label{Eqn:PowerAllocationReformulate}\\
\!\!\!\!\Leftrightarrow\!\!\! \underset{\{ \textbf{g}_{k}^{(l)} \},
\{ \textbf{g}_{B}^{(k,\,l)} \}}{\arg \max} \!\! \min \!\Bigg\{\!\!
\min_{\forall\, k \,\in\, \mathcal{K}} \!\Bigg\{
\underset{\lambda_{k}^{(l)}}{\arg \max} \displaystyle\min_{\forall\,
l \,\in\, \mathcal{L}_{k}^{}} \!\! \big\{ \big( [\,
\boldsymbol{\omega}_{U}^{(k)} \,]_{(l)}^{} \big)_{}^{-1}
\widetilde{\gamma}_{\,U}^{(k,\,l)} \, \big\} \!\Bigg\}, \underset{\{
\lambda_{B}^{(k,\,l)} \}}{\arg \max} \min_{\substack{\forall\, k
\,\in\, \mathcal{K}\\\forall\, l \,\in\, \mathcal{L}_{k}^{}}} \big\{
\big( [\, \boldsymbol{\omega}_{\!D}^{(k)} \,]_{(l)}^{} \big)_{}^{-1}
\widetilde{\gamma}_{D}^{(k,\,l)} \, \big\} \!\Bigg\}.
\end{IEEEeqnarray*}
Therefore, for fixed beam directions $\{\,
\textbf{g}_{\,k}^{\,(l)}\!,\, \textbf{g}_{B}^{\,(k,\,l)} \,\}$, the
power allocation at each node can be \emph{separately} determined.
For instance, the power allocation at the \hbox{$k$-th MS} can be
determined according to $\big( \lambda_{\,k}^{(l)}
\big)_{}^{\!\star} = \arg \max_{\lambda_{k}^{(l)}}
\displaystyle\min_{\forall\, l \,\in\, \mathcal{L}_{k}^{}} \!\!
\big\{ \big( [\, \boldsymbol{\omega}_{U}^{(k)} \,]_{(l)}^{}
\big)_{}^{-1} \widetilde{\gamma}_{\,U}^{(k,\,l)} \, \big\}$, whereby
\begin{IEEEeqnarray*}{l}
\big( [\, \boldsymbol{\omega}_{U}^{(k)} \,]_{(1)}^{} \big)_{}^{-1}
\widetilde{\gamma}_{\,U}^{(k,\,1)} = \ldots = \big( [\,
\boldsymbol{\omega}_{U}^{(k)} \,]_{(L_{k}^{})}^{} \big)_{}^{-1}
\widetilde{\gamma}_{\,U}^{(k,\,L_{k}^{})}\IEEEyesnumber\label{Eqn:PowerAllocationMsCondition}\\
\Leftrightarrow \big( [\, \boldsymbol{\omega}_{U}^{(k)} \,]_{(1)}^{}
\big)_{}^{-1} \kappa_{\,U}^{(k,\,1)} \big( \lambda_{\,k}^{(1)}
\big)_{}^{\!\star} = \ldots = \big( [\,
\boldsymbol{\omega}_{U}^{(k)} \,]_{(L_{k}^{})}^{} \big)_{}^{-1}
\kappa_{\,U}^{(k,\,L_{k}^{})} \big( \lambda_{\,k}^{(L_{k}^{})}
\big)_{}^{\!\star}.
\end{IEEEeqnarray*}
It can be shown that \eqref{Eqn:PowerAllocationMsCondition} is
satisfied with $\big( \lambda_{\,k}^{(l)} \big)_{}^{\!\star} \!=\!
\frac{\textstyle[\, \boldsymbol{\omega}_{U}^{(k)} \,]_{(l)}^{} \big(
\kappa_{\,U}^{(k,\,l)} \big)_{}^{-1} P_{\,k}^{}}{\textstyle\sum_{\,q
= 1}^{\,L_{k}^{}} \, [\, \boldsymbol{\omega}_{U}^{(k)} \,]_{(q)}^{}
\big( \kappa_{\,U}^{(k,\,q)} \big)_{}^{-1}}$, which in turn yields
weighted UL SINRs of $\big( [\, \boldsymbol{\omega}_{U}^{(k)}
\,]_{(l)}^{} \big)_{}^{-1} \widetilde{\gamma}_{\,U}^{(k,\,l)} =
\frac{\textstyle P_{\,k}^{}}{\textstyle\sum_{\,q = 1}^{\,L_{k}^{}}
\, [\, \boldsymbol{\omega}_{U}^{(k)} \,]_{(q)}^{} \big(
\kappa_{\,U}^{(k,\,q)} \big)_{}^{-1}}$. Analogously, the power
allocation at the BS is given by $\big( \lambda_{B}^{(k,\,l)}
\big)_{}^{\!\star} \!=\! \frac{\textstyle[\,
\boldsymbol{\omega}_{\!D}^{(k)} \,]_{(l)}^{} \big(
\kappa_{D}^{(k,\,l)} \big)_{}^{-1} P_{B}^{}}{\textstyle\sum_{m =
1}^{K} \sum_{\,q = 1}^{\,L_{m}^{}} \, [\,
\boldsymbol{\omega}_{\!D}^{(m)} \,]_{(q)}^{} \big(
\kappa_{D}^{(m,\,q)} \big)_{}^{-1}}$, which yields weighted DL SINRs
of $\big( [\, \boldsymbol{\omega}_{\!D}^{(k)} \,]_{(l)}^{}
\big)_{}^{-1} \widetilde{\gamma}_{D}^{(k,\,l)} = \frac{\textstyle
P_{B}^{}}{\textstyle\sum_{m = 1}^{K} \sum_{\,q = 1}^{\,L_{m}^{}} \,
[\, \boldsymbol{\omega}_{\!D}^{(m)} \,]_{(q)}^{} \big(
\kappa_{D}^{(m,\,q)} \big)_{}^{-1}}$.

\section*{Appendix~C: Convergence of the Second Stage Processing}
At the \hbox{$q$-th} iteration of the second stage processing, we
denote the RS precoder matrix as $\textbf{F}_{\!R\,}^{}[\, q \,]$,
the BS equalizer matrix as $\textbf{V}_{\!B\,}^{}[\, q \,]$, the
\hbox{$k$-th} MS equalizer matrix as $\textbf{V}_{k\,}^{}[\, q \,]$,
and the minimum weighted per stream SINR as $\gamma[\, q \,]$.
Moreover, we denote as $\gamma_{\,U}^{(k,\,l)}\big\{
\textbf{F}_{\!R\,}^{}[\, a \,],\! \textbf{V}_{\!B\,}^{}[\, b \,]
\big\}$ and $\gamma_{D}^{(k,\,l)}\big\{ \textbf{F}_{\!R\,}^{}[\, a
\,],\! \textbf{V}_{k\,}^{}[\, b \,] \big\}$ the UL and DL per stream
SINR given $\textbf{F}_{\!R\,}^{}[\, a \,]$,
$\textbf{V}_{\!B\,}^{}[\, b \,]$, and $\textbf{V}_{k\,}^{}[\, b
\,]$. We show that each iteration of the second stage processing
monotonically increases the minimum weighted per stream SINR.

In Step~2.1, given the BS and MS equalizer matrices $\big\{
\textbf{V}_{\!B\,}^{}[\, q-1 \,], \{ \textbf{V}_{k\,}^{}[\, q-1 \,]
\}_{k=1}^{K} \big\}$, we solve for the RS precoder matrix
$\textbf{F}_{\!R\,}^{}[\, q \,]$ to improve the minimum weighted per
stream SINR, i.e.,
\begin{IEEEeqnarray*}{l}
\gamma_{\,0}^{} \!=\! \min_{\substack{\forall\, k \,\in\,
\mathcal{K}\\\forall\, l \,\in\, \mathcal{L}_{k}^{}}}
\textstyle\big\{ \big( [\, \boldsymbol{\omega}_{U}^{(k)}
\,]_{(l)}^{} \big)_{}^{-1} \gamma_{\,U}^{(k,\,l)}\big\{
\textbf{F}_{\!R\,}^{}[\, q \,],\! \textbf{V}_{\!B\,}^{}[\, q\!-\!1
\,] \big\}, \big( [\, \boldsymbol{\omega}_{\!D}^{(k)} \,]_{(l)}^{}
\big)_{}^{-1} \gamma_{D}^{(k,\,l)}\big\{ \textbf{F}_{\!R\,}^{}[\, q
\,],\! \textbf{V}_{k\,}^{}[\, q\!-\!1 \,] \big\}
\big\}\IEEEyesnumber\label{Eqn:Convergence2_1}\\
\ge \min_{\substack{\forall\, k \,\in\, \mathcal{K}\\\forall\, l
\,\in\, \mathcal{L}_{k}^{}}} \textstyle\big\{ \big( [\,
\boldsymbol{\omega}_{U}^{(k)} \,]_{(l)}^{} \big)_{}^{-1}
\gamma_{\,U}^{(k,\,l)}\big\{ \textbf{F}_{\!R\,}^{}[\, q\!-\!1 \,],\!
\textbf{V}_{\!B\,}^{}[\, q\!-\!1 \,] \big\}, \big( [\,
\boldsymbol{\omega}_{\!D}^{(k)} \,]_{(l)}^{} \big)_{}^{-1}
\gamma_{D}^{(k,\,l)}\big\{ \textbf{F}_{\!R\,}^{}[\, q\!-\!1 \,],\!
\textbf{V}_{k\,}^{}[\, q\!-\!1 \,] \big\} \big\}\;\;\\
= \gamma[\, q-1 \,].
\end{IEEEeqnarray*}

In Step~2.2 and Step~2.3, given the RS precoder matrix
$\textbf{F}_{\!R\,}^{}[\, q \,]$, we solve for the BS and MS
equalizer matrices $\big\{ \textbf{V}_{\!B\,}^{}[\, q \,], \{
\textbf{V}_{k\,}^{}[\, q \,] \}_{k=1}^{K} \big\}$ to improve the per
stream SINRs, i.e.,
\begin{IEEEeqnarray*}{l}
\gamma_{\,U}^{(k,\,l)}\big\{ \textbf{F}_{\!R\,}^{}[\, q \,],\!
\textbf{V}_{\!B\,}^{}[\, q \,] \big\} \!\!\ge\!
\gamma_{\,U}^{(k,\,l)}\big\{ \textbf{F}_{\!R\,}^{}[\, q \,],\!
\textbf{V}_{\!B\,}^{}[\, q\!-\!1 \,] \big\}, \;
\gamma_{D}^{(k,\,l)}\big\{ \textbf{F}_{\!R\,}^{}[\, q \,],\!
\textbf{V}_{k\,}^{}[\, q \,] \big\} \!\!\ge\!
\gamma_{D}^{(k,\,l)}\big\{ \textbf{F}_{\!R\,}^{}[\, q \,],\!
\textbf{V}_{k\,}^{}[\, q\!-\!1 \,] \big\};
\end{IEEEeqnarray*}
hence the minimum weighted per stream SINR satisfies
\begin{IEEEeqnarray*}{l}
\gamma[\, q \,] \!=\! \min_{\substack{\forall\, k \,\in\,
\mathcal{K}\\\forall\, l \,\in\, \mathcal{L}_{k}^{}}}
\textstyle\big\{ \big( [\, \boldsymbol{\omega}_{U}^{(k)}
\,]_{(l)}^{} \big)_{}^{-1} \gamma_{\,U}^{(k,\,l)}\big\{
\textbf{F}_{\!R\,}^{}[\, q \,],\! \textbf{V}_{\!B\,}^{}[\, q \,]
\big\}, \big( [\, \boldsymbol{\omega}_{\!D}^{(k)} \,]_{(l)}^{}
\big)_{}^{-1} \gamma_{D}^{(k,\,l)}\big\{ \textbf{F}_{\!R\,}^{}[\, q
\,],\! \textbf{V}_{k\,}^{}[\, q \,] \big\} \big\} \!\ge\!
\gamma_{\,0}^{}.\;\;\;\;\;\;\IEEEyesnumber\label{Eqn:Convergence2_2}
\end{IEEEeqnarray*}

It follows from \eqref{Eqn:Convergence2_1} and
\eqref{Eqn:Convergence2_2} that the weighted minimum per stream SINR
increases with each iteration, i.e., $\gamma[\, q \,] \ge
\gamma_{\,0}^{} \ge \gamma[\, q-1 \,]$, and the second stage
processing must converge.

\section*{Appendix~D: Derivation of Algorithm~\ref{Algorithm:RsPrecoder}}
For ease of exposition, we express the constraints in
Problem~$\mathcal{B}_{R}^{}$ in vector form. The transmit power
constraint \eqref{Eqn:ProbSndStagePrecoderStdProbPowerConstraint}
can be expressed as
\begin{IEEEeqnarray*}{l}
\sqrt{P_{R}^{}} \ge \sqrt{ ||\, \textbf{F}_{\!R\,}^{}
\boldsymbol{\Phi} ||^{\,2\,} + ||\, \textbf{F}_{\!R\,}^{}
\boldsymbol{\Psi} ||^{\,2\,} + N_{\,0}^{} ||\, \textbf{F}_{\!R\,}^{}
\textbf{A}_{R}^{} \,||^{\,2\,} } \stackrel{(a)}{=} ||\,
\boldsymbol{\rho}
\,||,\IEEEyesnumber\label{Eqn:ProbSndStagePrecoderStdProbPowerConstraintDerive1}
\end{IEEEeqnarray*}
\begin{IEEEeqnarray*}{l}
\boldsymbol{\rho} \triangleq \Big( \big( \big(\, \boldsymbol{\Phi} (
\boldsymbol{\Phi} )_{}^{\:\!\!\dag\,} \!+\! \boldsymbol{\Psi} (
\boldsymbol{\Psi} )_{}^{\:\!\!\dag\,} \!+\! N_{\,0}^{}\,
\textbf{A}_{R\,}^{} ( \textbf{A}_{R\,}^{} )_{}^{\:\!\!\dag}
\,\big)_{}^{\!1\!/\:\!2\,} \big)_{}^{T} \!\otimes
\textbf{I}_{N_{R}^{}}^{} \Big) \, \textrm{vec} (\,
\textbf{F}_{\!R}^{} \,)
\IEEEyesnumber\label{Eqn:ProbSndStagePrecoderStdProbPowerConstraintsDerive}
\end{IEEEeqnarray*}
and equality (a) follows from the Kronecker product property
$\textrm{vec} ( \textbf{X} \textbf{Y} \textbf{Z} ) = ( ( \textbf{Z}
)_{}^{T} \!\otimes\! \textbf{X} ) \, \textrm{vec} ( \textbf{Y} )$.
With some algebraic manipulations and using the aforementioned
Kronecker product property, the UL SINR constraints
\eqref{Eqn:ProbSndStagePrecoderStdProbSinrConstraints} can be
expressed as
\begin{IEEEeqnarray*}{l}
\big( [\, \boldsymbol{\omega}_{U}^{(k)} \,]_{(l)}^{} \big)_{}^{-1}
\gamma_{\,U}^{(k,\,l)} \ge \gamma_{\,0}^{} \Leftrightarrow
\alpha_{\,U}^{\,(k,\,l)} \geq \big|\big| \big[\,
\boldsymbol{\beta}_{\,U}^{\,(k,\,l)} ; \delta_{\,U}^{\,(k,\,l)}
\big]
\big|\big|,\IEEEyesnumber\label{Eqn:ProbSndStagePrecoderStdProbSinrConstraintsUlDerive1}
\end{IEEEeqnarray*}
\begin{IEEEeqnarray*}{l}
\!\!\!\!\!\!\!\!\alpha_{\,U}^{\,(k,\,l)} \triangleq \sqrt{1 \!+\!
\big( [\, \boldsymbol{\omega}_{U}^{(k)} \,]_{(l)}^{} \gamma_{\,0}^{}
\big)_{}^{-1}} \, |\, [\, \boldsymbol{\phi}_{}^{\,(k)}
\,]_{(l)}^{}\,| \, |\, [\, \textbf{V}_{\!B}^{\,(k)} \,]_{(l,\,:)}^{}
( \textbf{H}_{R\,,B}^{} )_{}^{T} [\, \textbf{F}_{\!R}^{\,(k)}
\,]_{(:,\,l)}^{} \,|,\; \delta_{\,U}^{\,(k,\,l)} \triangleq
\sqrt{N_{\,0}^{}}\, ||\, [\, \textbf{V}_{\!B}^{\,(k)}
\,]_{(l,\,:)}^{} \,||,\\
\!\!\!\!\!\!\!\!\boldsymbol{\beta}_{\,U}^{\,(k,\,l)} \triangleq
\Big( \big( \big(\, \boldsymbol{\Phi} ( \boldsymbol{\Phi}
)_{}^{\:\!\!\dag\,} \!+\! N_{\,0}^{}\, \textbf{A}_{R\,}^{} (
\textbf{A}_{R\,}^{} )_{}^{\:\!\!\dag} \,\big)_{}^{\!1\!/\:\!2\,}
\big)_{}^{T} \!\otimes [\, \textbf{V}_{\!B}^{\,(k)} \,]_{(l,\,:)}^{}
( \textbf{H}_{R\,,B}^{} )_{}^{T\,} \Big) \, \textrm{vec} (\,
\textbf{F}_{\!R}^{} \,).
\end{IEEEeqnarray*}
In the same manner, the DL SINR constraints can be expressed as
\begin{IEEEeqnarray*}{l}
\big( [\, \boldsymbol{\omega}_{\!D}^{(k)} \,]_{(l)}^{} \big)_{}^{-1}
\! \gamma_{D}^{(k,\,l)} \ge \gamma_{\,0}^{} \Leftrightarrow
\alpha_{D}^{\,(k,\,l)} \geq \big|\big| \big[\,
\boldsymbol{\beta}_{D}^{\,(k,\,l)} ; \delta_{D}^{\,(k,\,l)} \big]
\big|\big|,\IEEEyesnumber\label{Eqn:ProbSndStagePrecoderStdProbSinrConstraintsDlDerive1}
\end{IEEEeqnarray*}
\begin{IEEEeqnarray*}{l}
\!\!\!\!\!\!\!\!\alpha_{D}^{\,(k,\,l)} \triangleq \sqrt{1 \!+\!
\big( [\, \boldsymbol{\omega}_{\!D}^{(k)} \,]_{(l)}^{}
\gamma_{\,0}^{} \big)_{}^{-1}} \, |\, [\,
\boldsymbol{\psi}_{}^{\,(k)} \,]_{(l)}^{}\,| \, |\, [\,
\textbf{V}_{k}^{} \,]_{(l,\,:)}^{} ( \textbf{H}_{R\,,\,k}^{}
)_{}^{T} [\, \textbf{F}_{\!R}^{\,(k)} \,]_{(:,\,l)}^{} \,|,\;
\delta_{D}^{\,(k,\,l)} \triangleq \sqrt{N_{\,0}^{}}\, ||\, [\,
\textbf{V}_{k}^{} \,]_{(l,\,:)}^{} \,||,\\
\!\!\!\!\!\!\!\!\boldsymbol{\beta}_{D}^{\,(k,\,l)} \triangleq \Big(
\big( \big(\, \boldsymbol{\Psi} ( \boldsymbol{\Psi}
)_{}^{\:\!\!\dag\,} \!+\! \widetilde{\boldsymbol{\Phi}}_{k}^{} (
\widetilde{\boldsymbol{\Phi}}_{k}^{} )_{}^{\:\!\!\dag\,} \!+\!
N_{\,0}^{}\, \textbf{A}_{R\,}^{} ( \textbf{A}_{R\,}^{}
)_{}^{\:\!\!\dag} \,\big)_{}^{\!1\!/\:\!2\,} \big)_{}^{T} \!\otimes
[\, \textbf{V}_{k}^{} \,]_{(l,\,:)}^{} ( \textbf{H}_{R\,,\,k}^{}
)_{}^{T} \Big) \, \textrm{vec} (\, \textbf{F}_{\!R}^{} \,).
\end{IEEEeqnarray*}
Therefore, Problem~$\mathcal{B}_{R}^{}$ can be \emph{equivalently}
expressed as\\\indent\;\;$\big\{ \textbf{F}_{\!R}^{\,\star},
\gamma_{0}^{} \big\} \!:= \mathcal{B}_{R}^{}\big\{
\textbf{V}_{\!B}^{}, \{ \textbf{V}_{k}^{} \}_{k=1}^{K},
\textbf{W}_{\!B}^{\,\star}, \{\:\! \textbf{W}_{k}^{\,\star\,}
\}_{k=1}^{K}, \textbf{A}_{R}^{\!\star}, P_{R}^{}, \{
\boldsymbol{\omega}_{U}^{(k)}\!, \boldsymbol{\omega}_{\!D}^{(k)\,}
\}_{k=1}^{K} \big\}$\nopagebreak[4]
\begin{subnumcases}{\label{Eqn:ProbSndStagePrecoderStdProbAlt}}
\!\underset{\substack{\textbf{F}_{\!R\,}^{},\, \gamma_{0}^{}}}{\max}
\!\!\!\!\!\!\!\!\!\! & \!\!\!\!\!\!\!\!
$\gamma_{\,0}^{}$\label{Eqn:ProbSndStagePrecoderStdProbAltObjective}\\
\;\textrm{s.t.} & \!\!\!\!\!\!\!\! $\alpha_{\,U}^{\,(k,\,l)} \geq
\big|\big| \big[\, \boldsymbol{\beta}_{\,U}^{\,(k,\,l)} ;
\delta_{\,U}^{\,(k,\,l)} \big] \big|\big|, \;\;
\alpha_{D}^{\,(k,\,l)} \geq \big|\big| \big[\,
\boldsymbol{\beta}_{D}^{\,(k,\,l)} ; \delta_{D}^{\,(k,\,l)} \big]
\big|\big|, \;\;\forall k \in \mathcal{K}, \forall l \in
\mathcal{L}_{\,k}^{},$\;\;\label{Eqn:ProbSndStagePrecoderStdProbAltSinrConstraints}\\
& \!\!\!\!\!\!\!\! $\sqrt{P_{R}^{}} \ge ||\, \boldsymbol{\rho}
\,||.$\label{Eqn:ProbSndStagePrecoderStdProbAltPowerConstraint}
\end{subnumcases}

Note that the transmit power constraint
\eqref{Eqn:ProbSndStagePrecoderStdProbAltPowerConstraint} is convex
in the RS precoder matrix $\textbf{F}_{\!R}^{}$, but the SINR
constraints
\eqref{Eqn:ProbSndStagePrecoderStdProbAltSinrConstraints} are
non-convex in $\textbf{F}_{\!R}^{}$ and the minimum weighted per
stream SINR slack variable $\gamma_{\,0}^{}$ since
$\alpha_{\,U}^{\,(k,\,l)}$ and $\alpha_{D}^{\,(k,\,l)}$ are not
affine in $\textbf{F}_{\!R}^{}$ and $\gamma_{\,0}^{}$.

In order to obtain a mathematically tractable solution to
Problem~$\mathcal{B}_{R}^{}$, we cast the SINR constraints as convex
functions in $\textbf{F}_{\!R}^{}$ by \emph{tightening} these
constraints as follows. We define
\begin{IEEEeqnarray*}{l}
\widetilde{\alpha}_{\,U}^{\,(k,\,l)} \triangleq \sqrt{1 \!+\! \big(
[\, \boldsymbol{\omega}_{U}^{(k)} \,]_{(l)}^{} \gamma_{\,0}^{}
\big)_{}^{-1}} \, |\, [\, \boldsymbol{\phi}_{}^{\,(k)}
\,]_{(l)}^{}\,| \, \Re \big( [\, \textbf{V}_{\!B}^{\,(k)}
\,]_{(l,\,:)}^{} ( \textbf{H}_{R\,,B}^{} )_{}^{T} [\,
\textbf{F}_{\!R}^{\,(k)} \,]_{(:,\,l)}^{}
\big),\IEEEyesnumber\label{Eqn:ProbSndStagePrecoderStdProbSinrConstraintsUlDlDerive2aLbU}\\
\widetilde{\alpha}_{D}^{\,(k,\,l)} \triangleq \sqrt{1 \!+\! \big(
[\, \boldsymbol{\omega}_{\!D}^{(k)} \,]_{(l)}^{} \gamma_{\,0}^{}
\big)_{}^{-1}} \, |\, [\, \boldsymbol{\psi}_{}^{\,(k)}
\,]_{(l)}^{}\,| \, \Re \big( [\, \textbf{V}_{k}^{} \,]_{(l,\,:)}^{}
( \textbf{H}_{R\,,\,k}^{} )_{}^{T} [\, \textbf{F}_{\!R}^{\,(k)}
\,]_{(:,\,l)}^{}
\big),\IEEEyesnumber\label{Eqn:ProbSndStagePrecoderStdProbSinrConstraintsUlDlDerive2aLbD}
\end{IEEEeqnarray*}
which are affine in $\textbf{F}_{\!R}^{}$. Since $|\, [\,
\textbf{V}_{\!B}^{\,(k)} \,]_{(l,\,:)}^{} ( \textbf{H}_{R\,,B}^{}
)_{}^{T} [\, \textbf{F}_{\!R}^{\,(k)} \,]_{(:,\,l)}^{} \,| \ge \Re (
[\, \textbf{V}_{\!B}^{\,(k)} \,]_{(l,\,:)}^{} (
\textbf{H}_{R\,,B}^{} )_{}^{T} [\, \textbf{F}_{\!R}^{\,(k)}
\,]_{(:,\,l)}^{} )$ and $|\, [\, \textbf{V}_{k}^{} \,]_{(l,\,:)}^{}
( \textbf{H}_{R\,,\,k}^{} )_{}^{T} [\, \textbf{F}_{\!R}^{\,(k)}
\,]_{(:,\,l)}^{} \,| \ge \Re ( [\, \textbf{V}_{k}^{}
\,]_{(l,\,:)}^{} ( \textbf{H}_{R\,,\,k}^{} )_{}^{T} [\,
\textbf{F}_{\!R}^{\,(k)} \,]_{(:,\,l)}^{} )$,
$\alpha_{\,U}^{\,(k,\,l)}$ is lower bounded by
$\widetilde{\alpha}_{\,U}^{\,(k,\,l)}$ and $\alpha_{D}^{\,(k,\,l)}$
is lower bounded by $\widetilde{\alpha}_{D}^{\,(k,\,l)}$. If we
\emph{tighten} the SINR constraints as
\begin{IEEEeqnarray*}{l}
\widetilde{\alpha}_{\,U}^{\,(k,\,l)} \geq \big|\big| \big[\,
\boldsymbol{\beta}_{\,U}^{\,(k,\,l)} ; \delta_{\,U}^{\,(k,\,l)}
\big] \big|\big|, \;\; \widetilde{\alpha}_{D}^{\,(k,\,l)} \geq
\big|\big| \big[\, \boldsymbol{\beta}_{D}^{\,(k,\,l)} ;
\delta_{D}^{\,(k,\,l)} \big]
\big|\big|,\IEEEyesnumber\label{Eqn:ProbSndStagePrecoderStdProbSinrConstraintsUlDl2aLb}
\end{IEEEeqnarray*}
then the SINR constraints degenerate into convex functions in
$\textbf{F}_{\!R}^{}$. Yet, the SINR constraints are still
non-convex in $\gamma_{\,0}^{}$. Altogether, the precoder design
problem is quasi-convex and it can be solved using the bisection
method \cite[Section~4.2.5]{Bok:Convex_optimization:Boyd}.
Specifically, we define the SOCP feasibility problem of designing
$\textbf{F}_{\!R}^{}$ that achieves a \emph{target} value of
$\gamma_{\,0}^{}$ as\\\indent\;\;$\textbf{F}_{\!R}^{\,\star} :=
\widetilde{\mathcal{B}}_{R}^{}\big\{ \gamma_{\,0}^{},\!
\textbf{V}_{\!B}^{}, \{ \textbf{V}_{k}^{} \}_{k=1}^{K},
\textbf{W}_{\!B}^{\,\star}, \{\:\! \textbf{W}_{k}^{\,\star\,}
\}_{k=1}^{K}, \textbf{A}_{R}^{\!\star}, P_{R}^{}, \{
\boldsymbol{\omega}_{U}^{(k)}\!, \boldsymbol{\omega}_{\!D}^{(k)\,}
\}_{k=1}^{K} \big\}$\nopagebreak[4]
\begin{subnumcases}{\label{Eqn:ProbSndStagePrecoderFeasibilityProbAppendix}}
\textrm{find} & \!\!\!\!\!\!\!\!\!\!
$\textbf{F}_{\!R}^{}$\label{Eqn:ProbSndStagePrecoderFeasibilityProbObjectiveAppendix}\\
\,\textrm{s.t.} & \!\!\!\!\!\!\!\!\!\!\! $\big[\,
\widetilde{\alpha}_{\,U}^{\,(k,\,l)} \,;\,
\boldsymbol{\beta}_{\,U}^{\,(k,\,l)} \,; \delta_{\,U}^{\,(k,\,l)}
\,\big] \!\succeq_{K}^{}\! 0, \;\; \big[\,
\widetilde{\alpha}_{D}^{\,(k,\,l)} \,;\,
\boldsymbol{\beta}_{D}^{\,(k,\,l)} \,; \delta_{D}^{\,(k,\,l)}
\,\big] \!\succeq_{K}^{}\! 0, \;\;\forall k \in \mathcal{K}, \forall
l \in
\mathcal{L}_{\,k}^{},\;$\label{Eqn:ProbSndStagePrecoderFeasibilityProbSinrConstraintsAppendix}\\
& \!\!\!\!\!\!\!\!\!\!\! $\big[\, \sqrt{P_{R}^{}} \;;\,
\boldsymbol{\rho} \,\big] \!\succeq_{K}^{}\!
0,$\label{Eqn:ProbSndStagePrecoderFeasibilityProbPowerConstraintAppendix}
\end{subnumcases}
where
\eqref{Eqn:ProbSndStagePrecoderFeasibilityProbSinrConstraintsAppendix}
and
\eqref{Eqn:ProbSndStagePrecoderFeasibilityProbPowerConstraintAppendix}
correspond to
\eqref{Eqn:ProbSndStagePrecoderStdProbSinrConstraintsUlDl2aLb} and
\eqref{Eqn:ProbSndStagePrecoderStdProbPowerConstraintDerive1}
expressed as SOC constraints, respectively. Starting with an
interval that is expected to contain the \emph{optimum} value of
$\gamma_{\,0}^{}$, we repeatedly bisect the interval and select the
subinterval in which Problem~$\widetilde{\mathcal{B}}_{R}^{}$ is
feasible until $\gamma_{\,0}^{}$ converges.

\bibliographystyle{IEEEtran}
\bibliography{IEEEabrv,myBibFile}

\begin{thebibliography}{10}
\providecommand{\url}[1]{#1}
\csname url@samestyle\endcsname
\providecommand{\newblock}{\relax}
\providecommand{\bibinfo}[2]{#2}
\providecommand{\BIBentrySTDinterwordspacing}{\spaceskip=0pt\relax}
\providecommand{\BIBentryALTinterwordstretchfactor}{4}
\providecommand{\BIBentryALTinterwordspacing}{\spaceskip=\fontdimen2\font plus
\BIBentryALTinterwordstretchfactor\fontdimen3\font minus
  \fontdimen4\font\relax}
\providecommand{\BIBforeignlanguage}[2]{{%
\expandafter\ifx\csname l@#1\endcsname\relax
\typeout{** WARNING: IEEEtran.bst: No hyphenation pattern has been}%
\typeout{** loaded for the language `#1'. Using the pattern for}%
\typeout{** the default language instead.}%
\else
\language=\csname l@#1\endcsname
\fi
#2}}
\providecommand{\BIBdecl}{\relax}
\BIBdecl

\bibitem{Jnl:Cooperative_diversity:Wornell}
J.~N. Laneman, D.~N.~C. Tse, and G.~W. Wornell, ``Cooperative diversity in
  wireless networks: {Efficient} protocols and outage behavior,'' \emph{{IEEE}
  Trans. Inf. Theory}, vol.~50, pp. 3062--3080, Dec. 2004.

\bibitem{Jnl:Cooperative_diversity:Erkip1}
A.~Sendonaris, E.~Erkip, and B.~Aazhang, ``User cooperation diversity -- {Part
  I: System} description,'' \emph{{IEEE} Trans. Commun.}, vol.~51, pp.
  1927--1938, Nov. 2003.

\bibitem{Jnl:Cooperative_diversity:Erkip2}
------, ``User cooperation diversity -- {Part II: Implementation} aspects and
  performance analysis,'' \emph{{IEEE} Trans. Commun.}, vol.~51, pp.
  1939--1948, Nov. 2003.

\bibitem{Jnl:Relay_capacity_scaling_laws:Bolcskei}
H.~B$\ddot{\textrm{o}}$lcskei, R.~U. Nabar, O.~Oyman, and A.~J. Paulraj,
  ``Capacity scaling laws in {MIMO} relay networks,'' \emph{{IEEE} Trans.
  Wireless Commun.}, vol.~5, pp. 1433--1444, Jun. 2006.

\bibitem{Std:16m}
\emph{{Draft Amendment to IEEE Standard for Local and Metropolitan Area
  Networks, Part 16: Air Interface for Fixed and Mobile Broadband Wireless
  Access Systems}}, IEEE Std. P802.16m/D12, 2011.

\bibitem{Jnl:3GPP_LTE-A_relays}
S.~W. Peters, A.~Y. Panah, K.~T. Truong, and {R. W. Heath, Jr.}, ``Relay
  architectures for {3GPP} {LTE-Advanced},'' \emph{{EURASIP} J. Wireless Comm.
  and Networking}, 2009.

\bibitem{Std:Winner_final_innovation_report}
{Celtic Project CP5-026 WINNER+}, ``{Final Innovation Report},'' Apr. 2010.

\bibitem{Jnl:Multiuser_two_way_relaying_no_SI_cancellation:Sayed}
J.~Joung and A.~H. Sayed, ``Multiuser two-way amplify-and-forward relay
  processing and power control methods for beamforming systems,'' \emph{{IEEE}
  Trans. Signal Process.}, vol.~58, pp. 1833--1846, Mar. 2010.

\bibitem{Jnl:Two_way_relay_intro:Wittneben}
B.~Rankov and A.~Wittneben, ``Spectral efficient protocols for half-duplex
  fading relay channels,'' \emph{{IEEE} J. Sel. Areas Commun.}, vol.~25, pp.
  379--389, Feb. 2007.

\bibitem{Cnf:Pnc:Zhang}
S.~Zhang, S.~C. Liew, and P.~P. Lam, ``Hot topic: {Physical}-layer network
  coding,'' in \emph{Proc. {ACM} {MobiCom}'06}, 2006.

\bibitem{Jnl:Two_way_single_user_AF_single_stream}
R.~Zhang, Y.~C. Liang, C.~C. Chai, and S.~Cui, ``Optimal beamforming for
  two-way multi-antenna relay channel with analogue network coding,''
  \emph{{IEEE} J. Sel. Areas Commun.}, vol.~27, pp. 699--712, Jun. 2009.

\bibitem{Thesis:Two_way_relay:Unger}
T.~Unger, ``Multiple-antenna two-hop relaying for bi-directional transmission
  in wireless communication systems,'' Ph.D. dissertation, Technische
  Universit$\ddot{\textrm{a}}$t Darmstadt, Germany, 2009.

\bibitem{Jnl:Single_user_two_way_relaying_gradient:Korea}
K.-J. Lee, H.~Sung, E.~Park, and I.~Lee, ``Joint optimization for one and
  two-way {MIMO} {AF} multiple-relay systems,'' \emph{{IEEE} Trans. Wireless
  Commun.}, vol.~9, pp. 3671--3681, Dec. 2010.

\bibitem{Cnf:ANOMAX_rank_restored}
F.~Roemer and M.~Haardt, ``A low-complexity relay transmit strategy for two-way
  relaying with {MIMO} amplify and forward relays,'' in \emph{Proc. {IEEE}
  {ICASSP}'10}, 2010.

\bibitem{Jnl:Single_user_two_way_relaying_adaptive_mod:Alouini}
K.-S. Hwang, Y.-C. Ko, and M.-S. Alouini, ``Performance analysis of two-way
  amplify and forward relaying with adaptive modulation over multiple relay
  network,'' \emph{{IEEE} Trans. Commun.}, vol.~59, pp. 402--406, Feb. 2011.

\bibitem{Jnl:AF_two-way_relay_error_exponents}
H.~Q. Ngo, T.~Q.~S. Quek, and H.~Shin, ``Amplify-and-forward two-way relay
  networks: {Error} exponents and resource allocation,'' \emph{{IEEE} Trans.
  Commun.}, vol.~58, pp. 2653--2666, Sep. 2010.

\bibitem{Jnl:Multiuser_two_way_relaying_no_SI_cancellation_User_Sel:Sayed}
J.~Joung and A.~H. Sayed, ``User selection methods for multiuser two-way relay
  communications using space division multiple access,'' \emph{{IEEE} Trans.
  Wireless Commun.}, vol.~9, pp. 2130--2136, Jul. 2010.

\bibitem{Cnf:Two_way_one_BS_multi_MS_AF:Toh}
S.~Toh and D.~T.~M. Slock, ``A linear beamforming scheme for multi-user {MIMO}
  {AF} two-phase two-way relaying,'' in \emph{Proc. {IEEE} {PIMRC}'09}, 2009.

\bibitem{Jnl:Multiuser_two_way_relaying_scheme_and_analysis:Ding}
Z.~Ding, I.~Krikidis, J.~Thompson, and K.~K. Leung, ``Physical layer network
  coding and precoding for the two-way relay channel in cellular systems,''
  \emph{{IEEE} Trans. Signal Process.}, vol.~59, pp. 696--712, Feb. 2011.

\bibitem{Cnf:Two_way_one_BS_multi_MS_AF:Sun}
C.~Sun, Y.~Li, B.~Vucetic, and C.~Yang, ``Transceiver design for multi-user
  multi-antenna two-way relay channels,'' in \emph{Proc. {IEEE} {GLOBECOM}'10},
  2010.

\bibitem{Jnl:Pnc_channel_coding:Zhang}
S.~Zhang and S.-C. Liew, ``Channel coding and decoding in a relay system
  operated with physical-layer network coding,'' \emph{{IEEE} J. Sel. Areas
  Commun.}, vol.~27, pp. 788--796, Jun. 2009.

\bibitem{Jnl:Pnc_modulation:Tarokh}
T.~Koike-Akino, P.~Popovski, and V.~Tarokh, ``Optimized constellations for
  two-way wireless relaying with physical network coding,'' \emph{{IEEE} J.
  Sel. Areas Commun.}, vol.~27, pp. 773--787, Jun. 2009.

\bibitem{Jnl:Mimo_Y_channel:Lee}
N.~Lee, J.-B. Lim, and J.~Chun, ``Degrees of freedom of the {MIMO} {Y} channel:
  {Signal} space alignment for network coding,'' \emph{{IEEE} Trans. Inf.
  Theory}, vol.~56, pp. 3332--3342, Jul. 2010.

\bibitem{Jnl:IA:Cadambe_Jafar}
V.~R. Cadambe and S.~A. Jafar, ``Interference alignment and degrees of freedom
  of the {$K$}-user interference channel,'' \emph{{IEEE} Trans. Inf. Theory},
  vol.~54, pp. 3425--3441, Aug. 2008.

\bibitem{Jnl:Distributed_IA:Gomadam_Cadambe_Jafar}
K.~Gomadam, V.~R. Cadambe, and S.~A. Jafar, ``A distributed numerical approach
  to interference alignment and applications to wireless interference
  networks,'' \emph{{IEEE} Trans. Inf. Theory}, vol.~57, pp. 3309--3322, Jun.
  2011.

\bibitem{Jnl:Piaid:Huang}
H.~Huang and V.~K.~N. Lau, ``Partial interference alignment for {$K$}-user
  {MIMO} interference channels,'' \emph{{IEEE} Trans. Signal Process.},
  vol.~59, pp. 4900--4908, Oct. 2011.

\bibitem{Cnf:Two-way_relay_eigen-direction_alignment}
T.~Yang, X.~Yuan, P.~Li, I.~B. Collings, and J.~Yuan, ``A new eigen-direction
  alignment algorithm for physical-layer network coding in {MIMO} two-way relay
  channels,'' in \emph{Proc. {IEEE} {ISIT}'11}, 2011.

\bibitem{Jnl:Signal_alignment_CDM}
T.~Liu and C.~Yang, ``Signal alignment for multicarrier code division multiple
  user two-way relay systems,'' \emph{{IEEE} Trans. Wireless Commun.}, vol.~10,
  pp. 3700--3710, Nov. 2011.

\bibitem{Jnl:Analysis_max-min_weighted_SINR_DL}
D.~W.~H. Cai, T.~Q.~S. Quek, and C.~W. Tan, ``A unified analysis of max-min
  weighted {SINR} for {MIMO} downlink system,'' \emph{{IEEE} Trans. Signal
  Process.}, vol.~59, pp. 3850--3862, Aug. 2011.

\bibitem{Jnl:MBS-SDMA_using_SDR_with_perfect_CSI:Sidiropoulos_Luo}
E.~Karipidis, N.~D. Sidiropoulos, and Z.-Q. Luo, ``Quality of service and
  max-min fair transmit beamforming to multiple cochannel multicast groups,''
  \emph{{IEEE} Trans. Signal Process.}, vol.~56, pp. 1268--1279, Mar. 2008.

\bibitem{Bok:Convex_optimization:Boyd}
S.~Boyd and L.~Vandenberghe, \emph{Convex Optimization}.\hskip 1em plus 0.5em
  minus 0.4em\relax Cambridge, UK: Cambridge University Press, 2004.

\bibitem{Jnl:Interference_channel_DoF:Jafar}
S.~A. Jafar and M.~J. Fakhereddin, ``Degrees of freedom for the {MIMO}
  interference channel,'' \emph{{IEEE} Trans. Inf. Theory}, vol.~53, pp.
  2637--2642, Jul. 2007.

\bibitem{Jnl:AF_relay_scaling:Gallager}
S.~Borade, L.~Zheng, and R.~Gallager, ``Amplify-and-forward in wireless relay
  networks: Rate, diversity, and network size,'' \emph{{IEEE} Trans. Inf.
  Theory}, vol.~53, pp. 3302--3318, Oct. 2007.

\bibitem{Jnl:Coordinated_beamforming:Chae}
\BIBentryALTinterwordspacing
C.-B. Chae, I.~Hwang, {R. W. Heath, Jr.}, and V.~Tarokh, ``Interference
  aware-coordinated beamforming system in a two-cell environment,'' submitted
  to \textit{{IEEE} J. Sel. Areas Commun.}, Special Issue on Cooperative
  Communications in MIMO Cellular Networks, Sep. 2009. [Online]. Available:
  \url{http://nrs.harvard.edu/urn-3:HUL.InstRepos:3293263}
\BIBentrySTDinterwordspacing

\bibitem{Jnl:Unified_framework_convex:Palomar}
D.~P. Palomar, J.~M. Cioffi, and M.~A. Lagunas, ``Joint {Tx}-{Rx} beamforming
  design for multicarrier {MIMO} channels: {A} unified framework for convex
  optimization,'' \emph{{IEEE} Trans. Signal Process.}, vol.~51, pp.
  2381--2401, Sep. 2003.

\bibitem{Jnl:Linear_precoding_conic:Eldar}
A.~Wiesel, Y.~C. Eldar, and S.~Shamai, ``Linear precoding via conic
  optimization for fixed {MIMO} receivers,'' \emph{{IEEE} Trans. Signal
  Process.}, vol.~54, pp. 161--176, Jan. 2006.

\bibitem{Jnl:Transceiver_QoS_per_antenna_power:Tolli}
A.~T$\ddot{\textrm{o}}$lli, M.~Codreanu, and M.~Juntti, ``Linear multiuser
  {MIMO} transceiver design with quality of service and per-antenna power
  constraints,'' \emph{{IEEE} Trans. Signal Process.}, vol.~56, pp. 3049--3055,
  Jul. 2008.

\bibitem{Cnf:Analog_feedback_Moto1:Thomas}
T.~A. Thomas, K.~L. Baum, and P.~Sartori, ``Obtaining channel knowledge for
  closed-loop multi-stream broadband {MIMO-OFDM} communications using direct
  channel feedback,'' in \emph{Proc. {IEEE} {GLOBECOM}'05}, Dec. 2005, pp.
  3907--3911.

\end{thebibliography}

\begin{algorithm}
\caption{Top-Level Algorithm}\label{Algorithm:TopLevel}
\begin{algorithmic}[0]
\State \textbf{Outputs}: $\textbf{W}_{\!B}^{\,\star}, \{
\textbf{W}_{k}^{\,\star\,} \}_{k=1}^{K},\!
\textbf{W}_{\!R}^{\,\star}, \textbf{V}_{\!B}^{\,\star}, \{
\textbf{V}_{k}^{\,\star\,} \}_{k=1}^{K}$\;\;\; \textbf{Inputs}:
$P_{B}^{}, \{ P_{\,k}^{} \}_{k=1}^{K}, P_{R}^{}, \{
\boldsymbol{\omega}_{U}^{(k)}\!, \boldsymbol{\omega}_{\!D}^{(k)\,}
\}_{k=1}^{K}$ \State \emph{First Stage Processing} \State ~~
\textbf{Step 1}: Solve $\big\{ \textbf{W}_{\!B}^{\,\star}, \{\:\!
\textbf{W}_{k}^{\,\star\,} \}_{k=1}^{K}, \textbf{A}_{R}^{\!\star}
\big\} := \mathcal{M}\big\{ P_{B}^{}, \{ P_{\,k}^{} \}_{k=1}^{K}, \{
\boldsymbol{\omega}_{U}^{(k)}\!, \boldsymbol{\omega}_{\!D}^{(k)\,}
\}_{k=1}^{K} \big\}$ using Lemma~\ref{Lemma:BeamSelPowerAllocation}.
\State \emph{Second Stage Processing} \State ~~ \textbf{Step 2.\,0}:
Initialize $\textbf{V}_{\!B}^{} = ( \textbf{W}_{\!B}^{\,\star\,}
)_{}^{T}$ and $\textbf{V}_{k}^{} = ( \textbf{W}_{k}^{\,\star\,}
)_{}^{T}$, $\forall k \in \mathcal{K}$. \State ~~ \textbf{Repeat}
\State ~~~~~ \textbf{Step 2.\,1}: Solve $\big\{
\textbf{F}_{\!R}^{\,\star}, \gamma_{0}^{} \big\} \!:=
\mathcal{B}_{R}^{}\big\{ \textbf{V}_{\!B}^{}, \{ \textbf{V}_{k}^{}
\}_{k=1}^{K}, \textbf{W}_{\!B}^{\,\star}, \{\:\!
\textbf{W}_{k}^{\,\star\,} \}_{k=1}^{K}, \textbf{A}_{R}^{\!\star},
P_{R}^{}, \{ \boldsymbol{\omega}_{U}^{(k)}\!,
\boldsymbol{\omega}_{\!D}^{(k)\,} \}_{k=1}^{K} \big\}$ \State
~~~~~~~~~~~~~~~~~~using Algorithm~\ref{Algorithm:RsPrecoder}. \State
~~~~~ \textbf{Step 2.\,2}: Solve $\textbf{V}_{\!B}^{\,\star} :=
\mathcal{B}_{B}^{}\big\{ \textbf{F}_{\!R}^{},
\textbf{W}_{\!B}^{\,\star}, \{\:\! \textbf{W}_{k}^{\,\star\,}
\}_{k=1}^{K}, \textbf{A}_{R}^{\!\star}, \{
\boldsymbol{\omega}_{U}^{(k)} \}_{k=1}^{K} \big\}$ using
\eqref{Eqn:BsEqualizerSoln}. \State ~~~~~ \textbf{Step 2.\,3}: Solve
$\textbf{V}_{k}^{\,\star\,} := \mathcal{B}_{k}^{}\big\{
\textbf{F}_{\!R}^{}, \textbf{W}_{\!B}^{\,\star}, \{\:\!
\textbf{W}_{k}^{\,\star\,} \}_{k=1}^{K}, \textbf{A}_{R}^{\!\star},
\boldsymbol{\omega}_{\!D}^{(k)} \big\}$, $\forall k \in
\mathcal{K}$, using \eqref{Eqn:MsEqualizerSoln}. \State ~~
\textbf{Until} the minimum weighted per stream SINR converges.
\State \textbf{Step 3}: Set $\textbf{W}_{\!R}^{\,\star} =
\textbf{F}_{\!R\,}^{\,\star} \textbf{A}_{R\,}^{\!\star}$.
\end{algorithmic}
\end{algorithm}
\vspace{-2em}
\begin{algorithm}
\caption{RS Precoder Matrix
Optimization}\label{Algorithm:RsPrecoder}
\begin{algorithmic}[0]
\State \textbf{Outputs}: $\textbf{F}_{\!R}^{\,\star}$\;\;\;
\textbf{Inputs}: $\textbf{V}_{\!B}^{}, \{ \textbf{V}_{k}^{}
\}_{k=1}^{K}, \textbf{W}_{\!B}^{\,\star}, \{\:\!
\textbf{W}_{k}^{\,\star\,} \}_{k=1}^{K}, \textbf{A}_{R}^{\!\star},
P_{R}^{}, \{ \boldsymbol{\omega}_{U}^{(k)}\!,
\boldsymbol{\omega}_{\!D}^{(k)\,} \}_{k=1}^{K}$ \State \textbf{Step
0}: Initialize $\displaystyle \gamma_{\textrm{min}}^{} =\!
\min_{\substack{\forall\, k \,\in\, \mathcal{K}, \forall\, l \,\in\,
\mathcal{L}_{k}^{}}} \!\!\textstyle\big\{ \big( [\,
\boldsymbol{\omega}_{U}^{(k)} \,]_{(l)}^{} \big)_{}^{-1}
\gamma_{\,U}^{(k,\,l)}, \big( [\, \boldsymbol{\omega}_{\!D}^{(k)}
\,]_{(l)}^{} \big)_{}^{-1} \gamma_{D}^{(k,\,l)} \big\}$ \State
~~~~~~~~\, and $\displaystyle \gamma_{\textrm{max}}^{} =\!
\max_{\substack{\forall\, k \,\in\, \mathcal{K}, \forall\, l \,\in\,
\mathcal{L}_{k}^{}}} \!\!\textstyle\big\{ \big( [\,
\boldsymbol{\omega}_{U}^{(k)} \,]_{(l)}^{} \big)_{}^{-1}
\gamma_{\,U}^{(k,\,l)}, \big( [\, \boldsymbol{\omega}_{\!D}^{(k)}
\,]_{(l)}^{} \big)_{}^{-1} \gamma_{D}^{(k,\,l)} \big\}$. \State
\textbf{Repeat} \Comment{Bisection Method} \State ~~ \textbf{Step
1}: Set the target minimum weighted per stream SINR $\gamma_{\,0}^{}
= ( \gamma_{\textrm{min}}^{} \!+ \gamma_{\textrm{max}\,}^{} ) / 2$.
\State ~~~~~~~~~~~\, Solve $\textbf{F}_{\!R}^{\,\star} :=
\widetilde{\mathcal{B}}_{R}^{}\big\{ \gamma_{\,0}^{},\!
\textbf{V}_{\!B}^{}, \{ \textbf{V}_{k}^{} \}_{k=1}^{K},
\textbf{W}_{\!B}^{\,\star}, \{\:\! \textbf{W}_{k}^{\,\star\,}
\}_{k=1}^{K}, \textbf{A}_{R}^{\!\star}, P_{R}^{}, \{
\boldsymbol{\omega}_{U}^{(k)}\!, \boldsymbol{\omega}_{\!D}^{(k)\,}
\}_{k=1}^{K} \big\}$. \State ~~ \textbf{Step 2}: If
Problem~$\widetilde{\mathcal{B}}_{R}^{}$ is feasible, set
$\gamma_{\textrm{min}}^{} = \gamma_{\,0}^{}$; else set
$\gamma_{\textrm{max}\,}^{} = \gamma_{\,0}^{}$. \State
\textbf{Until} $\gamma_{\,0}^{}$ converges.
\end{algorithmic}
\end{algorithm}

\newpage

\begin{figure}[!ht]
\centering \subfloat[][]{\includegraphics[width=1.5in]{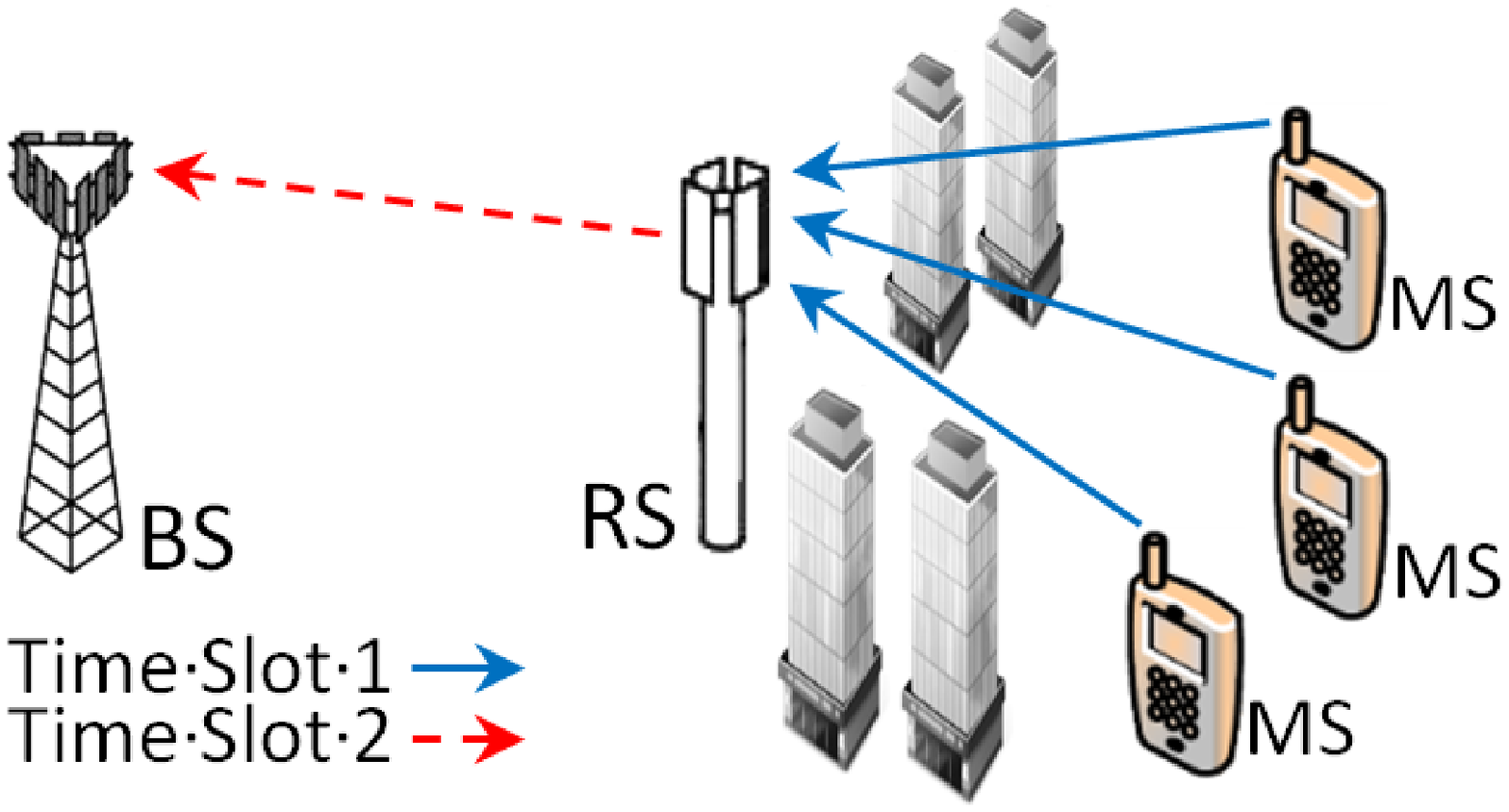}
\label{Fig:OneWayUl}}~~~~~
\subfloat[][]{\includegraphics[width=1.5in]{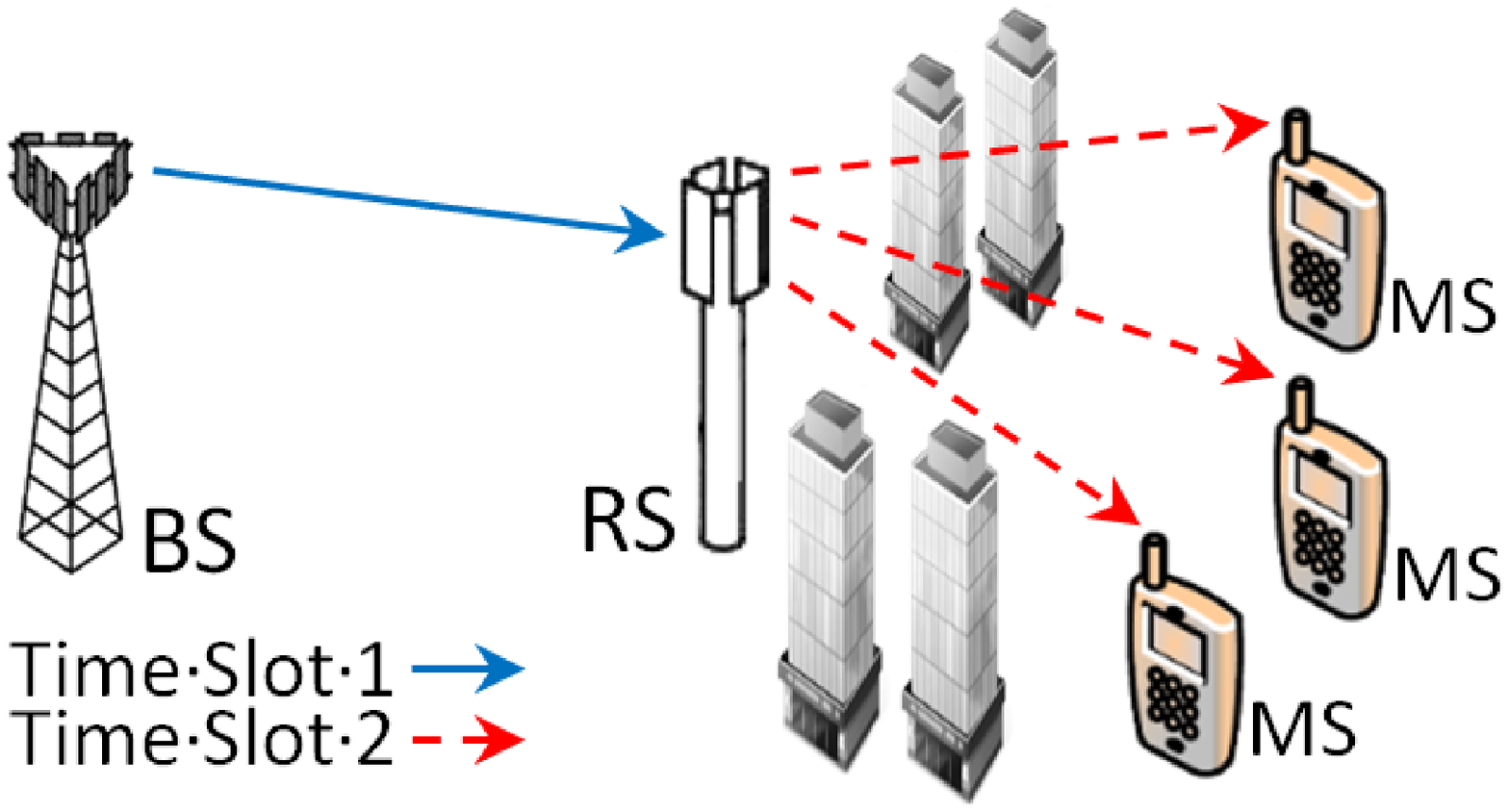}
\label{Fig:OneWayDl}}\\
\subfloat[][]{\includegraphics[width=1.5in]{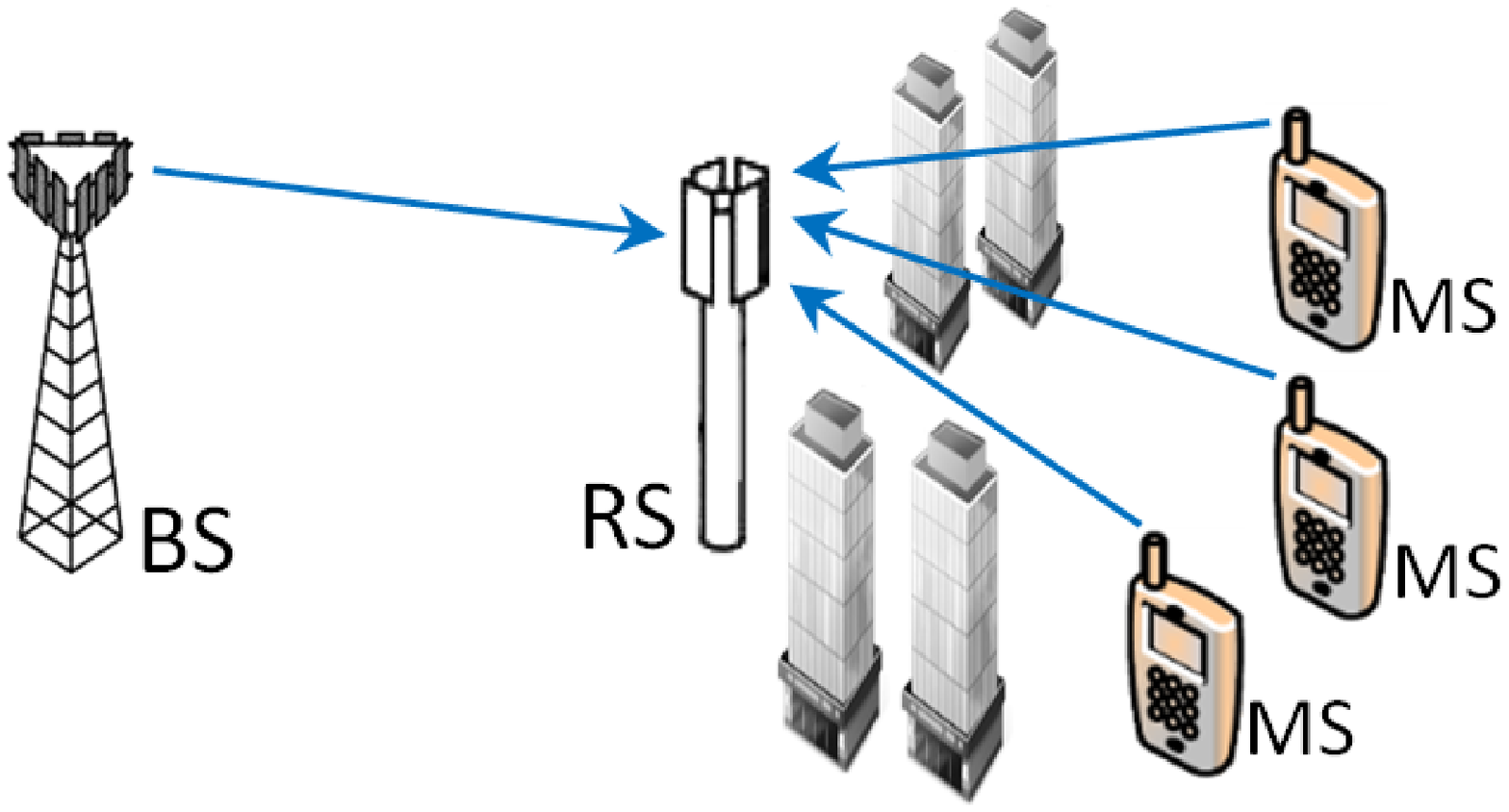}
\label{Fig:TwoWayMac}}~~~~~
\subfloat[][]{\includegraphics[width=1.5in]{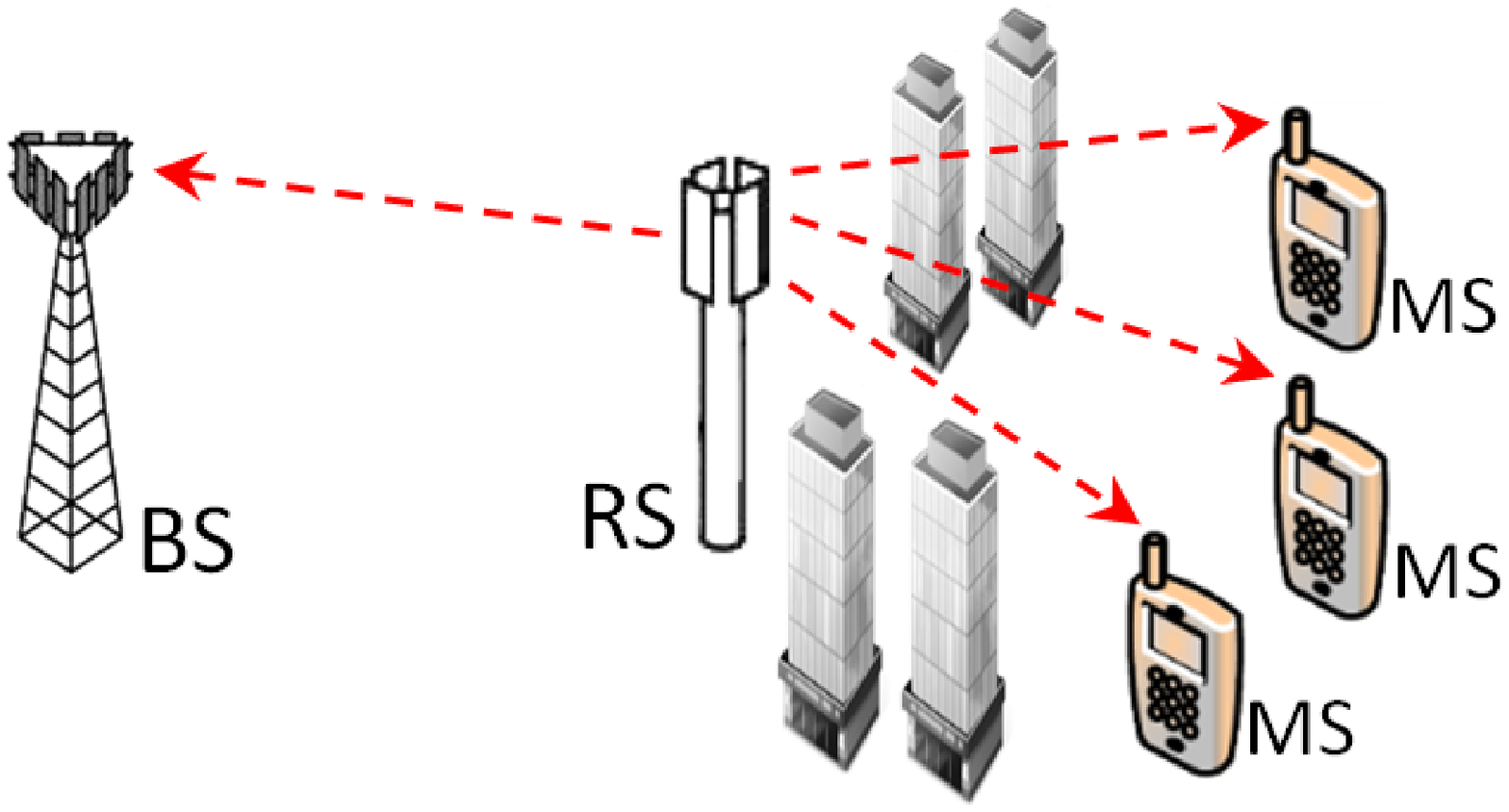}
\label{Fig:TwoWayBc}} \caption{Multi-user relay system. (a) One-way
transmission UL. (b) One-way transmission DL. (c) Two-way
transmission MAC phase. (d) Two-way transmission BC phase.}
\label{Fig:MultiUserRelaySystem}
\end{figure}

\begin{figure}[!ht]
\centering
\includegraphics[width = 3.5in]{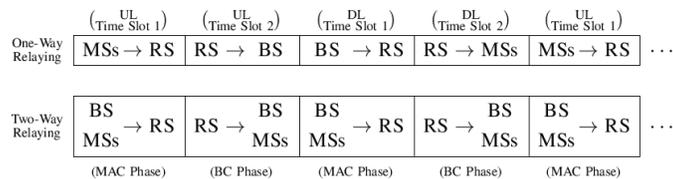}
\caption{Timing diagram for multi-user one- and two-way relaying in
TDD systems.}\label{Fig:TimeLine}
\end{figure}

\begin{figure}[!ht]
\centering
\includegraphics[width = 2.6in]{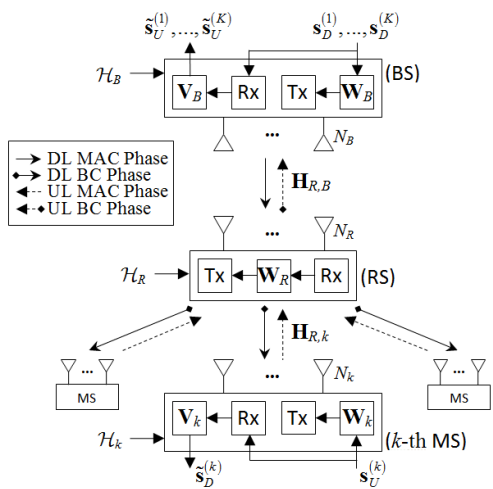}
\caption{System model of two-way relaying between one BS and $K$
MSs. The BS is equipped with $N_{B}^{}$ antennas, the RS is equipped
with $N_{R}^{}$ antennas, and the \hbox{$k$-th} MS is equipped with
$N_{k}^{}$ antennas. The BS and the \hbox{$k$-th} MS exchange
$L_{k}^{}$ data streams.}\label{Fig:SystemModel}
\end{figure}

\begin{figure}[!ht]
\centering
\includegraphics[width=4.5in]{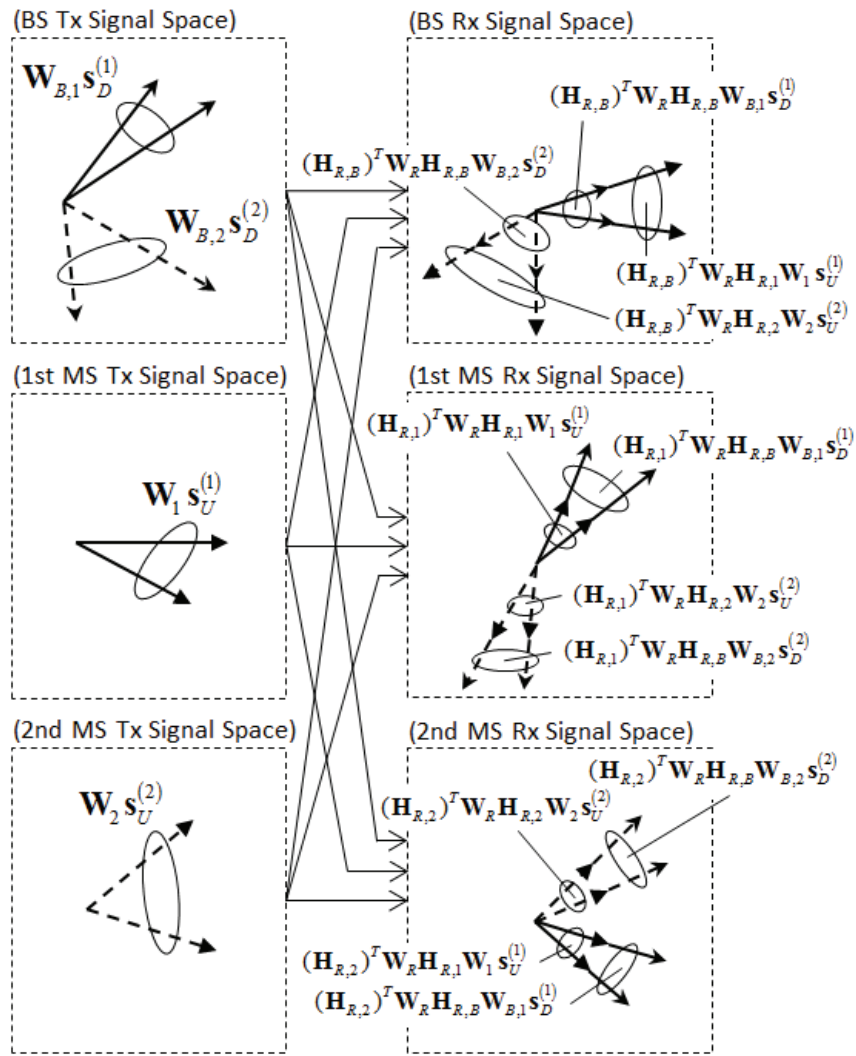}
\caption{Illustration of signal space alignment. The BS exchanges
two data streams with each of two MSs. At each node, desired signals
are aligned with the backward propagated
self-interference.}\label{Fig:SignalSpace}
\end{figure}

\begin{figure}[!ht]
\centering
\includegraphics[width = 4.5in]{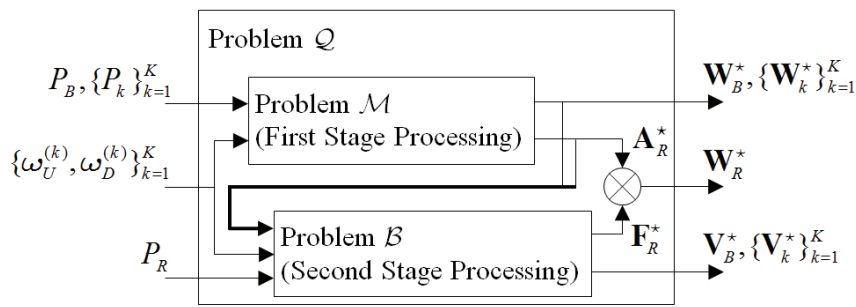}
\caption{Illustration of the two-stage transceiver design
algorithm.}\label{Fig:Problems}
\end{figure}

\begin{figure}[!ht]
\centering
\includegraphics[width = 4in]{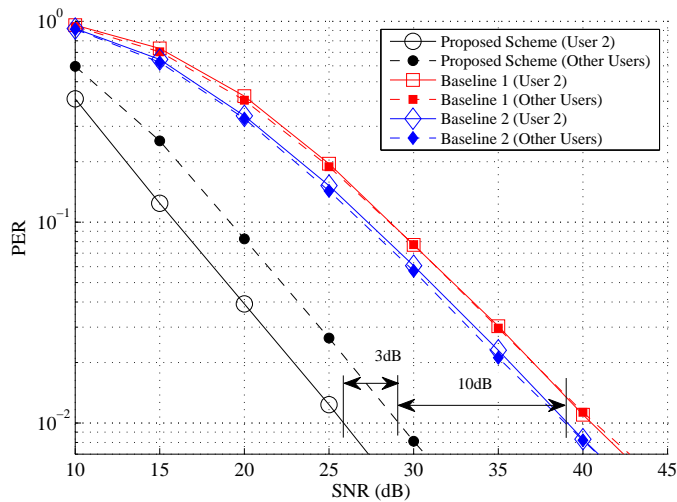}
\caption{Packet error rate versus SNR. The BS, the RS, and the MSs
are equipped with $N_{B}^{} = 4$, $N_{R}^{} = 4$, and $N_{k}^{} = 2$
antennas, respectively. The BS exchanges $L_{\,1}^{} = 2$,
$L_{\,2}^{} = 1$, and $L_{\,3}^{} = 1$ data streams with the MSs.
\hbox{User 2} has higher service priority than the other users: $[\,
\boldsymbol{\omega}_{U}^{(2)} \,]_{(1)}^{} = [\,
\boldsymbol{\omega}_{\!D}^{(2)} \,]_{(1)}^{} = 2$ and $[\,
\boldsymbol{\omega}_{U}^{(k)} \,]_{(l)}^{} = [\,
\boldsymbol{\omega}_{\!D}^{(k)} \,]_{(l)}^{} = 1$
otherwise.}\label{Fig:Per}
\end{figure}

\begin{figure}[!ht]
\centering
\includegraphics[width = 4in]{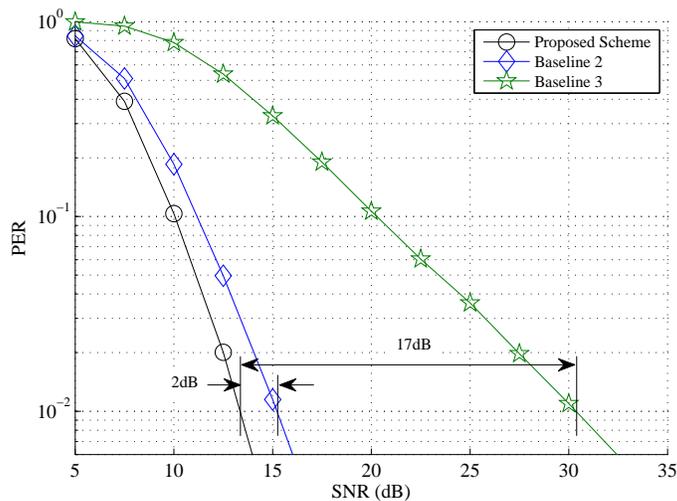}
\caption{Packet error rate versus SNR. The BS, the RS, and the MSs
are equipped with $N_{B}^{} = 4$, $N_{R}^{} = 8$, and $N_{k}^{} = 2$
antennas, respectively. The BS exchanges $L_{\,1}^{} = 2$,
$L_{\,2}^{} = 1$, and $L_{\,3}^{} = 1$ data streams with the MSs.
All users have the same service priority: $[\,
\boldsymbol{\omega}_{U}^{(k)} \,]_{(l)}^{} = [\,
\boldsymbol{\omega}_{\!D}^{(k)} \,]_{(l)}^{} = 1$, $\forall k,
l$.}\label{Fig:PerSdma}
\end{figure}

\begin{figure}[!ht]
\centering
\subfloat[][]{\includegraphics[width=4in]{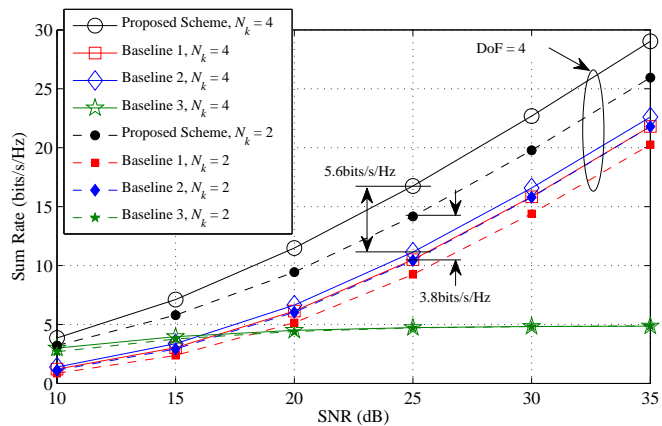}
\label{Fig:SumRate}}\\
\subfloat[][]{\includegraphics[width=4in]{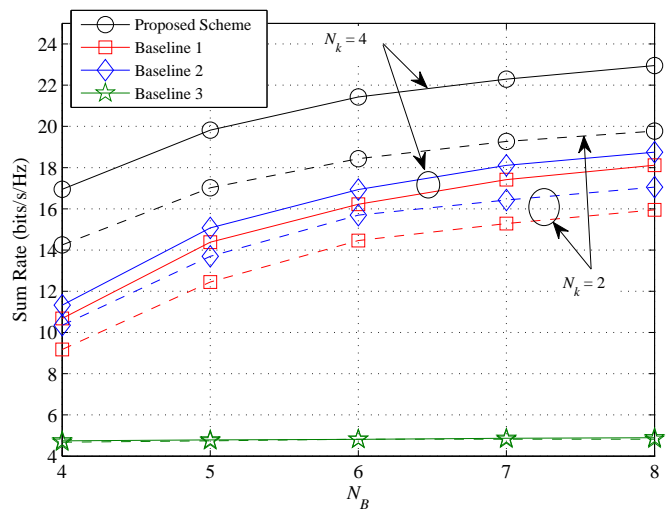}\label{Fig:SumRateVsNb}}
\caption{Average sum rate comparison. The RS and the MSs are
equipped with $N_{R}^{} = 4$ and $N_{\,k}^{} = \{ 2, 4 \}$ antennas,
respectively. The BS exchanges $L_{\,1}^{} = 2$, $L_{\,2}^{} = 1$,
and $L_{\,3}^{} = 1$ data streams with the MSs. \hbox{User 2} has
higher service priority than the other users: $[\,
\boldsymbol{\omega}_{U}^{(2)} \,]_{(1)}^{} = [\,
\boldsymbol{\omega}_{\!D}^{(2)} \,]_{(1)}^{} = 2$ and $[\,
\boldsymbol{\omega}_{U}^{(k)} \,]_{(l)}^{} = [\,
\boldsymbol{\omega}_{\!D}^{(k)} \,]_{(l)}^{} = 1$ otherwise. (a)
Average sum rate versus SNR when the BS is equipped with $N_{B}^{} =
4$ antennas. (b) Average sum rate versus the number of BS antennas
$N_{B}^{}$ at 25~dB SNR.} \label{Fig:SumRateAll}
\end{figure}

\end{document}